\documentclass[onecolumn,aps,prx,longbibliography,superscriptaddress]{revtex4-2}

\usepackage{amsmath}
\usepackage{pifont}
\usepackage{amssymb}
\usepackage{amsthm}
\usepackage{graphicx}
\usepackage{xcolor}
\usepackage[english]{babel}
\usepackage{csquotes}
\usepackage{braket}
\usepackage{hyperref}
\usepackage{mathtools}
\usepackage{enumitem}
\usepackage{IEEEtrantools}
\usepackage{relsize}
\usepackage{placeins}
\usepackage{comment}
\usepackage[normalem]{ulem}
\usepackage{systeme}
\usepackage{physics}
\usepackage{algpseudocode,algorithm}
\newcommand{\cmark}{\ding{51}}%
\newcommand{\xmark}{\ding{55}}%
\usepackage{times}

\usepackage{titlesec}

\usepackage{algorithm}% http://ctan.org/pkg/algorithms
\usepackage{algpseudocode}% http://ctan.org/pkg/algorithmicx

\usepackage{tikz-cd} 
\usepackage{adjustbox}

\usepackage{tikz}						%geometric/algebraic description.
\usetikzlibrary{math, arrows,shapes.misc,snakes,
		       automata,backgrounds,
		       petri,topaths, decorations.pathmorphing, tikzmark}	

\usepackage{pgffor}

\usepackage{pgfplots}
\usetikzlibrary{pgfplots.groupplots}
\pgfplotsset{compat=1.11}
\usepgfplotslibrary{fillbetween}

\usepackage{tikz}
\usetikzlibrary{
  shapes,
  shapes.geometric,
	trees,
	matrix,
  positioning,
    pgfplots.groupplots,
  }

\PassOptionsToPackage{pdftex,hyperfootnotes=false,pdfpagelabels}{hyperref}
\hypersetup{
	colorlinks = true,
	linkcolor = [rgb]{0.70,0.13,0.13},
	citecolor = [rgb]{0.13,0.55,0.13},
	urlcolor = [rgb]{0.25, 0.41, 0.88}}
\usepackage[capitalize]{cleveref}

\newcommand{\Id}{\ensuremath{\mathbb{I}}}

\providecommand{\calP}{\ensuremath{\mathcal{P}}}

\providecommand{\calV}{\ensuremath{\mathcal{V}}}

\providecommand{\bbE}{\ensuremath{\mathbb{E}}}

\providecommand{\N}{\ensuremath{\mathcal{N}}}

\providecommand{\calP}{\ensuremath{\mathcal{P}}}

\providecommand{\calV}{\ensuremath{\mathcal{V}}}

\providecommand{\bbE}{\ensuremath{\mathbb{E}}}

\def\01{\{0,1\}}
\newcommand{\eps}{\varepsilon}

\newcommand{\C}{\ensuremath{\mathcal{C}}}

\makeatletter
\let\oldabs\abs
\def\abs{\@ifstar{\oldabs}{\oldabs*}}
\let\oldnorm\norm
\def\norm{\@ifstar{\oldnorm}{\oldnorm*}}
\makeatother

\newtheorem{theorem}{Theorem}
\newtheorem{claim}[theorem]{Claim}
\newtheorem{lemma}[theorem]{Lemma}
\newtheorem{proposition}[theorem]{Proposition}
\newtheorem{definition}[theorem]{Definition}
\newtheorem{remark}[theorem]{Remark}
\newtheorem{corollary}[theorem]{Corollary}

\newtheorem{observation}[theorem]{Observation}
\newtheorem{example}[theorem]{Example}

\definecolor{antonio}{rgb}{.2,.5,.1}

\definecolor{armando}{rgb}{.5,.1,-.7}

\allowdisplaybreaks

\definecolor{jens}{rgb}{.5,.1,-.7}

\allowdisplaybreaks

\definecolor{yihui}{rgb}{.8,0,.5}

\allowdisplaybreaks

\definecolor{daniel}{rgb}{.8,.5,.3}

\definecolor{Soumik}{rgb}{.9,.5,.1}

\newcommand{\AntM}[1]{\textcolor{black}{#1}}

\newcommand{\Flip}{\mathbb{F}}

\newcommand{\Idd}{\mathbb{I}}

\newcommand{\CNOT}{\ensuremath{\operatorname{CNOT}}}
\newcommand{\Hadamard}{\ensuremath{\operatorname{H}}}
\newcommand{\PhaseS}{\ensuremath{\operatorname{S}}}

\newcommand{\supp}{\ensuremath{\operatorname{supp}}}

\newcommand{\ExU}{\underset{U\sim\mu_H}{\mathbb{E}}}

\newcommand{\MatC}[1]{\mathcal{L}\left(\mathbb{C}^{#1}\right)}

\newcommand{\Var}{\mathrm{Var}}
\newcommand{\Ex}{\mathbb{E}}

\newcommand{\hs}[2]{\langle #1, #2\rangle_{HS}}
\renewcommand{\Id}{I}
\newcommand{\Ug}{\mathrm{U}}

\newcommand{\fu}{Dahlem Center for Complex Quantum Systems, Freie Universit\"{a}t Berlin, 14195 Berlin, Germany}

\begin{document}

\title{Noise-induced shallow circuits and the absence of barren plateaus}

\date{\today}%May 12, 2022}

\author{Antonio Anna Mele}
\email{a.mele@fu-berlin.de}
\affiliation{\fu}

\author{Armando Angrisani}
\email{armando.angrisani@epfl.ch}
\affiliation{LIP6, CNRS, Sorbonne Université, 75005 Paris, France}

\affiliation{Institute of Physics, Ecole Polytechnique Fédérale de Lausanne (EPFL), CH-1015 Lausanne, Switzerland}

\author{Soumik Ghosh}
\affiliation{Department of Computer Science, University of Chicago, Chicago, Illinois 60637, USA}

\author{\\Sumeet Khatri}
\affiliation{\fu}

\author{Jens Eisert}
\affiliation{\fu}
\affiliation{Fraunhofer Heinrich Hertz Institute, 10587 Berlin, Germany}

\author{Daniel Stilck França}
\email{dsfranca@math.ku.dk}
\affiliation{Department of Mathematical Sciences, University of Copenhagen, 2100 Copenhagen, Denmark}
\affiliation{Univ Lyon, ENS Lyon, UCBL, CNRS, Inria, LIP, F-69342, Lyon Cedex 07, France}

\author{Yihui Quek}
\email{yihuiquek3.14@gmail.com}
\affiliation{Departments of Mathematics and Physics, Massachusetts Institute of Technology, 182 Memorial Drive, Cambridge, MA 02138, USA}
\begin{abstract}
\textcolor{black}{Without a successful implementation of fault-tolerant quantum error correction, calculations on quantum computers are subject to noise that limits their capabilities. Motivated by realistic near-term hardware considerations, we study the impact of uncorrected local noise on logical quantum circuits. We first show that in the task of estimating observable expectation values any noise truncates most quantum circuits to effectively logarithmic depth. We then prove that quantum circuits under any non-unital noise do not exhibit barren plateaus for cost functions composed of local observables. However, by using the effective shallowness, we also design an efficient classical algorithm to estimate observable expectation values within any constant additive accuracy, with high probability over the choice of the circuit, in any circuit architecture. Taken together, our results establish that, unless we carefully engineer quantum circuits to take advantage of the noise, noisy quantum circuits are unlikely to offer an advantage over shallow ones for algorithms that output observable expectation value estimates, such as many variational quantum machine learning proposals.
}
\end{abstract}

\begin{comment}
    \begin{abstract}
Motivated by realistic hardware considerations of the pre-fault-tolerant era, we comprehensively study the impact of uncorrected noise on quantum circuits. We first show that any noise `truncates' most quantum circuits to effectively logarithmic depth, in the task of estimating observable expectation values. We then prove that quantum circuits under any non-unital noise exhibit lack of barren plateaus for cost functions composed of local observables. But, by leveraging the effective shallowness, we also design an efficient classical algorithm to estimate observable expectation values within any constant additive accuracy, with high probability over the choice of the circuit, in any circuit architecture. %The runtime of the algorithm is independent of circuit depth, and for any inverse-polynomial target accuracy, it operates in polynomial time in the number of qubits for one-dimensional architectures and quasi-polynomial time for higher-dimensional ones. 
Taken together, our results showcase that, unless we carefully engineer the circuits to take advantage of the noise, it is unlikely that noisy quantum circuits are preferable over shallow quantum circuits for algorithms that output observable expectation value estimates, like many variational quantum machine learning proposals. Moreover, we anticipate that our work could provide valuable insights into the fundamental open question about the complexity of sampling from (possibly non-unital) noisy random circuits.
\end{abstract}

\end{comment}

\maketitle

%\newpage 
%
%\renewcommand{\theequation}{S\arabic{equation}}
%\setcounter{equation}{0}
%\setcounter{thm}{0}
%\setcounter{figure}{1}
%\setcounter{table}{0}
%\setcounter{section}{0}
%\setcounter{page}{1}
\makeatletter

\setcounter{secnumdepth}{3}
%\newpage

\vspace{-0.5cm}

\section{Introduction}

\noindent 
One of the most important questions for quantum computers of today is to understand the behavior and impact of noise~\cite{Preskill_2018}. It is crucial to understand whether noisy quantum computers provide any advantage, both for practically relevant problems~\cite{cai2023shors,Kim2023}, or even as proof-of-principle~\cite{GoogleSupremacy}, or whether we ultimately need error-corrected logical qubits to achieve this goal~\cite{Bluvstein_2023,MindTheGaps}. 
Past years have seen a tussle between demonstrations of quantum advantage~\cite{GoogleSupremacy,ZHU2022240,MadsenPH,SupremacyReview,Kim2023,MindTheGaps} and subsequent efficient classical simulation~\cite{pan2021simulating,GbosonCLASS,clifford2020faster,gao2018efficient,PhysRevLett.128.030501,Kechedzhi_2024,begušić2023fast,liao2023simulation,rudolph2023classical,Begu_i__2024,patra2023efficient}.
The issue of noise shows up in a multi-faceted way in various areas of near-term quantum computation.
In quantum machine learning, certain types of noises cause `barren plateaus' in optimization landscapes -- that is, the optimization landscape becomes flat, and any quantum signal is destroyed~\cite{nibp,Variational}. In random quantum circuit sampling~\cite{SupremacyReview}, a popular framework for demonstrations of quantum advantage, certain type of noise makes it possible to simulate the systems efficiently classically~\cite{Aharonov_2023}.  
However, the vast majority of prior work assumes the noise to be \textit{local, unital,} and \textit{primitive} (e.g., depolarizing)---where a quantum channel is deemed unital if it maps the maximally mixed state onto itself, and primitive if any state converges to the maximally mixed state if the channel is applied often enough. However, for a 
number of current physical platforms, it is by far more natural 
and realistic to consider the noise as \textit{non-unital}~\cite{Kandala,GoogleSupremacy,Chirolli_2008,Pino}, which can decrease the entropy of the system -- to the extent that depolarizing noise 
can be a misleading model. Previous results have studied the qualitatively different behavior of this noise in certain contexts; for instance, fault tolerance \cite{refrigerator} and random circuit sampling \cite{fefferman2023effect}.

%\subsection{Overview}
In this work, we make significant strides in presenting a comprehensive understanding of the impact of possibly non-unital noise on typical quantum circuits. Our assumptions about the noise are minimal, in particular we consider it to be \emph{local} and \emph{incoherent}, i.e., the noise present in the device has a tensor product structure and it is not given by a unitary channel. Our main results can be succinctly stated as follows:

\begin{itemize}
\item \textbf{Effective depth:} We show that arbitrary deep random quantum circuits, under \emph{any} uncorrected, possibly non-unital noise, effectively get `truncated' in the following sense: the influence of gates on observable expectation values decreases exponentially in their distance from the last layer, i.e., only the last layers can contribute significantly to the expectation value. \textcolor{black}{This also implies that, for typical (random) circuits, the complexity under non-unital and unital noise becomes essentially the same, within the well-defined context of estimating expectation values. Indeed, it is well known that under certain types of unital noise, like depolarizing, meaningful quantum computation must be confined to \emph{logarithmic} depth~\cite{aharonov1996limitations}, and what we prove is that the same limitation arises for \emph{typical} circuits with \emph{arbitrary} local noise—even non-unital: the initial portion of the circuit beyond the last $\log n$ layers can be discarded without affecting the predicted expectation values.}

\item \textbf{Lack of barren plateaus:} 
%Because the effective circuit is always `shallow', we get a provable lack of barren plateaus for cost functions made out of local observables---that is, the cost landscape is never flat, and the gradient of the cost function never vanishes---at any depth, under non-unital noise. 
Under non-unital noise, we get a provable lack of barren plateaus for cost functions made out of local observables---that is, the cost landscape is never flat, and the gradient of the cost function never vanishes---at any depth. 
This also implies that local expectation values of arbitrary deep random quantum circuits with non-unital noise are \emph{not too concentrated} towards a fixed value, in stark contrast to the unital noise scenario~\cite{nibp}. This phenomenon, however, is not good news for variational quantum algorithms~\cite{Variational}, as we show that such circuits behave like shallow circuits, which have limited computational power. 

\item \textbf{Classical simulation:} \textcolor{black}{Furthermore, by exploiting the effective shallowness of the circuits, we show how to \emph{efficiently} classically simulate, on average over the class of noisy quantum circuits,
expectation values of any local observable, up to constant additive precision, at $\emph{any}$ depth and in any circuit architecture.}
\end{itemize}

\AntM{In this work we focus on the problem of estimating expectation values of observables, rather than on full output distributions or sampling tasks. 
This focus is motivated by two considerations. 
First, in condensed-matter physics and quantum simulation, the physically relevant quantities that diagnose phases of matter, reveal order parameters, or determine response functions are expectation values of local or few-body observables. 
Second, in variational quantum algorithms and quantum machine learning, the central task is to estimate cost functions made by observables expectation values on noisy near-term devices. 
For these reasons, expectation values are the natural figure of merit for the questions we target, and form the basis of our analysis.}

%Furthermore, by exploiting the effective shallowness of the circuits, we show how to classically simulate, on average over the class of noisy quantum circuits, expectation values of any local observable, up to $\varepsilon$ additive precision, at $\emph{any}$ depth, within runtime $\mathrm{exp} (O (\mathrm{log}^\mathrm{D} (\varepsilon^{-1} ) ))$, where $\mathrm{D}$ is the spatial dimension of the system.  
%So, for constant precision, the algorithm is efficient for $\emph{any}$ constant spatial dimension, and for inverse polynomial precision, the algorithm is efficient for one-dimensional architectures. 

\begin{figure}[h]
\centering
\includegraphics[width=0.58\textwidth]{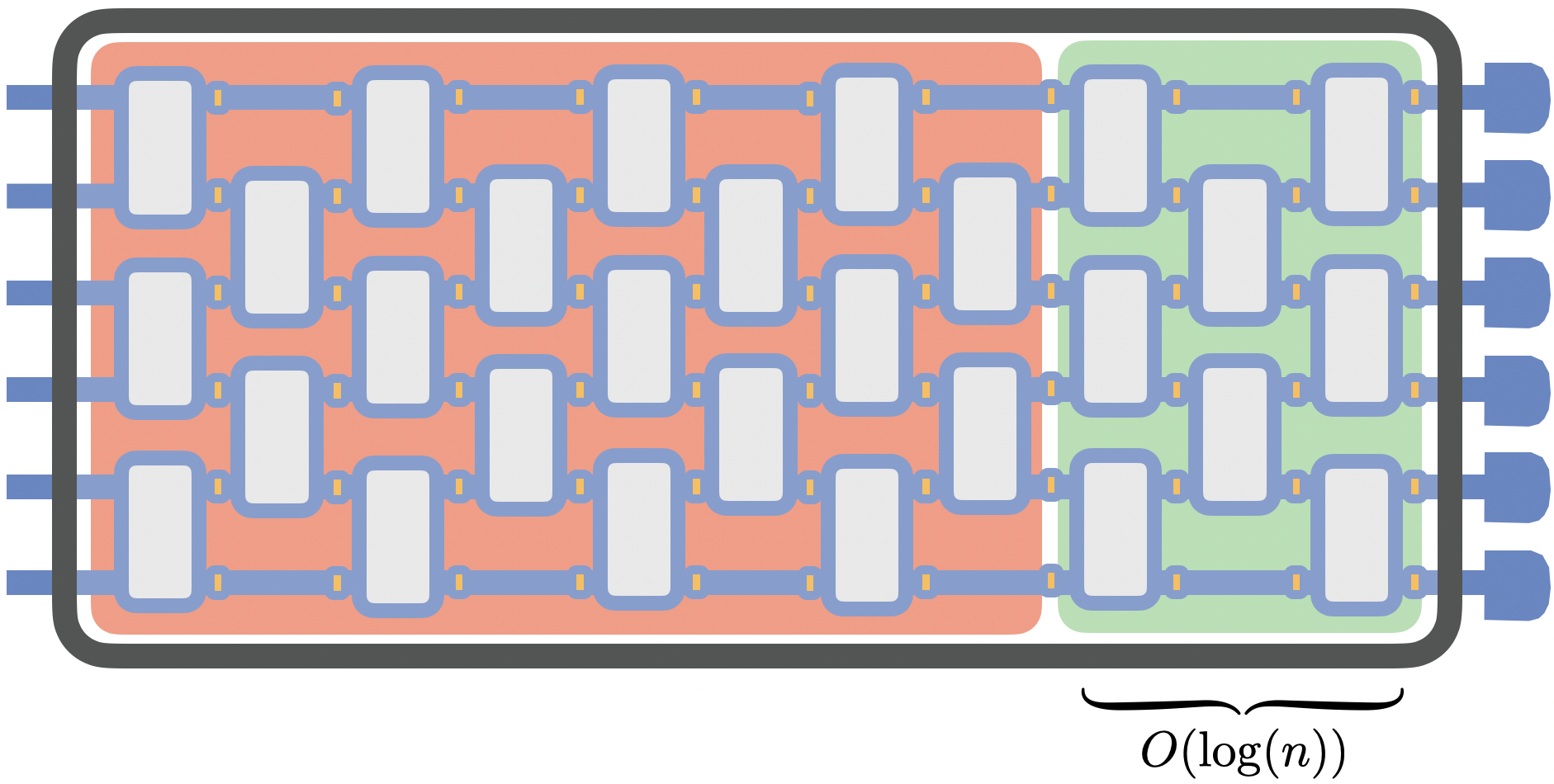}

\caption{For most of the quantum circuits with any possibly non-unital noise only the last $O(\log(n))$ influence significantly observables expectation values. Here, $n$ is the number of qubits.} 
\label{Fig:circ2}
\end{figure}

In summary, our results show that \emph{most} quantum circuits with non-unital noise at \emph{any} depth behave qualitatively as (noisy) shallow quantum circuits for the task of estimating observable expectation values.  
Beyond this task, we further establish that the majority of noisy quantum circuits $\Phi$, with a depth at least linear in the number of qubits, become independent to the initial states -- specifically, for any two states $\rho$ and $\sigma$, the trace distance between $\Phi(\rho)$ and $\Phi(\sigma)$ vanishes exponentially in the number of qubits. 

Note that although we are significantly more general than many previous results when it comes to noise model, our results hold only on average over a well-motivated class of quantum circuits and do not apply to every circuit in the class. However, this limitation is necessary; specifically, it reflects the fact that not every quantum circuit without access to fresh auxiliary systems becomes computationally trivial after a certain number of operations under more general noise, unlike circuits subjected to depolarizing noise~\cite{DanielPaper}. 
%For instance, it is possible to perfectly perform classical computations if the noise model is purely dephasing~\cite{nielsen_chuang_2010}. 
For instance, Ref.~\cite{refrigerator} has shown that it is possible to perform exponentially long quantum computations under non-unital noise, with specially constructed circuits. Because of this, we cannot expect to prove our statements for \emph{all} quantum circuits with non-unital noise.

From a technical perspective, our results rely on bounding various second moments of observable expectation values under noisy random quantum circuits. In particular, we show how combining a normal form of qubit channels~\cite{King2001} with a reduction to ensembles of random Clifford circuits renders most computations tractable. The only assumption we need is satisfied for any architecture where the local gates form $2$-designs~\cite{MeleHaar,Designs}, making our results widely applicable.

Taken together, our results substantially advance our understanding of the noise impact on near-term quantum computation and showcase that, unless we carefully engineer the circuits to take advantage of the noise (e.g., as in Ref.~\cite{refrigerator}), it is unlikely that a quantum computer with non-unital noise is preferable over one with depolarizing noise.

\subsection{Related works}
Previous studies have shown that circuits subjected to possibly non-unital noise can produce output states that are effectively independent from their input states~\cite{fawzi2022lower,ErrorMitigationObstructions}, i.e., they present a notion of an \emph{effective depth}. However, these studies explore regimes significantly and distinctly different from ours. Specifically, Ref.~\cite{fawzi2022lower} has shown that the output states of any circuit with depth more than exponential in the number of qubits $n$, when interleaved with non-unital noise, become effectively independent of the input state. In our work, instead, we show that even linear-depth circuits can accomplish this -- although only on average over the circuit class. The same cannot be shown for worst-case circuits because non-unital noise can be leveraged to perform fault-tolerant quantum computation up to $\exp(O(n))$ depth without needing fresh auxiliary qubits~\cite{refrigerator}.
Furthermore, Ref.~\cite{ErrorMitigationObstructions} has shown that circuits, composed of $\Omega(n)$ layers of global $2$-designs interspersed with non-unital noise, yield output states that are effectively independent of their initial states. In contrast, we only assume that each of the local two-qubit gate forms a $2$-design, which is arguably a more realistic circuit model.
Moreover, Refs.~\cite{limitations2021, de2021quantum} provide contraction bounds for quantum annealers and limitations for the \emph{quantum approximate optimization algorithm} (QAOA) which holds for several families of non-unital noise. Here, we show that most quantum circuits with non-unital noise behave like shallow quantum circuits for the task of estimating expectation values. In this regard, a line of research has delved into the limitations of shallow quantum circuits in terms of their computational capabilities. Notably, across a broad spectrum of combinatorial optimization problems, classical algorithms outperform quantum algorithms implemented by shallow quantum circuits~\cite{bravyi2020obstacles, farhi2020quantum, farhi2020quantum-bis, limitations2021, de2021quantum, anshu2023concentration}. On the other hand, it has been shown that shallow quantum circuits are unconditionally more powerful than shallow classical circuits~\cite{Bravyi_2018,watts2024unconditional,Bravyi_2020SHA}.

The interplay between barren plateaus, random circuits and noise has been thoroughly explored by the previous literature in the context of \emph{variational quantum algorithms} (VQAs)~\cite{Variational}. In the noiseless scenario, it has been highlighted~\cite{McClean_2018,Holmes_2022} that if the parameter distribution underlying the parametrized quantum circuit forms a global $2$-design~\cite{Brand_o_2016}, then any associated cost function exhibits barren plateaus.
The influence of the locality of observables on the onset of barren plateaus has 
been explored in Refs.~\cite{Cerezo_2021, uvarov2021barren, napp2022quantifying}. 
Methods for avoiding or mitigating barren plateaus in noiseless scenarios have been proposed, primarily relying on specific heuristic-based initialization strategies~\cite{Grant_2019,Sack_2022,Mele_2022,PhysRevResearch.5.L032040,rudolph2023synergy, Jain_2022, Holmes_2022,shi2024avoiding}, as well as by constraining the circuit expressibility~\cite{schatzki2022theoretical, zhang2020trainability, Volkoff_2021, Pesah_2021,liu2022mitigating, Meyer_2023, park2023hamiltonian, Larocca_2022, zhang2022escaping,park2024hardwareefficient,Zhang_2024}. 
Furthermore, it has been recently conjectured that methods to `provably' avoid barren plateaus typically enable also efficient classical simulation, either with purely classical resources or after an initial data acquisition phase which may require a quantum computer~\cite{cerezo2023does}.
In the context of noisy scenarios, Ref.~\cite{nibp} showed that under the action of certain types of unital noise, both cost functions and gradients experience exponential decay in the circuit depth. This phenomenon has been dubbed `\emph{noise-induced barren plateaus}' (NIBPs), and it has been considered as a major hurdle for VQAs since NIBPs kick in regardless of the used initialization strategy~\cite{nibp} -- crucially, assuming unital noise.
Another recent work~\cite{schumann2023emergence} shows how barren plateaus can emerge even with non-unital noise when the circuit is made by global $2$-designs interleaved by non-unital noise. Our work is different from Ref.~\cite{schumann2023emergence}, because we consider the arguably more realistic model in which the noise acts after each local gate that composes the circuit. In Ref.~\cite{sannia2023engineered}, the authors observed that properly engineered Markovian noise can prevent barren plateaus, although their analysis holds for a single noisy layer. 

In our work, we show how non-unital noise leads to strikingly different conclusions than unital noise (e.g., depolarizing) in the context of barren plateaus. We prove that parametrized quantum circuits made by $2$-qubit gates (each of them being drawn from a $2$ -design) do not exhibit exponential cost function concentration at any depth, and that the gates in the last $O(\log(n))$ layers of the circuit are (the only) trainable, implying absence of barren plateaus. In contrast, in the depolarizing noise case, no gates are trainable for sufficiently deep circuits~\cite{nibp}, implying exponential cost function concentration and the onset of barren plateaus. In other contexts, it has been pointed out how non-unital noise can lead to drastically different conclusions compared to unital noise. For example, in the context of fault-tolerant quantum computation, it was shown~\cite{refrigerator} that while in the non-unital noise scenario quantum computation is possible up to $\exp(O(n))$ depth without the need for introducing \emph{fresh} auxiliary qubits, in the depolarizing noise scenario the same is possible only up to $O(\log(n))$ depth~\cite{aharonov1996limitations,DanielPaper}, or in the case of dephasing noise only up to $O(\mathrm{poly}(n))$ depth.
Moreover, Ref.~\cite{fefferman2023effect} showed how existing easiness and hardness proofs~\cite{CharactGOOGLE,Deshpande_2022,Harrow_2023} of random circuit sampling break down, under non-unital noise, because the output distribution of the circuit is not `flat enough'---or in more technical terms, it does not `anticoncentrate'~\cite{dalzell2022random}. 
\textcolor{black}{Additionally, various studies from the many-body physics literature have explored the behavior of entanglement in noisy quantum circuits under different types of noise~\cite{Liu_2023,Liu_2024,liu2024noiseinducedphasetransitionshybrid}.}

The pursuit of efficient classical simulation algorithms for quantum circuits has also been addressed in previous research.
The task of simulating expectation values was addressed in several works (see, for instance, Ref.~\cite{Bravyi_2020} for shallow circuits and Ref.~\cite{fontana2023classical, shao2024} for circuits interspersed by Pauli noise). In particular, if the circuit is shallow and the observable is local, then it is well-known that standard light-cone arguments suffice to compute efficiently the expectation value~\cite{Bravyi_2006,Bravyi_2020}. (See Ref.~\cite{Schwarz_2017} for a discussion on the strategy of considering light-cone arguments for estimating expectation values of local observables with respect to tensor network states.)
In the case of noisy tensor network quantum states~\cite{Borregaard_2021,Schwarz_2017}, it has already been pointed out that only the last layers of the circuit preparing such states are sufficient for the estimation of local observables. In our work, we show that most of the states prepared by arbitrary deep noisy random quantum circuits have this property, leading to an efficient (average-case) classical simulation algorithm leveraging this effective shallowness and light-cone arguments. While for sufficiently deep circuits affected by depolarizing noise, estimating the expectation value merely requires outputting the result obtained from the maximally mixed state~\cite{DanielPaper,nibp}---yielding zero for Pauli observables---in scenarios involving non-unital noise, such an approach does not work, as we show that local expectation values of random quantum circuits with non-unital noise are not very concentrated around a fixed value.
For the task of approximately generating samples from the output distribution of random circuits interspersed by depolarizing noise, an efficient classical algorithm has been provided in Ref.~\cite{Aharonov_2023}. However, the same analysis of~\cite{Aharonov_2023} does not straightforwardly generalize to non-unital noise, as noted in Ref.~\cite{fefferman2023effect}, because of the very peaked nature of the output distribution. %because there are significantly more Pauli paths to consider. %``anticoncentration" assumption used does not hold~\cite{fefferman2023effect}. 
Thus, designing an efficient classical algorithm for this sampling task under non-unital noise remains an open problem to date.

\textcolor{black}{On a more technical note, our work contributes to the expanding literature on random quantum circuits by exploring their behavior under general local noise channels. While their behavior has been extensively studied in the noiseless scenario~\cite{Brand_o_2016,Harrow_2023,mittal2023localrandomquantumcircuits,belkin2024approximatetdesignsgenericcircuit,chen2024incompressibilityspectralgapsrandom}, rigorous research on their behavior under noise has been mainly restricted to the unital scenario~\cite{dalzell2022random,Deshpande_2022,Aharonov_2023}, with only a notable exception~\cite{fefferman2023effect}.}

\textcolor{black}{Recently, before our work, an independent study~\cite{singkanipa2024unitalnoisevariationalquantum} has been released that also analyzes the barren plateaus phenomenon beyond unital noise, reaching conclusions in agreement with ours. However, our work differs significantly from theirs both in techniques and assumptions. In particular, our results apply to any local noise channels, whereas the results in Ref.~\cite{singkanipa2024unitalnoisevariationalquantum} are restricted to a specific class of local and global noise channels (the so-called Hilbert-Schmidt-contractive channels).}
\textcolor{black}{After our work has been released, a study~\cite{schuster2024polynomialtimeclassicalalgorithmnoisy} has been published, demonstrating how to simulate classically the expectation values of quantum circuits with certain types of noise with high probability over random choices of input states. Specifically, the noise channels considered have been depolarizing noise and a constrained version of the amplitude damping channel, whereas our classical algorithm is applicable to any noise channel and works on average over the circuit ensemble rather than on average over the inputs.}

\section{Results}

Our work is organized as follows:
\begin{itemize}
    \item[-] In Subsection~\ref{sub:effMAIN}, we prove the effective depth picture.
    Namely, we show that the layers preceding the last $\Theta(\log(n))$ have negligible influence for the task of estimating observable expectation values. Notably, here the only assumption on the noise is that the associated channel is not a unitary channel. Then, we also show that the trace distance between two arbitrary states affected by the same linear depth noisy random quantum circuit becomes exponentially small in the number of qubits. In addition, we also investigate the impact of non-unital noise on worst-case (i.e., any possible) quantum circuit. In particular, we prove that, if the noise strength \emph{exceeds} a certain threshold, the layers at superlogarithmic distance from the end bear a negligible influence on the output state of \emph{any} circuit.

    \item[-] In Subsection~\ref{sub:classicalMAIN}, leveraging the concept of effective shallowness, we present a classical simulation algorithm for estimating local expectation values of arbitrary deep noisy random quantum circuits. The algorithm's core idea capitalizes on the fact that local expectation values of shallow quantum circuits can be computed efficiently through standard light-cone arguments. 
    Furthermore, we discuss a \emph{verification condition} that certifies the success of the classical simulation for a given fixed circuit.

    \item[-] In Subsection~\ref{sub:Lack}, we show that expectation values of local observables of random quantum circuits with non-unital noise are not exponentially concentrated towards a fixed value, in stark contrast to the unital noise scenario. Moreover, we complement the results in Subsection~\ref{sub:effMAIN} by revealing that the last $\Theta(\log(n))$ layers of random quantum circuits with non-unital noise have a non-trivial impact on local expectation values. This implies that non-unital noise induces an absence of barren plateaus at \emph{any} depth, for cost functions defined with respect to local expectation values. 
    Furthermore, we observe that cost functions defined with respect to global Pauli observables exhibit exponential cost-concentration and barren plateaus at any depth. Similarly, we show that fidelity quantum kernels~\cite{thanasilp2022exponential,Jerbi_2023,Schuld_2019}, well-studied quantities in the context of quantum machine learning, experience exponential concentration at any depth under various noise models, including both unital and non-unital channels. %-- in contrast to projected quantum kernels~\cite{Huang_2021,thanasilp2022exponential}. 
    For the special case of unital noise, we also present upper bounds for barren plateaus and fidelity kernels with a substantial improvement compared to previous results~\cite{nibp, thanasilp2022exponential}, which are restricted only to certain types of unital noise, whereas our results holds for general unital noise and have a tighter dependence on the circuit depth.
\end{itemize}
The detailed technical aspects are covered extensively in the Supplementary Material -- here we give an overview of the main results.

\subsection{Preliminaries and definitions}
\label{ref:prelMAIN}
As many of our results rely on understanding how noise behaves \emph{on average} over circuits, we first introduce the circuit ensembles over which these averages are taken.
We consider $n$-qubit quantum circuits $\Phi$ consisting of layers of two-qubit gates interleaved by local (single-qubit) noise, with a final layer of single-qubit gates. All gates are assumed to be drawn from a $2$-design, and we make no assumptions about geometric locality, except where explicitly mentioned.
We express our circuits as
\begin{align}
\Phi \coloneqq \mathcal{V}^{\mathrm{single}} \circ \mathcal{N}^{\otimes n} \circ \mathcal{U}_{L} \circ \cdots \circ \mathcal{N}^{\otimes n} \circ \mathcal{U}_{1},
\label{eq:randcirc_main}
\end{align}
where $\mathcal{V}^{\mathrm{single}} \coloneqq V(\cdot)V^{\dagger}$, with $V \coloneqq \bigotimes^{n}_{i=1} U_i$, is a layer of single-qubit gates, $L$ represents the number of layers, also referred to as circuit depth, $\mathcal{U}_{i}$ corresponds to the channel associated with the $i$-th unitary circuit layer for $i \in [L]\coloneqq\{1,2,\dotsc,L\}$, and $\mathcal{N}$ is a single-qubit quantum channel. Although we find that the final layer of single-qubit gates is not essential for our results and, if desired, might be omitted with minor adjustments, we retain it to simplify our proofs.

Exploiting the fact that we draw gates from a two-qubit 2-design and we consider up to second-moment quantities, without loss of generality, the underlying distributions of unitaries do not change if we compose $\mathcal{N}$ on the left and right by arbitrary unitary channels.
Then, using the so-called `normal' form representation of the channel~\cite{King2001,BETHRUSKAI2002159}, without loss of generality, $\mathcal{N}$ can be defined in terms of two real vectors $\bold{t} \coloneqq (t_X, t_Y, t_Z) \in \mathbb{R}^3$ and $\bold{D} \coloneqq (D_X, D_Y, D_Z) \in \mathbb{R}^3$, describing its action on a quantum state written in its Bloch sphere representation as
\begin{align}
    \mathcal{N} \!\left(\frac{I + \bold{w} \cdot \boldsymbol{\sigma}}{2}\right) = \frac{I}{2} + \frac{1}{2}(\bold{t} +  D\bold{w}) \cdot \boldsymbol{\sigma},
    \label{eq:normSSmain}
\end{align}
where $D\coloneqq\mathrm{diag}(\bold{D})$, $\bold{w} \in \mathbb{R}^3$ with $\|\bold{w}\|_2 \leq 1$, and $\boldsymbol{\sigma} \coloneqq (X, Y, Z)$ is the vector of single-qubit Pauli matrices.
In our work, it is crucial to consider the constant $c \coloneqq \frac{1}{3}(\|\bold{D}\|^2_2 + \|\bold{t}\|^2_2)$, as we show that it quantifies the contraction rate with respect to the depth of our noisy random circuits. Particularly, we show that $c \leq 1$, and the equality holds if and only if $\mathcal{N}$ is a unitary channel. We will often analyze noisy circuits in the Heisenberg picture. In particular, the adjoint channel $\mathcal{N}^{*}$ acts on $Q\in \{X,Y,Z\}$ as 
\begin{align}
\label{eq:adjointeq}
    \mathcal{N}^{*}(Q)=t_{Q} I + D_{Q}Q 
\end{align}
and on $I$ as $\mathcal{N}^{*}(I)= I$.
Note that if the noise is unital, i.e., $\mathcal{N}(I)= I$, then we have $t_Q=0$ for all $Q\in \{X,Y,Z\}$.

\textcolor{black}{We also often need to denote circuits derived from $\Phi$ by retaining only a subset of layers from the start or the end. 
In these cases, we use subscripts to indicate the relevant range of layers. That is, for $a \leq b \in [L]$,
\begin{align}
\Phi_{[a,b]} \coloneqq  
 (\mathcal{V}_b^{\mathrm{single}} \circ \mathcal{N}^{\otimes n} \circ \mathcal{U}_b) \circ \cdots \circ 
 (\mathcal{V}_a^{\mathrm{single}} \circ \mathcal{N}^{\otimes n} \circ \mathcal{U}_{a}).
\end{align}
In particular, $\Phi_{[L-m,L]}$ denotes the truncated noisy circuit obtained by discarding all but the last $m$ layers. 
A more detailed discussion of this notation is provided in the Supplementary Material.}

\textcolor{black}{While we have assumed the same quantum channel $\mathcal{N}$ acting in every layer and on every qubit, this assumption is not crucial and can be relaxed. In particular, our proofs also apply to spatially varying (inhomogeneous) local noise, provided each site is acted upon by a local non-unitary channel; uniformity in space is not required.
Our analysis also naturally extends to the case where each two-qubit gate is followed by a correlated two-qubit noise channel of the form $\mathcal{E} = \tilde{\mathcal{U}}_2 \circ \left( \mathcal{N}_1 \otimes \mathcal{N}_2 \right) \circ \tilde{\mathcal{U}}_1$, where $\tilde{\mathcal{U}}_1$ and $\tilde{\mathcal{U}}_2$ are arbitrary two-qubit unitary channels, and $\mathcal{N}_1$, $\mathcal{N}_2$ are single-qubit non-unital noise channels. Thanks to the left- and right-invariance properties of the Haar distribution used in our random circuit model, our key structural results—including logarithmic effective depth and lack of concentration—continue to hold in this more general setting.
}

%\subsection{Our work in context}
%\label{sub:context}

\subsection{Noise-induced effective shallow circuits}
\label{sub:effMAIN}
Here, we first present a high-level motivation for our investigation of the effective depth of noisy quantum circuits with respect to the task of estimating observables expectation values and then move to present our results formally.

\subsubsection{Effective depth: a high-level motivation}
The question of how uncorrected noise affects quantum computation has gained prominence in the era of \emph{noisy intermediate-scale quantum devices} (NISQ)~\cite{Preskill_2018, chen2022complexity}, where resources for fault-tolerance are scarce.
To address this question for a commonly studied task, we investigate how possibly non-unital noise impacts quantum computations whose output is given by the expectation value of some observable of interest, e.g., a Pauli observable $P$, or an estimation thereof. A centerpiece of our work is the  bound 
\begin{equation}
    \Ex_{\Phi}[|\Tr(P \Phi(\rho - \sigma))|] \leq \exp(-\Omega(\mathrm{depth}(\Phi)))
    \label{eq:counterpiece}
\end{equation}
on how distinct the expectation values of two different states can be, under noise, 
which holds for any two input states $\rho$ and $\sigma$ that are fed into a noisy circuit. Here the expected value $\Ex_{\Phi}$ is taken over the local gates that compose the (possibly non-unital) noisy circuit $\Phi$ and that are assumed to be distributed according to a local $2$-design. The above bound means that, on average over circuits, even two orthogonal input states fed into the same noisy circuit get closer to each other (relative to a Pauli $P$) exponentially fast in the {\em depth} of the circuit. 

Computing the depth of a noisy circuit at which the input information gets ``erased" is one way to understand the destructive effects of noise. This is the basis of the contraction-type results in several early works ~\cite{aharonov1996limitations,Aharonov_2023,alex2021random,nibp} which modelled circuit noise as \emph{single-qubit depolarizing noise}. One result in this line~\cite{DanielPaper} asserts that any circuit $ \Phi^{\mathrm{dep}}$ with depth $L=\Omega(\log(n))$ interspersed with layers of depolarizing noise outputs the maximally-mixed state up to error which vanishes exponentially in $L$ in trace distance as
\begin{equation}\label{eq:dep}
    \lVert \Phi^{\mathrm{dep}}(\rho) - \mathbb{I}/2^n \rVert_1 \leq O(\sqrt{n} b^{L}),
\end{equation}
for some constant $b\in (0,1)$ that is independent of the state $\rho$. This implies that also %$\lVert \Phi^{\mathrm{dep}}(\rho - \sigma))\rVert_1$ 
$|\Tr(P \Phi^{\mathrm{dep}}(\rho - \sigma))|$ is bounded from above by $O(\sqrt{n} b^{L})$, for any two states $\rho$ and $\sigma$ and any Pauli operator $P$. We refer to Section~\ref{subsub:tracedep} for a concise derivation of Eq.~\eqref{eq:dep}.
What is notable about this expression is that a {\em single} state---the maximally-mixed state---is the `limit' to which all circuits affected by depolarizing noise converge, independent of what gates are actually in the circuit, or its input state. The depolarizing noise assumption, if true, enormously constraints noisy quantum algorithms outputting expectation values and without access to fresh auxiliary qubits: to perform useful computation, they should be executable within $O(\log(n))$ depth.
According to Eq.~\eqref{eq:dep}, any noisy circuit with depth larger than $O(\log(n))$ cannot be efficiently distinguished from the maximally mixed state. Consequently, when estimating Pauli expectation values with inverse polynomial accuracy, one can always output the fixed quantity $\Tr(P \, \mathbb{I}/2^n) = 0$.

For circuits with depth $O(1)$ in the number of qubits, in any spatial dimensionality (including all-to-all connectivity), Pauli expectation values of quantum circuits interspersed with depolarizing noise can be computed efficiently. This can be achieved either through standard light-cone arguments (if the observable is local) or by leveraging the fact that Pauli expectation values are exponentially suppressed with the Pauli weight (if the observable is global).
More generally, if the circuit depth is $O(\log(n)^{1/\mathrm{D}})$, where $\mathrm{D}$ is the spatial dimensionality of the circuit, then Pauli expectation values can still be efficiently computed using the same methods. In particular, for one-dimensional architectures, this provides an efficient classical simulation algorithm to estimate Pauli expectation values for any circuit under depolarizing noise, regardless of depth.

However, one could justifiably raise concerns about this conclusion on physical grounds. First, the noise in real hardware is not necessarily depolarizing or even unital~\cite{Kandala,GoogleSupremacy,Chirolli_2008,Pino}; secondly, even a slight departure from the depolarizing noise assumption causes the picture to change dramatically. For instance, consider the single-qubit dephasing noise $\mathcal{D}$, which acts as $\mathcal{D}(\rho)=(1-p) \rho+p \operatorname{diag}(\rho_{0,0}, \rho_{1,1})$ for some $p\in [0,1]$.
Dephasing noise preserves the diagonal elements of its input and thus still falls within the category of unital noise. Now, consider a three-qubit circuit consisting of $L$ consecutive layers of Toffoli gates interspersed with layers of dephasing noise: this noisy circuit acts as the identity on the input states $\rho_0 = \ketbra{0,0,0}{0,0,0}$ and $\sigma_0 = \ketbra{1,0,0}{1,0,0}$, no matter how deep it is. For these two input states, Eq.~\eqref{eq:dep} does not hold. 
This example illustrates the impossibility of drawing any conclusions about expectation values like $\Tr(Z_1 \Phi(\rho))$ that depend only on the {\em depth} of the circuit---when the noise is non-depolarizing, the actual gates in the circuit, as well as the input state $\rho$, can affect the result very much.

Broadening our scope to general non-unital noise, the non-trivial nature of analyzing arbitrary circuits under such noise arises from the following intuition. Unital noise preserves the maximally-mixed state, and so does any unitary. Thus, a circuit consisting of alternating layers of unitaries and unital noise has both its unitary and noise components acting `in tandem' to drive the state towards the maximally-mixed state. 
When the noise is non-unital, however, the unitary and the noise components of the circuit do not have the same fixed point, and thus heuristically act `in different directions'. 
Little is known about the behaviour of these noisy circuit ensembles--- a situation to which we hope to contribute with Eq.~\eqref{eq:counterpiece}. 
The expected value $\mathbb{E}_{\Phi}$ in Eq.~\eqref{eq:counterpiece} is taken to avoid pathological cases: as previously noticed, on worst-case circuits the quantity $|\Tr(P \Phi(\rho - \sigma))|$ may attain a constant value independent of depth for certain kinds of noise, such as the dephasing channel. Moreover, the so-called \emph{quantum refrigerator} construction~\cite{refrigerator} shows surprisingly how non-unital noise can be leveraged to perform fault-tolerant quantum computation for up to exponential depth with circuits similar to ours. Therefore, for these special classes of circuits, we do not expect this worst-case upper bound to hold.

Although contraction results for worst-case circuits have appeared in the literature~\cite{DanielPaper,limitations2021,hirche2022contraction}, they encompass a less general class of channels than those considered in the present work. Thus, proving a contraction result for general---possibly non-unital---noise is far from trivial, and to our knowledge, our work provides the strongest result in this regard, by crucially leveraging the randomness of the circuit.

%We remark that the preceding inequality holds not only for all Pauli $P$, but it holds more generally for any observable $O$ with bounded operator norm $\norm{O}_{\infty}$ (or more generally with bounded normalized Frobenious norm $\norm{O}^2_2/2^n$, as we point out in the appendix).

\subsubsection{Effective depth: formal result}
In this subsection, we formally present our results concerning the effective depth of arbitrarily deep random quantum circuits under possibly non-unital noise, with respect to the task of estimating observable expectation values. Our results reveal that the influence of gates on expectation values decreases exponentially with their distance from the last layer. See Fig.~\ref{Fig:circ2} for a graphical representation. Our main result can be stated as follows.

\begin{theorem}[Effective logarithmic depth]
\label{th:lastlogMAIN}
   Let \( O \) be an observable, \( \rho_0 \) be any initial state \textcolor{black}{(possibly complex)}, \( L \) be the depth of the noisy circuit \( \Phi \), and \( m \in \mathbb{N} \). We assume that the noise in the circuit is local and decomposable into single-qubit non-unitary channels.
   Then, we have
    \begin{align}
       \Ex_{\Phi_{[L-m,L]}} \left[ \left| \Tr(O\Phi(\rho_0)) - \Tr(O\Phi_{[L-m,L]}(\sigma_0)) \right| \right] \le  \norm{O}_{\infty} \exp(-\alpha m),
    \end{align} 
    where $\sigma_0$ is any preferred initial state.
    Here, \( \Phi_{[L-m,L]}(\cdot) \) refers to the noisy circuit where only the last \( m \) layers are considered, the average \( \Ex \) is taken with respect to the 2-design distribution of every two-qubit gate that composes the circuit part \( \Phi_{[L-m,L]} \), and \(\alpha > 0\) is a quantity which depends only on the noise parameters.
\end{theorem}

\textcolor{black}{We remark that the previous bound can be strengthened: instead of \( \norm{O}_{\infty} \), one can use the so-called normalized Frobenious norm \( \norm{O}^2_2 / 2^n \).}
By Markov's inequality, this result directly implies that, with high probability over the choice of the noisy quantum circuits, considering only the last $\Theta(\log(n))$ layers suffices to estimate any observable expectation value with precision scaling inverse-polynomially with the number of qubits.

The previous result is implied by the following statement.
\begin{theorem}[General scaling]
\label{th:theor1}
Let $\rho$ and $\sigma$ be arbitrary quantum states, $P\in \{I,X,Y,Z\}^{\otimes n}$ with Pauli weight $|P|$, and $m$ be the depth of the noisy circuit. We assume that the noise in the circuit is local and decomposable into single-qubit non-unitary channels. Then, we have
\begin{align}
    \Ex_{\Phi}[(\Tr(P\Phi(\rho-\sigma)))^2] \leq \exp(-\Omega(m+|P|)).
\end{align}
\end{theorem}
Specifically, the upper bound that we prove is $O(c^{m+|P|})$, where $c<1$ is the constant defined in the preliminary Subsection~\ref{ref:prelMAIN}. We now provide a sketch of the proof of Theorem~\ref{th:theor1}. At a high level, the proof works by going into the Heisenberg picture, `peeling off' a unitary and noisy layer of the depth-$m$ circuit, and then applying the adjoint of these layers to the Pauli $P$. Using properties of the random circuit and that of the noise channel, we then show that this recovers a Pauli expectation value on a noisy circuit of depth $m-1$, which puts us in the position to reiterate the argument.

\begin{proof}[Proof sketch of Theorem~\ref{th:theor1}] 
Because our local two-qubit gates are drawn from a $2$-design and we are computing a second moment, we can assume that all gates are Clifford~\cite{gottesman1998heisenberg} and that the noise channels are in their normal form. 
Let $\Tilde{\Phi}$ be a noisy circuit of depth $m-1$ and let $\Phi$ be a noisy circuit of depth $m$. That is, $\Phi = \mathcal{V}^{\mathrm{single}} \circ \mathcal{N}^{\otimes n}\circ \mathcal{U}_m\circ \Tilde{\Phi}$, where we have used the same notation defined in Eq.~\eqref{eq:randcirc_main}.
Averaging over the last layer of single-qubit unitaries, we have
\begin{align}
    \Ex[(\Tr(P \Phi(\rho-\sigma)  ))^2]=\frac{1}{3^{|P|}}\sum_{\substack{Q \in \{I,X,Y,Z\}^{\otimes n}:\\ \supp(Q)=\supp(P) } }\Ex\left[\left(\Tr(Q \, \mathcal{N}^{\otimes n}\circ \mathcal{U}_m\circ \Tilde{\Phi}(\rho-\sigma))\right)^2\right].
\end{align}
Taking the adjoint of the layer of noise $\mathcal{N}^{\otimes n}$ and using properties of the random gates that compose the unitary layer $\mathcal{U}_m$, we have 
\begin{align}
    \Ex\!\left[\left(\Tr(Q \,\mathcal{N}^{\otimes n}\circ \mathcal{U}_m\circ \Tilde{\Phi}(\rho-\sigma) )\right)^2\right] & =\sum_{a\in \{0,1\}^{|Q|}}\prod_{j \in \supp(Q)} (t^{a_j}_{Q_j } D^{1-a_j}_{Q_j})^2\Ex\!\left[\Tr(\left(\bigotimes_{k \in \supp(Q)}  Q^{1-a_k}_k\right) \mathcal{U}_m\circ \Tilde{\Phi}(\rho-\sigma) )^2\right].
\end{align}
Thus, taking the adjoint of the last unitary layer on the Pauli and using the fact that it is made of Clifford gates, which map a Pauli observable to another Pauli observable, we have
\begin{align}
    \nonumber
     \Ex[\Tr(P \Phi(\rho-\sigma)  )^2]&\le \frac{1}{3^{|P|}}\sum_{\substack{Q \in \{I,X,Y,Z\}^{\otimes n}:\\ \supp(Q)=\supp(P) } }\sum_{a\in \{0,1\}^{|Q|}}\prod_{j \in \supp(Q)} (t^{a_j}_{Q_j } D^{1-a_j}_{Q_j})^2\max_{Q \in \{I,X,Y,Z\}^{\otimes n}}\Ex\!\left[\Tr(Q\,  \Tilde{\Phi}(\rho-\sigma))^2\right]\\
     &= \frac{1}{3^{|P|}}(\|\bold{D}\|^2_2+\|\bold{t}\|^2_2)^{|P|} \max_{Q \in \{I,X,Y,Z\}^{\otimes n}}\Ex\!\left[\Tr(Q \Tilde{\Phi}(\rho-\sigma))^2\right]\\
     \nonumber
     &=c^{|P|} \max_{Q \in \{I,X,Y,Z\}^{\otimes n}}\Ex\!\left[\Tr(Q \Tilde{\Phi}(\rho-\sigma))^2\right],
     \nonumber
\end{align}
where we have used the multinomial theorem and the definition of $c$ given in Subsection~\ref{ref:prelMAIN}. In the Supplementary Material, we prove that $c< 1$ for every single-qubit channel that is not unitary.
Moreover, we can assume that the maximum over the Paulis is not achieved by the identity because otherwise the RHS would be zero since $\Tilde{\Phi}$ is trace preserving and $\rho-\sigma$ is traceless. Reiterating the argument to the remaining circuit layers establishes the claim.
\end{proof}
\begin{comment}
\begin{proof}[Proof sketch]
    The proof is conducted by working in the ``Heisenberg picture". More specifically, we start from the end of the circuit and take the adjoint of the last layer of noise and $2$-qubit gates layer over the Pauli observable. By leveraging the fact that our circuit is made by local $2$-design gates, we obtain an upper bound of the same form as the quantity we started with, but with the rescaled factor $c^{|P|}\le c$, where $c<1$. Hence, by reiterating the strategy for each of the remaining $m-1$ layers, we get a contraction factor of at least $c^{m-1}$. In the end, H\"older inequality suffices to conclude.
\end{proof}
\end{comment}
%\noindent 
By Jensen's inequality, we obtain 
\begin{equation}
\Ex_{\Phi}[|\Tr(P \Phi(\rho)) - \Tr(P \Phi(\sigma))|] \leq \exp(-\Omega(m+|P|)).
\end{equation}

As an implication of such result, we can show that for depth greater or equal to linear $m= \Omega(n)$, the average trace distance between two states affected by the same noisy quantum circuit becomes exponentially small in the depth as 
\begin{align}
       \Ex_{\Phi}[\norm{\Phi(\rho) - \Phi(\sigma)}_{1}]\le \exp(-\Omega(m)).
       \label{eq:trmain}
\end{align}  
This implies that the application of the same linear depth {\em random} circuit affected by {\em any} amount of noise on two different input states renders them effectively indistinguishable (because of the Holevo-Helstrom theorem~\cite{wilde_2013}).
To our knowledge, this is the first result of this kind; except for the result of Ref.~\cite{fawzi2022lower}, which only applies to exponential depths but which holds for worst-case circuits, whereas our statement hold on average. As mentioned in Subsection~\ref{sub:effMAIN}, it is in principle not possible to prove our result for worst-case non-unital noisy circuits, because there are some special classes of circuits~\cite{refrigerator} that would violate a worst-case version of our inequality (i.e., Eq.\eqref{eq:trmain} without the expectation value).
However, when the noise strength exceeds a certain threshold, we can show a worst-case upper bound on the trace distance that decays exponentially in the number of qubits. Specifically, we show the following.

\begin{proposition}[Worst-case effective depth with high noise]
Let \(\mathcal N\) be a single-qubit channel in normal form and assume
\(D_P\neq 1\) for every \(P\in\{X,Y,Z\}\). Let \(\tau\) be its unique fixed
point and let \(\mathcal E_\tau(X):=\operatorname{Tr}(X)\tau\). Define $b_{\mathcal N}:=3\|\mathcal N-\mathcal E_\tau\|_\diamond$. 
Let \(\Phi\) be a noisy quantum circuit of depth \(m\). Then
\begin{align}
\|\Phi(\rho)-\Phi(\sigma)\|_1
\le
n b_{\mathcal N}^{m}\|\rho-\sigma\|_1 .
\end{align}
Thus, if \(b_{\mathcal N}<1\), then for
\begin{align}
m\ge
\frac{\log(2n/\varepsilon)}{\log(b_{\mathcal N}^{-1})}    
\end{align}
we have \(\|\Phi(\rho)-\Phi(\sigma)\|_1\le\varepsilon\) for all states
\(\rho,\sigma\). A sufficient coefficient-only condition is
$12\max_{P\in\{X,Y,Z\}}
\left|\frac{D_P}{1-D_P}\right|<1$ .

If \(0\le D_P<1\) for all \(P\), this is implied by \(D_P<1/13\) for all \(P\).
\end{proposition}
 
The proof relies on passing to the quantum Wasserstein distance of order $1$~\cite{de2021quantum}, computing contraction coefficients and passing back to the trace distance using techniques 
laid out in Ref.~\cite{de2021quantum}. Such results are known in the literature as reverse threshold theorems, as they show that for high enough noise rate error correction is not possible\ \cite{razborov2003upper, kempe2008upper, hirche2023quantum}. In contrast to previous results, we also extend them to non-unital channels.

\subsection{Classical simulation of random quantum circuits with possibly non-unital noise}
\label{sub:classicalMAIN}
%\textcolor{black}{Re-structured many of things in this subsection}
\textcolor{black}{In this section, we address the problem of estimating expectation values of noisy random quantum circuits under any noise. More specifically, if we are given an instance of a noisy circuit $\Phi$ where the 2-qubit gates are sampled uniformly at random in a fixed architecture, the problem that we want to solve is to estimate the expectation value of $\Tr(O\Phi(\rho_0))$ with accuracy $\varepsilon$ with high probability over the choice of the random circuit, where $O$ is a given observable and $\rho_0$ is an initial state. Here, we focus on $O$ being a linear combination of polynomially $M$ many local Pauli operators. We focus only on local Pauli operators because it can be shown that the high Pauli weight components do not matter, i.e., they get exponentially suppressed with their Pauli weight. The time complexity of the algorithm will be linear in $M$, so without loss of generality, we consider $M=1$ for simplicity, that is, the task of estimating the expectation value of a given local Pauli operator.}

\subsubsection{Average classical simulation of local expectation values}

We have shown that the presence of any non-unitary noise renders \emph{most} circuits effectively shallow for the task of estimating observable expectation values. Specifically, Theorem~\ref{th:lastlogMAIN} establishes that for any $L$-depth circuit architecture $\Phi$, any Pauli operator $P$, and any initial state $\rho_0$, the  inequality
\begin{align}
    \Ex_{\Phi_{[L-m,L]}}\left[ \left|\Tr(P\Phi(\rho_0)) - \Tr(P\Phi_{[L-m,L]}(\ketbra{0^n}{0^n})) \right| \right] \leq \exp(-\Omega(m + |P|))
    \label{eq:classicalstart}
\end{align}
holds,  
where $\Phi_{[L-m,L]}(\cdot)$ denotes the noisy circuit considering only the last $m$ layers.

Equation~\eqref{eq:classicalstart} leads to a straightforward classical algorithm for estimating local Pauli expectation values with accuracy exponentially small in $m$. This algorithm operates in the Heisenberg picture by propagating the local Pauli operator $P$ backwards through only $m$ layers. Specifically, one computes the matrix $P_m \coloneqq \Phi^{*}_{[L-m,L]}(P)$ and then evaluates $\Tr(P_m \ketbra{0^n}{0^n})$.

The following proposition provides the formal guarantees of this classical simulation approach, which can be proven using Markov's inequality applied to Eq.~\eqref{eq:classicalstart}, along with its time complexity derived via light-cone arguments.

\begin{proposition}[Average classical simulation of local expectation values]
\label{prop:classimMA}
Let $\varepsilon, \delta > 0$. Consider a local Pauli operator $P$ and any initial state $\rho_0$. For a noisy quantum circuit $\Phi$ of depth $L$, sampled according to the described circuit distribution, there exists a classical algorithm that outputs a value $\hat{C}$ satisfying
\begin{align}
    | \hat{C} - \Tr(P \Phi(\rho_0)) | \leq \varepsilon
\end{align}
with success probability at least $1 - \delta$ over the choice of the random circuit. Specifically, the classical algorithm involves computing $\hat{C} \coloneqq \Tr(P\Phi_{[L-m,L]}(\ketbra{0^n}{0^n}))$ with
\begin{align}
    m \coloneqq \left\lceil \frac{1}{\log(c^{-1})} \log\left(\frac{4}{\delta \varepsilon^2}\right)\right\rceil,
\end{align}
where $c$ is the noise parameter defined in Subsection~\ref{ref:prelMAIN}.
The time complexity of this algorithm is given by
\begin{align}
    \text{Runtime} &\leq 
    \begin{cases}
        \exp\!\left(O\!\left( m^{\mathrm{D}} \right)\right) = \exp\!\left(O\!\left(\log^{\mathrm{D}}(\varepsilon^{-1})\right)\right), & \text{for $\mathrm{D}$-geometrically-local architectures}, \\[1ex]
        \exp\!\left(\exp(O(m))\right) =  \exp\!\left(\mathrm{poly}(\varepsilon^{-1})\right), & \text{for all-to-all connected architectures}.
    \end{cases}
    \label{eq:runt}
\end{align}
where in the last equation we assumed constant noise rate $c$ and constant failure probability $\delta$.
\end{proposition}

In particular, if the desired accuracy is a constant $\varepsilon = O(1)$, then the algorithm is efficient for any architecture. If the desired accuracy is inverse-polynomial in the number of qubits, then the algorithm runs in polynomial time for one-dimensional architectures and quasi-polynomial time for higher-dimensional ones.

\textcolor{black}{This time-complexity bound arises from the complexity of evaluating $\hat{C} = \Tr(P_m \ketbra{0^n}{0^n})$, with $P_m \coloneqq \Phi^{*}_{[L-m,L]}(P)$, which is exponential in the number of effective qubits over which $P_m$ is supported (i.e., the light-cone of $P$ with respect to $\Phi^{*}_{[L-m,L]}$). For geometrically-local circuit architectures with constant spatial dimension $\mathrm{D}$, the number of qubits on which $P_m$ is supported is $O(m^{\mathrm{D}})$. In all-to-all connected circuit architectures, where no geometrical locality is assumed, the number of qubits on which $P_m$ is supported is at most $2^m$.}

Furthermore, in the Supplementary Material, we provide an alternative \emph{early-break condition} that, if met at some step $t$, guarantees an $\eps$ approximation. This allows us to stop 'unraveling' the circuit and simply output $P_t \coloneqq \Phi^{*}_{[L-t,L]}(P)$ with the current value of $P_t$.
The condition to check is 
\begin{equation}\label{eq:SDP_1MA}
    \min_{c\in \mathbb{R}} \, \lVert P_t - c \Id \rVert_{\infty} \le \varepsilon/2.
\end{equation}
Moreover, if this condition is satisfied, we can be certain to have successfully estimated the expectation value. 
Checking this condition can be done with the same runtime reported in Eq.~\eqref{eq:runt} by solving a \emph{semi-definite problem} (SDP) with an analytical solution which we provide in the Supplementary Material. The intuition behind this verification step is that if $P_t$ were proportional to the identity, then adding further (adjoint) layers would not change the matrix $P_t$, due to unitality of the adjoint channel.

Upon inspecting our proof of Theorem~\ref{th:theor1}, it is evident that we often make conservative estimates on the contraction at each step, as it is proportional to exponential in the locality of the input Pauli, and we always bound this from below by $1$. Thus, by incorporating a finer analysis of the dynamics of the weight, we conjecture that it is possible to show that this early-break condition will be met before the runtime stated in Proposition~\ref{prop:classimMA} with high probability. We leave the proof of this conjecture open for future work.

The above algorithm can be efficient if the weight of the Pauli $P$ is $|P|=O(\log(n))$ (i.e.,  its time complexity runs polynomially in the number of qubits).
However, in the regime of high Pauli weight, no algorithm execution is necessary. Instead, we can simply output zero and ensure, with high probability, an inverse-polynomial accuracy in the number of qubits, as detailed in our Supplementary Material. This stems from the fact that for any noise that is not a unitary channel, we have
\begin{align}
    \Ex_{\Phi}(\Tr(P \Phi(\rho_0)))=0, \quad \Var_{\Phi}(\Tr(P \Phi(\rho_0)))\le\exp(\Omega(-|P|)),
\end{align}
followed by an application of Chebyshev's inequality. We can then summarize our result as follows.
\begin{remark}[Estimating classically any Pauli expectation value]
    For any target constant accuracy any circuit architecture, there is an efficient classical algorithm for estimating any Pauli expectation value of noisy random quantum circuits at any depth. For any inverse polynomial accuracy, the runtime of the algorithm is polynomial in the number of qubits for one-dimensional architectures and quasi-polynomial for higher-dimensional ones.
\end{remark}
We point out that our classical algorithm is not restricted to only Pauli observables; it can be similarly applied to any observable with a bounded operator norm that can be expressed as a linear combination of a polynomial number of Pauli expectation values.

It is also worth mentioning that quantum circuits with non-unital noise can have Pauli expectation values of local observables significantly far from zero also in the high-depth regime. We show this in the next section. This is in stark contrast to the case of circuits with, e.g., depolarizing noise, in which Pauli expectation values decay exponentially to zero in the circuit depth~\cite{DanielPaper,nibp}, which implies that outputting zero suffices to estimate accurately the expectation value for sufficiently deep circuits. This strategy does not work as a classical simulation algorithm for circuits affected by non-unital noise, for the reason mentioned in the first sentence of this paragraph. 

\subsubsection{Improved classical simulation}
\textcolor{black}{In the previous section, we introduced a classical algorithm that estimates local Pauli expectation values with runtime $\exp(O(\log(\varepsilon^{-1}))^D)$, where $\varepsilon$ is the target accuracy, for any fixed $D$-dimensional architecture and any initial state. For all-to-all architectures, the runtime was $\exp(\mathrm{poly}(\varepsilon^{-1}))$. Consequently, this algorithm is efficient whenever $\varepsilon$ is constant in the number of qubits $n$. However, for inverse-polynomial accuracy, the runtime scales as $\exp(O(\log n)^D) = n^{O(\log n)^{D-1}}$ in constant dimensions and becomes exponential for all-to-all architectures.}

We now present an improved algorithm that substantially reduces these runtimes. Its runtime is $n^{O((D-1)\log\!\log n)}$ for geometrically local architectures in any dimension $D$ and $n^{O(\log n)}$ for all-to-all connected architectures, independent of both circuit depth and initial state. Since $\log\!\log n$ grows extremely slowly, the former can be regarded as effectively polynomial for all practical system sizes. The algorithm combines the effective-depth result with an additional truncation in Pauli weight: in addition to retaining only the last few layers and applying a light-cone argument as in the previous section, we further restrict the simulation to Pauli strings of small weight within the light cone. As we showed that high-weight Pauli terms contribute only exponentially little to the expectation value, they can be safely discarded. 

The main result concerning this algorithm is summarized below, while further details and a full proof are provided in Appendix~\ref{app:improved-simulation}.
\begin{theorem}[Improved classical algorithm]
Let $\Phi$ be a quantum circuit sampled uniformly at random from a fixed architecture, possibly including any type of local (non-unitary) noise and any number of layers. For any Pauli operator $P \in \{I,X,Y,Z\}^{\otimes n}$ of constant weight and any initial state $\rho_0$, the expectation value $\Tr[P \Phi(\rho_0)]$ can be estimated to accuracy $\varepsilon = \mathrm{poly}(n^{-1})$ with high success probability. The runtime of this algorithm is $n^{O((D-1)\log\!\log n)}\mathrm{poly}(n)$ for $D$-dimensional geometrically local architectures and $n^{O(\log n)}$ for all-to-all connected architectures.
\end{theorem}
The proof of correctness of this algorithm follows from combining the effective-depth and Pauli-weight suppression results developed in this work with the techniques from Ref.~\cite{angrisani2024classicallyestimatingobservablesnoiseless}, originally introduced for noiseless random circuits, which we extend to the noisy setting.

\begin{comment}
Specifically, the proof leverages the following more general result (see Theorem~\ref{th:mainth} in the appendix for a detailed version), which combines the effective-depth truncation of noisy random circuits with the truncation of high-weight Pauli operators within the light cone.
\begin{theorem}[Truncation in depth and Pauli weight; informal version of Theorem~\ref{th:mainth}]
Let $\rho_0$ be an initial state, $O$ an observable, and $\Phi \coloneqq \Phi_L \circ \cdots \circ \Phi_1$ a noisy quantum circuit, where $\{\Phi_j\}_{j=1}^L$ is the sequence of noisy circuit layers.
Define $O_{\Phi}^{(k,m)}$ to be the Heisenberg-evolved observable $O$ where the last $m$ layers of the circuit are kept and, layer by layer, truncated to Pauli weight $k$ (see Definition~\ref{def:truncwd} in the appendix).
Then
\begin{align}
\mathbb{E}_{\Phi}\big[ \big| \Tr(O \Phi(\rho_0)) - \Tr(O_{\Phi}^{(k,m)} \sigma_0) \big| \big]
    \le 2 \|O\|_{\infty} e^{-\alpha m} + m \left(\tfrac{2}{3}\right)^{k/2} \|O\|_{\infty},
\end{align}
where $\sigma_0$ is any preferred initial state.
Here, the average $\mathbb{E}_{\Phi}$ is taken over the $2$-design distribution of every two-qubit gate in the circuit, and $\alpha > 0$ is a constant depending only on the noise parameters.
\end{theorem}
\end{comment}

\subsection{Lack of barren plateaus with non-unital noise, but only the last $\Theta(\log(n))$ layers matter}\label{sub:Lack}
 
\noindent The \emph{barren plateaus phenomenon} \cite{McClean_2018,nibp}, has been considered one of the main hurdles for \emph{variational quantum algorithms} (VQAs)~\cite{Variational}. These algorithms involve encoding the solution to a problem in the minimization of a cost function, typically defined in terms of expectation values of observables, with the free parameters for optimization being the gate parameters. %The `barren plateaus' phenomenon is where the norm of the gradient of the cost function, on average, vanishes exponentially in the number $n$ of qubits with high probability over circuit instances, resulting in the need to take exponentially-many circuit runs to determine the gradient. This loses the quantum advantage. 
Consequently, in the presence of barren plateaus, randomly selecting an instance of the parameterized circuit would overwhelmingly lead to a circuit instance situated within a region of the landscape necessitating an exponential number of measurements to navigate, implying loss of any potential quantum advantage. 

There are two signatures of barren plateaus: exponential concentration of the cost function around a fixed value, and the fact that the norm of the gradient of the cost function is exponentially small. We show that both are avoided in the non-unital noise scenario for cost functions made by local observables, in stark constrast with the noiseless~\cite{McClean_2018} and unital~\cite{nibp} scenario. However, we will see that this absence of barren plateaus is only due to the last $\Theta(\log(n))$ layers of the circuit, while the gradient contribution coming from the preceding layers is negligible. This is essentially a consequence of what we have explored in the preceding sections: we are working with circuits that exhibit an effective logarithmic depth.

\subsubsection{Lack of exponential concentration for local cost functions with non-unital noise}
\noindent 
We consider a circuit architecture as we have described in Subsection~\ref{ref:prelMAIN}, where the local noise channels are characterized by the parameters of their normal form representation $\bold{t} \coloneqq (t_X,t_Y,t_Z)$ and $\bold{D} \coloneqq (D_X,D_Y,D_Z)$, which we assume to be constants with respect to the number of qubits. 
Our first main theorem is the following.

\begin{theorem}[Variance of expectation values of random circuits with non-unital noise]
\label{th:varianceMAIN}
Let $H \coloneqq \sum_{P\in \{I,X,Y,Z\}^{\otimes n}} a_P P$ be an arbitrary Hamiltonian, with $a_P\in \mathbb{R}$ for $P\in \{I,X,Y,Z\}^{\otimes n}$. Let $\rho$ be a quantum state. We assume that the noise is non-unital and that $\norm{\bold{t}}_2=\Theta(1)$. Then, at any depth, we have
\begin{align}
       \Var_{\Phi}[\Tr(H \Phi(\rho)  )] = \sum_{P\in \{I,X,Y,Z\}^{\otimes n}\setminus I^{\otimes n}} a_P^2 \exp(-\Theta(|P|)).
\end{align}  
\end{theorem}
%The main idea behind the proof is to start from the end of the circuit, use the Heisenberg picture to simplify the expression -- in other words, `peel off' the last layer of single qubit random gates and last layer of noise and apply their adjoint operators to the Pauli operator under consideration, and then use properties of the noise channel and random circuits to get the desired lower bound. 
The proof of this theorem is provided in the Supplementary Material. Note that the variance scaling in the non-unital noise scenario is independent of the circuit depth, which contrasts sharply with noiseless random quantum circuits or circuits with unital noise \cite{napp2022quantifying,nibp}, where the variance becomes exponentially small in the number of qubits at sufficiently high depth. Our Theorem~\ref{th:varianceMAIN} directly implies that, under non-unital noise, the variance of \emph{local cost functions}—i.e., expectation values of local observables, e.g., $\Tr(Z_1 \Phi(\rho))$—is significantly large.
\begin{corollary}[Lack of exponential concentration for local cost functions]
\label{th:variance_informal1}
Let $\rho$ be any quantum state, $P$ be a local Pauli ($|P|=\Theta(1)$), and $\Phi$ be the noisy random circuit ansatz. Under the same noise assumptions as above, at any circuit depth, we have
\begin{align}
       \Var_{\Phi}[\Tr(P \Phi(\rho)  )] = \Theta(1).
\end{align}  
\end{corollary}
In particular, the lower bound on the variance that we prove is $\left(\|t\|_2^2/3\right)^{|P|}.$ Here, $\|\bold{t}\|_2$ is the noise-parameter, which quantifies the non-unitality of the channel. As expected, if the noise is unital ($\|\bold{t}\|_2=0$), we get a vacuous zero lower bound. It is interesting to note that even a \emph{tiny} deviation from the unitality assumption, i.e., if the non-unital noise parameter $\|\bold{t}\|_2$ scales inverse-polynomially with the number of qubits, the variance is not exponentially small. This implies that local expectation values under non-unital noise are not `too concentrated' around their mean value $\Ex_{\Phi}[\Tr(P \Phi(\rho)  )]=0$, in contrast with the depolarizing noise scenario, in which case the variance decays exponentially with the depth~\cite{nibp}.  
However, we show that global cost functions---i.e., expectation values of high Pauli weight observables, e.g., $\Tr(Z^{\otimes n} \Phi(\rho))$---still exhibit exponential concentration, as in the noiseless case~\cite{Cerezo_2021,napp2022quantifying} and in the depolarizing noise scenario. 
Formally, we find the following cost concentration.

\begin{corollary}[Cost concentration for global expectation values]
\label{global expectation}
Let $P\in \{I,X,Y,Z\}^{\otimes n}$ with Pauli weight $|P|=\Theta(n)$. Assuming that the noise is not a unitary channel, for any constant noise parameters, we have
\begin{align}
    \Var_{\Phi}[\Tr(P \Phi(\rho) )] \leq  \exp(-\Omega(n)).
\end{align}
\end{corollary}

%The intuition why local cost functions are not too concentrated with non-unital noise is the following. We have seen that for local observables, non--unital noise also ``truncate" the circuit to have logarithmic depth, as is shown in Theorem \ref{th:lastlogMAIN} and shallow circuits are known to not concentrated too much~\cite{napp2022quantifying}. This intuition is borne out by Theorem~\ref{th:variance_informal1}.

\subsubsection{Lack of vanishing gradients for local cost functions, but only the last $\Theta(\log(n))$ layers are trainable}
\noindent
The effective depth of noisy random circuits, elucidated in Section~\ref{sub:effMAIN}, implies that, on average, changing the gates in the layers preceding the last $\Theta(\log(n))$ will not significantly influence the observable expectation value. This suggests that only the last few layers of the circuit can substantially alter the expectation value. We analyze this formally using the notion of `trainability' of a parametrized gate. As is common in the literature \cite{Variational,cerezo2023does}, we refer to a cost function $C$ as trainable with respect to a parametrized gate labeled by $\mu$ if and only if $\Var[\partial_\mu C]=\Omega(1/\mathrm{poly}(n))$, that is, the variance of the cost function partial derivative with respect to such a parameter vanishes no more than inverse-polynomially in the number of qubits. For a formal definition of cost function partial derivatives and detailed proofs of the following results, we refer to the Supplementary Material.
We rigorously show that for local cost functions, under non-unital noise, only the last $\Theta(\log(n))$ layers of the circuit are trainable, meaning they have significantly large partial derivatives. 
\begin{theorem}[Only the last few layers are trainable for local cost functions]
\label{th:trainabilityy}
    Let $C = \Tr(P \Phi(\rho_0))$ be a cost function associated with a local Pauli $P$ (with Pauli weight $|P|=\Theta(1)$), an initial state $\rho_0$, and a depth-$L$ noisy random circuit ansatz $\Phi$  in arbitrary dimension. Let $\mu$ be a parameter (in the light cone) of the $k$-th layer. Assume that the noise is not unital and it is not a replacer channel (i.e., it does not output a fixed state). Then, we have
    \begin{align}
        \Var[\partial_\mu C] =\exp(- \Theta(L-k)).
    \end{align}
\end{theorem}
Theorem~\ref{th:trainabilityy} states that local cost functions are mostly sensitive to changes of parameters in the last few layers. This is in stark contrast to the case of depolarizing noise or the noiseless case, in which the partial derivative variances with respect to gates in \emph{all} the layers are exponentially suppressed with the number of qubits at sufficiently high depth~\cite{nibp,McClean_2018,napp2022quantifying}. 

As a corollary of the previous theorem, we have that the norm of the gradient of local cost functions does not exponentially vanish in the number of qubits at any depth. Crucially, this fact is only due to the last few layers of the circuits, which we have shown to be trainable. 
\begin{corollary}[Non-vanishing gradient for local cost functions with non-unital noise]
\label{th:variance_informal}
Let $\rho_0$ be any quantum state and $P$ be a local Pauli. Let $C:=\Tr(P \Phi(\rho_0)  )$ be the cost function associated with a noisy random quantum circuit ansatz of any depth. We assume that the noise is non-unital and it is not a replacer channel (i.e., it does not output a fixed state). Then, we have
\begin{align}
       \Ex[\norm{\nabla C}_2^{2}] = \Theta(1).
\end{align}  
\end{corollary}
In the Supplementary Material, we also show that global cost functions---i.e., those cost functions associated with high-weight Pauli observables---have all partial derivative variances exponentially vanishing in the number of qubits, and therefore also their gradients.
Furthermore, in the Supplementary Material, we present numerical simulations that corroborate our findings and attempt to extend beyond the assumptions of our theorems, such as the assumption that two-qubit gates are sampled from a $2$-design. We provide evidence that even for more restricted classes of ans\"{a}tze, such as QAOA~\cite{farhi2014quantum}, the same qualitative behavior is observed.

\subsubsection{Improved upper bounds for barren plateaus in the unital noise scenario}
The technical tools developed in this work also allow us, in the unital-noise scenario, to improve upon the barren plateaus upper bounds presented in Ref.~\cite{nibp}.
In particular, we establish the following.
\begin{proposition}[Improved upper bound on the partial derivative for unital noise]
\label{th:ubvarianceUNITMAIN}
Let $C\coloneqq \Tr(P\Phi(\rho_0))$ be the cost function, where $P$ is a Pauli, $\rho_0$ is an arbitrary initial state, and $\Phi$ is a random quantum circuit ansatz of depth $L$ in arbitrary dimension. We assume that the noise is unital and not unitary. Let $\mu$ denote a parameter of any $2$-qubit gate $\exp(-i\theta_\mu H_\mu)$ in the circuit such that $\norm{H_\mu}_{\infty}\le 1$. Then, we have
\begin{align}
    \Var[\partial_\mu C]\le \exp(-\Omega(|P|+L)).
\end{align}
\end{proposition}

It is noteworthy that the upper bound of Ref.~\cite{nibp} has no dependence on the Pauli weight, unlike ours. Furthermore,  the upper bound of Ref.~\cite{nibp} includes an $n^{1/2}$ factor in front of the exponential decay in $L$, making it meaningful only at depths $\Omega(\log(n))$, whereas our result is without such a factor.
Moreover, our result is more general than that shown in~\cite{nibp} because it extends to any unital noise, whereas the results shown in Ref.~\cite{nibp} apply only to primitive Pauli noise, which is only a particular type of unital noise (e.g., dephasing is not included in this class).

\subsubsection{Exponential concentration and lack thereof in noisy kernels} 
\noindent 
As a complementary result, we demonstrate how quantum kernels~\cite{thanasilp2022exponential,Jerbi_2023,Schuld_2019} exhibit a similar behavior to cost functions under the influence of non-unital noise. In particular, we show that fidelity quantum kernels~\cite{thanasilp2022exponential,Jerbi_2023,Schuld_2019} can incur an exponential concentration at any depth.
%, whereas projected quantum kernels~\cite{Huang_2021,thanasilp2022exponential} never exhibit exponential concentration. This stark discrepancy is analogous to that witnessed for global cost functions and local cost functions.
Assuming that the noise parameters satisfy $\|\bold{t}\|_2^2 + \|\bold{D}\|_2^2 < 1$, we show that
\begin{equation}
\mathbb{E}_{\Phi, \Phi'}[\Tr[\Phi(\rho) \Phi'(\rho)] ]\leq \mathrm{exp}(-\Omega(n)),
\end{equation}
for any two noisy quantum circuits $\Phi$ and $\Phi'$, where the expectation is taken over the random gates in both the circuits.
This result follows by upper bounding the overlap of two states in terms of their purities by the Cauchy-Schwarz inequality, expanding the expected purities in the Pauli basis, and upper bounding the contribution of each Pauli by means of Theorem~\ref{th:varianceMAIN}.

Furthermore, we are able to show worst-case concentration bounds by introducing an additional assumption on the noise model. In particular, we assume that the noise channel $\widetilde{\mathcal{N}}$ is 
the composition of the depolarizing channel ${\mathcal{N}_p^{(\mathrm{dep})}}$ with constant $p>0$ and an arbitrary noise channel with $\norm{\bold{t}}_2 < 1$, i.e., $\widetilde{\mathcal{N}} := \mathcal{N} \circ {\mathcal{N}_p^{(\mathrm{dep})}}$.
Thus, we obtain that 
\begin{equation}
\Tr[\Phi(\rho) \Phi'(\rho)] \leq 2^{ n(\delta_L -1)},
\end{equation}
where $\delta_L\coloneqq (1-p)^{2L} + \|\bold{t}\|_2\frac{1 - (1-p)^{2L}}{2p - p^2}$.
%Note that the term $\delta_L$ converges exponentially fast to $\|\bold{t}\|_2/({2p - p^2})$, and thus in this regime the bound is non-trivial if $\|\bold{t}\|_2\leq 2p - p^2$. 
%
We emphasize that, when the noise is purely depolarizing, i.e., 
$\norm{\bold{t}}_2=0$, this bound predicts exponential concentration at any depth, thus improving a previous result given in Ref.\ \cite{thanasilp2022exponential}, predicting exponential concentration at linear depth.
%As in the previous case, this result can be found by upper bounding the purity of the output state.
\begin{comment}
As for the projected quantum kernels~\cite{Huang_2021,thanasilp2022exponential}, we show the result
\begin{align}
    \Var_{\Phi, \Phi'}\!\left[\sum_{k = 1}^n \left[\|\Phi_k(\rho)-\Phi'_k(\rho)\|_2^2 \right]\right] \geq \Omega  (n\|\bold{t}\|^4_2), 
\end{align}
where $\Phi_k(\rho) = \Tr_{[n]\setminus \{k\}}\Phi(\rho) $.
We remark that the analysis of projected quantum kernels involves some additional assumptions on the families of random circuits, which are formally defined in Section~\ref{sec:kernels}.
\end{comment}

\section{Conclusions and outlook}
\subsection{Summary}

Our work advances our understanding of how quantum circuits---specifically, random quantum circuits---behave under possibly non-unital noise on several fronts, developing a comprehensive picture of the impact of noise. First, we have shown what we called the effective depth picture, establishing that most quantum circuits under noise behave like shallow logarithmic-depth circuits, no matter how deep they are. 
\textcolor{black}{Importantly, this result is derived without imposing any specific assumptions on the noise model, except that it is incoherent and local. As a result, the findings apply to various types of physical noise, including dephasing noise, which may have implications for phenomena such as measurement-induced phase transitions~\cite{measind1,Skinner_2019}.}

We have also shown that, in contrast to the noiseless and unital cases, variational quantum algorithms do not suffer from barren plateaus in the presence of non-unital noise. This is a striking insight in its own right. However, it is crucial to contextualize this result and not interpret it as good news: indeed, we demonstrate that the absence of barren plateaus arises because these circuits behave similarly to shallow circuits, with only the last few layers being trainable.

We have also used this fact to come up with a simple classical simulation algorithm for estimating Pauli expectation values of quantum circuits under possibly non-unital noise, on average over the ensemble. Furthermore, we have shown that the algorithm is efficient for inverse-polynomial precision even for arbitrarily deep circuits in one-dimensional architectures, and for constant precision in any spatial dimensions. 

In short, a good way of summarizing our results is that we show that, unless we engineer circuits to exploit the structure of the noise affecting the quantum device, as long as we are interested in using a noisy quantum device to estimate observable expectation values, it is unlikely that a quantum computer affected by local depolarizing noise is preferable over one with non-unital noise.  
More specifically, this arises from the fact that quantum circuits subjected to depolarizing noise, without access to fresh auxiliary qubits, are meaningful only if their depth is at most logarithmic in the number of qubits~\cite{aharonov1996limitations,DanielPaper}. In contrast, with non-unital noise, it is in principle possible to design circuits that leverage the noise to surpass this logarithmic depth barrier, as demonstrated by the quantum refrigerator construction~\cite{refrigerator}. However, our results indicate that for most circuits affected by non-unital noise, the effective depth remains logarithmic, suggesting that this advantage does not hold for such generic circuits.
In this sense, both depolarizing and non-unital noise impose qualitatively similar logarithmic-depth limitations for typical circuits in expectation-value estimation.
Our results can also be seen as an invitation to 
more carefully investigate the underlying noise to draw precise conclusions about the anticipated quantum application, as assuming that all noise is depolarizing would both lead to the wrong conclusions and be the wrong model for many physical platforms.
%and that overly simplified models can be misleading.

\subsection{Open problems}
\noindent
Our results invite a number of follow-up research questions. 
%\textcolor{black}{Firstly, although the classical algorithm we developed in Section~\ref{sub:classicalMAIN} for estimating expectation values is efficient for constant precision across any architecture, it becomes less efficient when higher precision is required. Specifically, for a desired precision of $\mathrm{poly}(n^{-1})$, the algorithm operates in polynomial time for 1-D architectures and quasipolynomial time for 2-D architectures. Notably, other classical algorithms~\cite{Bravyi_2021} also exhibit quasipolynomial runtimes. A key open question is whether it is possible to design a polynomial-time algorithm that estimates local observable expectation values within inverse-polynomial precision for random circuits with non-unital noise in higher-dimensional architectures. Additionally, it is worth exploring whether our algorithm can be improved by incorporating more sophisticated techniques, such as tensor network methods.

\subsubsection{Tighter bounds, more general noise models and mid-circuit measurements}
From a more technical perspective, it is important to understand if the trace distance bound we have shown in Eq.~\eqref{eq:trmain} can be proven assuming smaller depth than linear in the number of qubits (for example, logarithmic). Additionally, it remains an open question whether the early break condition in Eq.~\eqref{eq:SDP_1MA} is met with high probability over the ensemble.

Another intriguing question is whether we can relax the assumption that the distribution of each two-qubit gate forms a two-design, potentially replacing it with a less stringent condition. Our numerical simulations suggest that this may be the case. \textcolor{black}{Another interesting direction for future work is to extend our results to more general and complex noise models; for example, to arbitrary two-qubit noise models for the quantum gates, in order to model spatially-correlated noise, with correlations decaying in distance, and to non-Markovian noise models, in order to model temporally-correlated noise, with correlations decaying in time.}
Beyond these technical refinements, it would be interesting to investigate whether our techniques can provide insights into other phenomena, such as measurement-induced phase transitions~\cite{measind1, Skinner_2019}.  

Moreover, on a more qualitative note, although our work reveals the negative result that for most quantum circuits affected by any type of noise, only the last few layers significantly contribute to the computation, this outcome could potentially offer guidance on constructing circuits that are more resilient to this effect by considering how to challenge our assumptions. For instance, a critical assumption in our analysis is that each layer operates independently of the others. This assumption, however, no longer holds, for example, once mid-circuit measurements and classical feedforward are introduced.

Furthermore, our work leaves open the following broader problem regarding random circuits with non-unital noise.

\subsubsection{Sampling from random quantum circuits under non-unital noise}
\noindent
Our analysis establishes rigorous structural limitations of noisy random circuits and their implications for expectation-value estimation and trainability. 
At the same time, the problem of classical sampling from such circuits remains an important and subtle open direction. 
It is worth emphasizing, however, that sampling is not required for the applications we target, such as diagnosing phases of matter in quantum simulation and condensed matter physics or evaluating cost functions in variational quantum algorithms. 
By focusing on expectation values, we provide rigorous guarantees for the figures of merit most relevant to these domains, while also laying the groundwork for future progress on the sampling problem.

Previous work has shown that efficient sampling is possible in certain noisy settings. 
For example, Ref.~\cite{Aharonov_2023} presented an algorithm to efficiently sample from random quantum circuits of $\Theta(\log n)$ depth under local depolarizing noise. 
However, these techniques do not carry over to the non-unital setting. 
The key obstacle is that circuits with non-unital noise do not anticoncentrate: their output distributions never become sufficiently uniform~\cite{fefferman2023effect}. 
Moreover, unlike depolarizing or other unital noise sources~\cite{aharonov1996polynomial,DanielPaper,Deshpande_2022}, non-unital noisy circuits do not converge to the maximally mixed state. 
As a result, the complexity of sampling from random circuits under non-unital noise remains unresolved.

Our work brings a new perspective that may help address this fundamental question. 
Specifically, we establish that random quantum circuits with non-unital noise behave effectively as logarithmic-depth noisy circuits, in the sense that only the last $O(\log n)$ layers significantly affect the output. 
In parallel, Ref.~\cite{fefferman2023effect} has shown that the outcome distribution of such circuits is never too flat, but instead remains peaked. 
Taken together, these findings reveal that random circuits under non-unital noise are both shallow and peaked, imposing considerable structure on the problem.

Interestingly, recent work has demonstrated efficient classical sampling for two-dimensional, constant-depth shallow circuits with peaked output distributions~\cite{bravyi2023classical}. 
While the definition of “peakedness” in Ref.~\cite{bravyi2023classical} is stronger than that in Ref.~\cite{fefferman2023effect}, and while their notion of shallowness refers to constant depth rather than logarithmic depth, the structural parallels are suggestive. 
Moreover, their algorithm establishes hardness in the worst case, whereas our setting requires only an average-case guarantee over the ensemble of random circuits, which could potentially simplify the problem. 

We leave it to future work to rigorously determine whether the structural insights developed here can be leveraged to design an efficient classical sampling algorithm, or conversely, to establish hardness results for this task. 
In either case, our findings provide a new lens through which to analyze the complexity of noisy random quantum circuits sampling with non-unital noise.

\section{Acknowledgements}
\noindent AAM thanks Lennart Bittel for insightful discussions on the partial derivative lower bound; AAM and AA thank Marco Cerezo and Samson Wang for helpful discussions on their previous work~\cite{nibp}; and Marco Cerezo, together with Manuel S. Rudolph and Zoe Holmes, for invaluable input on the improved classical simulation algorithm included in this revised manuscript. AA thanks Giacomo De Palma and Isabella Uta Meyer for discussions about channel distinguishability.
SG thanks Bill Fefferman, Kunal Sharma, Michael Gullans, and Kohdai Kuroiwa for insightful discussions regarding the Heisenberg picture and optimization problems under noise.
DSF acknowledges funding from the European Union under Grant Agreement 101080142 and the project EQUALITY and financial support from the Novo Nordisk Foundation (Grant No. NNF20OC0059939 Quantum for Life). The Berlin team has been supported by the BMFTR (DAQC, MUNIQC-Atoms, QuSol, Hybrid++), the Munich Quantum Valley (K-4 and K-8), the Quantum Flagship (PasQuans2, Millenion), QuantERA (HQCC), the Clusters of Excellence MATH+ and ML4Q, the DFG (CRC 183), the Einstein Foundation (Einstein Research Unit on Quantum Devices), Berlin Quantum, and the ERC (DebuQC).
AA acknowledges financial support from the QICS (Quantum Information Center Sorbonne). This work is supported by a collaboration between the US DOE and other Agencies. YQ 
acknowledges financial support from the U.S. Department of Energy, Office of Science, National Quantum Information Science Research Centers, Quantum Systems Accelerator. This work was done (in part) while YQ and SG were visiting the Simons Institute for the Theory of Computing.

\section*{Author Contributions}
The following describes the different contributions of all authors of this work, using roles defined by the CRediT (Contributor Roles Taxonomy) project: 
A.A.M.\ Conceptualization, formal analysis, methodology, software, writing (original draft), writing (review and editing); 
A.A.\ Conceptualization, formal analysis, methodology, writing (original draft), writing (review and editing); 
S.G.\ Conceptualization, methodology, writing (review and editing); 
S.K.\ Conceptualization, methodology, writing (review and editing); 
D.S.F.\ Conceptualization, formal analysis, methodology, writing (original draft), writing (review and editing); 
J.E.\ Conceptualization, methodology, writing (review and editing); 
Y.Q.\ Conceptualization, formal analysis, methodology, writing (original draft), writing (review and editing).

\section*{Competing Interests}
The authors declare no competing interests.

\bibliography{ref}

\clearpage

\onecolumngrid
\appendix
\begin{center}
	%\noindent\textbf{Supplementary Material}
	%\bigskip
		
	\noindent\textbf{\large{Supplementary Material}}
\end{center}

%\onecolumngrid
%\begin{center}
%\vspace*{\baselineskip}
%{\textbf{\large Supplemental Material
%}}\\
%\end{center}

%\renewcommand{\theequation}{S\arabic{equation}}
%\setcounter{equation}{0}
%\numberwithin{equation}{section}

\tableofcontents
\newpage

\section{Preliminaries}
\subsection{Notation and basic definitions}\label{sec-definitions}
\noindent Throughout this work, we employ the following notation and conventions. 

\begin{itemize}
\item $\mathcal{L}(\mathbb{C}^d)$ denotes the set of linear operators that act on the $d$-dimensional complex vector space $\mathbb{C}^d$. 
\item The \emph{unitary group}, denoted as $\Ug(d)$, comprises operators $U \in \mathcal{L}(\mathbb{C}^d)$ satisfying $U^\dagger U = I$, with $I$ representing the identity operator. Additionally, we use $[d]$ to represent the set of integers from $1$ to $d$, i.e., $[d] \coloneqq \{1,\dots, d\}$.

\item Given a vector $v \in \mathbb{C}^d$ and a value $p \in \left[1,\infty\right]$, we denote the $p$-norm of $v$ as $\norm{v}_p$, defined as 
\begin{equation}\norm{v}_p \coloneqq \left(\sum_{i=1}^d |v_i|^p\right)^{1/p}.
\end{equation}

\item \textbf{Norms:} For a matrix $A\in \MatC{d}$, its 
\emph{Schatten $p$-norm} is $\norm{A}_p\coloneqq \Tr\small(\small(\sqrt{A^\dagger A}\small)^p\small)^{1/p}$, corresponding to the $p$-norm of the vector of singular values of $A$. 

The trace norm and Hilbert-Schmidt norm, specific instances of Schatten $p$-norms, are, respectively, denoted as $\norm{\cdot}_1$ and $\norm{\cdot}_2$. 

The infinity norm, $\norm{\cdot}_\infty$, of a matrix is the maximum singular value, which is equal to the limit of the Schatten $p$-norm as $p$ approaches infinity. 

The Hilbert-Schmidt norm arises from the scalar product $\hs{A}{B}\coloneqq \Tr\!\left(A^\dagger B\right)$ for $A,B \in \mathcal{L}\left(\mathbb{C}^d\right)$. 

The \emph{H\"older inequality}, 
$\lvert \hs{A}{B} \rvert \le \norm{A}_p \norm{B}_q, $ holds for $1 \le p,q \le \infty $ such that $p^{-1}+q^{-1}=1$.

For all matrices $A$ and $1\le p\le q$, we have $\norm{A}_{q}\le \norm{A}_{p}$ and $\norm{A}_{p}\le \mathrm{rank}(A)^{(p^{-1}-q^{-1})}\norm{A}_{q}$. In particular, we have that $\norm{A}_1\leq\sqrt{d}\norm{A}_2$.

\item \textbf{Quantum states:} The set of density matrices (\emph{quantum states}) is 
\begin{equation}
    \mathcal{S}\!\left(\mathbb{C}^d\right)\coloneqq \{\rho \in \mathcal{L}\left(\mathbb{C}^d\right) \,:\,\rho \ge 0,\,\Tr(\rho)=1\}.
    \end{equation}
We adopt the bra-ket notation, denoting a vector $v \in \mathbb{C}^d$ as $\ket{v}$ and its adjoint as $\bra{v}$. A vector $\ket{\psi}\in \mathbb{C}^d$ is a (pure) state vector if $\norm{\ket{\psi}}_2=1$. The canonical basis of $\mathbb{C}^d$ is $\{\ket{i}\}_{i=1}^{d}$, and the non-normalized maximally entangled state vector is given by $\ket{\Omega}\coloneqq \sum^d_{i=1}\ket{i}\otimes\ket{i}=\sum^d_{i=1}\ket{i,i}.$

\item Given an element of $A\in \mathcal{L}(\mathbb{C}^{d_1}\otimes \mathbb{C}^{d_2})$, we indicate with $\Tr_1(A)$ the partial trace of $A$ with respect to the first subsystem, and similarly for $\Tr_2(A)$, for partial trace of $A$ with respect to the second subsystem. 

\item  When addressing a system of $n$ qubits, we use $I$ to denote the identity operator on the Hilbert space $\mathbb{C}^{2}$ of one qubit, while $I_n=I^{\otimes n}$ denotes the identity on the Hilbert space of $n$ qubits.

\item \textbf{Pauli basis:} Let $d=2^n$, where $n \in \mathbb{N}$. Elements  of the Pauli basis $\{I,X,Y,Z\}^{\otimes n}$ are Hermitian, unitary, trace-less, they square to the identity and they are orthogonal to each other with respect to the Hilbert-Schmidt Scalar product. The Pauli basis forms an orthogonal basis for the linear operators $\MatC{d}$. We denote $\boldsymbol{\sigma}\coloneqq (X,Y,Z)$ the vector of single qubit Pauli matrices.

Given $P\in \{I,X,Y,Z\}^{\otimes n}$ such that $P=Q_1 ~\otimes \cdots \otimes ~Q_n$, we define $\left[P\right]$ the set $\left[P\right]\coloneqq \{Q_1,\dotsc, Q_n\}$. 

\item \textbf{Support:} We define the \emph{support of} $P\in \{I,X,Y,Z\}^{\otimes n}$ as the set of integers containing the non-identity terms in $\left[P\right]$, i.e., \begin{equation}
\mathrm{supp}(P)\coloneqq \{i \in [n]\, :\, Q_i\neq I\}.\end{equation}
For example, $\mathrm{supp}(X\otimes I\otimes Y)=\{1,3\}$.

Similarly, if $H$ is an operator expressed in the Pauli basis as $H=\sum^{M}_{i=1} c_i P_i$, where $\{c_i\}^{M}_{i=1}$ are real non-zero numbers, and $\{P_i\}^{M}_{i=1}$ are elements of the Pauli basis, then
\begin{align}
    \mathrm{supp}(H)\coloneqq\bigcup^{M}_{i=1} \supp(P_i).
\end{align}
\item \textbf{Pauli weight:} The \emph{Pauli-weight} of $P$, denoted as $|P|$, is the number of Pauli in the tensor product decomposition of $P$ different from $I$, namely $|P|\coloneqq |\mathrm{supp}(P)|$. For example $|X\otimes I\otimes Y|=|\{1,3\}|=2$. The Pauli-$X$ weight of $P$, denoted as $|P|_X$, corresponds to the number of Pauli-$X$ operators in the tensor-product decomposition of $P$. Similarly, the Pauli-$Y$ weight of $P$, denoted as $|P|_Y$, represents the number of Pauli-$Y$ operators in the decomposition of $P$. Analogously, the Pauli-$Z$ weight is denoted as $|P|_Z$.

\item \textbf{Locality of an observable:} The \emph{locality} of a Hermitian operator $H$ is defined as $|\supp(H)|$.

\item \textbf{Light cone:} The light-cone of a Hermitian operator $H$ with respect to a linear map $\Phi$ is defined as 
    \begin{align}
       \mathrm{Light}(\Phi,H) \coloneqq  \supp
        (\Phi^*(H)),
    \end{align}
    where $\Phi^*$ denotes the adjoint of $\Phi$ with respect to the Hilbert-Schmidt scalar product.
    
\item \textbf{Circuit layer:} \textcolor{black}{A unitary operator is said to form a circuit layer if it can be expressed as a tensor product of two-qubit gates (not necessarily nearest neighbors).}

\item \textbf{Geometrical locality:} \textcolor{black}{A circuit \(U = U_L U_{L-1} \dots U_1\) is said to have geometrical locality \(\mathrm{D} > 0\) if, for any observable \(H\) and unitary layers \(U_1, \dots, U_L\), the Heisenberg-evolved observable \(U_{j+1}^\dag U_{j+2}^\dag \dots U^\dag_L H U_L \dots U_{j+2} U_{j+1}\) is supported on at most \(|\mathrm{supp}(H)|(2m)^{\mathrm{D}}\) qubits, where \(m \coloneqq L-j\). This definition aligns and it is implied by other common definitions of geometric locality. A circuit without any assumptions on geometric locality is said to have all-to-all connectivity, in which case the Heisenberg-evolved observable \(U_{j+1}^\dag U_{j+2}^\dag \dots U^\dag_L H U_L \dots U_{j+2} U_{j+1}\) is supported on at most \(|\mathrm{supp}(H)| 2^{m}\) qubits.}

\item \textbf{Asymptotic notation:} Big-O Notation: For a function \(f(n)\), if there exists a constant \(c\) and a specific input size \(n_0\) such that $f(n) \leq c \cdot g(n)$ for all \(n \geq n_0\), where \(g(n)\) is a well-defined function, then we express it as \(f(n) = O(g(n))\). This notation signifies the upper limit of how fast a function grows in relation to \(g(n)\).

Big-Omega Notation: For a function \(f(n)\), if there exists a constant \(c\) and a specific input size \(n_0\) such that $f(n) \geq c \cdot g(n)$ for all \(n \geq n_0\), where \(g(n)\) is a well-defined function, then we express it as \(f(n) = \Omega(g(n))\). This notation signifies the lower limit of how fast a function grows in relation to \(g(n)\).

Big-Theta Notation: For a function \(f(n)\), if $f(n) = O(g(n))$ and if $f(n) = \Omega(g(n))$, where \(g(n)\) is a well-defined function, then we express it as \(f(n) = \Theta(g(n))\). 

Little-Omega Notation: For a function \(f(n)\), if for any constant $c$, there exists an input size $n_0$ such that $f(n) > c \cdot g(n)$ for all \(n \geq n_0\), where \(g(n)\) is a well-defined function, then then we express it as \(f(n) = \omega(g(n))\). 
This notation implies that the function grows strictly faster than the provided lower bound.

\end{itemize}

\subsection{Haar measure and unitary designs}
\label{unitary design and haar measure}
\noindent In the following, for our proofs it will be useful to have some familiarity with the Haar measure, which formalizes the notion of uniform distribution over unitaries. 
For a more detailed explanation we refer to
Ref.~\cite{MeleHaar} -- here we state a few crucial properties that will be useful later on. We define as the \emph{Haar measure} $\mu_H(\Ug(d))$ the (unique) probability distribution over the unitary group $\Ug(d)$ which is left and right invariant, which means that for any integrable function $f$, we have
\begin{align}
    \ExU \left[f(U)\right]= \ExU \left[f(UV)\right]= \ExU \left[f(VU)\right],
\end{align}
for any $U,V \in \Ug(d)$. In the last equation, we have used $\mu_H\equiv \mu_H(\Ug(d))$, that is, we omitted to specify the group $\Ug(d)$, and we will do the same in subsequent sections.%, except in some cases where it will be necessary to specify the group for the sake of clarity.
 Moreover, it holds that
\begin{align}
    \ExU \left[f(U)\right]= \ExU \left[f(U^\dagger)\right].
\end{align}

\subsubsection{Useful relations}
\noindent We define the \emph{identity} $\mathbb{I}$ and the \emph{flip operator} $\mathbb{F}$, also known as the permutation operators associated to a tensor product of two Hilbert spaces, as 
\begin{align}
    \Idd\coloneqq \sum^d_{i,j=1} \ketbra{i,j}{i,j}, \quad\quad\quad \Flip \coloneqq \sum^d_{i,j=1} \ketbra{i,j}{j,i}.
\end{align}
From this definition, it can be observed that they satisfy 
\begin{align}
\mathbb{I}\left(\ket{\psi}\otimes\ket{\phi}\right)=\ket{\psi}\otimes\ket{\phi}, \quad \quad
\mathbb{F}\left(\ket{\psi}\otimes\ket{\phi}\right)=\ket{\phi}\otimes\ket{\psi},
\end{align}
for all $\ket{\psi},\ket{\phi} \in \mathbb{C}^d$.
Useful properties of the flip operator are the \emph{swap-trick} and the \emph{partial-swap-trick}
\begin{align}
    \Tr(A\otimes B \Flip)=\Tr(AB),\quad \Tr_2(A\otimes B \Flip)=AB,
\end{align}
equalities that can be easily verified in terms of tensor network diagrams.

Let $n$ now be the number of qubits in a system. If $d=2^n$, then the flip operator can be written in terms of the Pauli basis as
\begin{align}
    \Flip=\frac{1}{d}\sum_{P\in \{I,X,Y,Z\}^{\otimes n}} P\otimes P,
\end{align}
where we have used the fact that the Pauli basis is an orthogonal basis and the \emph{swap-trick}. 

Two formulas, involving expected values over the Haar measure and permutation operators, will be crucial in our proofs.
Given $O \in \MatC{d}$, we have the so called \emph{first-moment formula}, given by
\begin{align}
            \ExU \left[U O U^{\dagger}\right]&=\frac{\Tr\!\left(O\right)}{d} I.
            \label{eq:1momHaar}
\end{align}
Given $O \in \mathcal{L}((\mathbb{C}^d)^{\otimes 2})$, we have the second-moment formula
    \begin{align}
        \ExU \left[U^{\otimes 2} O U^{\dagger \otimes 2}\right]=c_{\Idd,O}\mathbb{I} + c_{\Flip,O} \mathbb{F},
        \label{eq:2momHaar}
    \end{align}
    where
    \begin{align}
        c_{\Idd,O}=\frac{\Tr\!\left(O\right)-d^{-1}\Tr\!\left(\mathbb{F}O\right)}{d^2-1}\quad \text{and} \quad c_{\Flip,O}=\frac{\Tr\!\left(\mathbb{F}O\right)-d^{-1}\Tr\!\left(O\right)}{d^2-1},
    \end{align}
(see Ref.~\cite{MeleHaar} for a proof of the previous two equations).
A probability distribution over unitaries $\nu$ is defined to be a \emph{$k$-design}~\cite{Designs}, for $k\in\mathbb{N}$, if and only if 
    \begin{align}
        \ExU \left[U^{\otimes k} O U^{\dagger \otimes k}\right]= \underset{V \sim \nu}{\mathbb{E}} \left[V^{\otimes k} O V^{\dagger \otimes k}\right].
        \label{eq:design}
    \end{align}
If the distribution $\nu$ is a $(k+1)$-design, then it is also a $k$-design. 
%
%The uniform distribution over the Pauli basis $\mathcal{P} \coloneqq \{I,X,Y,Z\}^{\otimes n}$ constitutes a $1$-design because, for any $O\in \MatC{d}$, we have
%\begin{align}\frac{1}{|\mathcal{P}|}\sum_{P\in \mathcal{P}} P O P^{\dagger}=\frac{1}{(2^{n})^2}\sum_{P\in \mathcal{P}}\Tr_2\left((I\otimes O) ( P\otimes P) \Flip\right)=\frac{1}{2^n}\Tr_2\left((I\otimes O) \Flip\Flip\right)=\frac{\Tr(O)}{2^n}=\ExU \left[U O U^{\dagger}\right],
%\end{align}
%where we have used the expression of the Flip operator in terms of the Pauli basis.
An important set of unitaries which will be useful in our work is the Clifford group~\cite{gottesman1998heisenberg}.
\begin{definition}[Clifford group~\cite{gottesman1998heisenberg,Aaronson_2004}]
\label{def:Clifford}
    The Clifford group $\mathrm{Cl}(n)$ is the set of unitaries which sends the Pauli group $\mathcal{P}_n \coloneqq  \{i^k\}^{3}_{k=0}\times \{I,X,Y,Z\}^{\otimes n}$ in itself under the adjoint operation:
\begin{align}
    \mathrm{Cl}(n)\coloneqq \{U\in \mathrm{U}(2^n)\colon UPU^{\dagger} \in \mathcal{P}_n \text{ for all $P\in  \{I,X,Y,Z\}^{\otimes n}$} \},
    \label{eq:Cliffordgroup}
\end{align}
and it is equivalent to the set of unitaries generated by $\{\Hadamard,\CNOT,\PhaseS\}$ where $\Hadamard$, $\CNOT$, and $\PhaseS$ are, respectively, the $\mathrm{Hadamard}$, $\mathrm{Controlled}$-$\mathrm{NOT}$, and $\mathrm{Phase}$ gates.
\end{definition}

We make extensive use of the following seminal result throughout our work.

\begin{lemma}[Clifford group is a $2$-design~\cite{webb2016clifford,zhu2016clifford}]
\label{le:Clifford3design}
    The uniform distribution over the Clifford group $ \mathrm{Cl}(n)$ is a $2$-design.
\end{lemma}
The Clifford group is actually also a $3$-design~\cite{webb2016clifford,zhu2016clifford}, but in our work we need only its $2$-design property. 
Now, we state an important formula -- the Pauli mixing formula -- that we use in many of the proofs.
\begin{lemma}[Pauli mixing]
Let $d=2^n$ and consider $\nu$ to be a $2$-design distribution. If $P_1,P_2\in \{I,X,Y,Z\}^{\otimes n}$ are elements of the Pauli basis, then 
\begin{align}
   \underset{U \sim \nu}{\mathbb{E}}  \left[U^{\otimes 2} (P_1\otimes P_2) U^{\dagger \otimes 2}\right]&=\delta_{P_1,P_2}\underset{U\sim\nu}{\mathbb{E}}\left[U^{\otimes 2}(P_1\otimes P_1)U^{\dagger\otimes 2}\right]
   \label{eq:paulimixing}\\
   &=\begin{cases}
    I\otimes I & \text{if } P_1= P_2 = I, \\
   \frac{1}{d^2-1} \sum_{P\in \{I,X,Y,Z\}^{\otimes n} \setminus I_n} P\otimes P & \text{if } P_1= P_2 \neq I, \\
   0 & \text{if } P_1\neq P_2. \end{cases}\label{eq:paulimixing_2}
\end{align}
\end{lemma}
This can be shown using the second-moment formula previously introduced and the decomposition of the flip operator in terms of the Pauli basis.

\subsubsection{Properties of layers of single-qubit random gates}

In our work, frequent calculations involve averages over the tensor product of single-qubit $2$-design gates. To facilitate these calculations, we introduce two Lemmas that will be instrumental in later sections of our work.
\begin{lemma}[A layer of $1$-qubit Haar random gates is a global $1$-design]
\label{le:qubitonedesign}
Let $\nu$ be a distribution over the tensor product of single-qubit $1$-design gates, namely over unitaries $U$ of the form $U=\bigotimes^{n}_{i=1} u_i$, where $u_i$ is a single-qubit unitary acting on the $i$-th qubit. Then, $\nu$ is a $n$-qubit $1$-design. 
\end{lemma}
\begin{proof}
Let $O\in\MatC{d}$ with $d=2^n$. By writing it in the Pauli basis, we have
\begin{align}
    \underset{U \sim \nu}{\mathbb{E}} \left[U O U^{\dagger}\right]=\frac{1}{d}\sum_{P\in \{I,X,Y,Z\}^{\otimes n}} \Tr(OP) \underset{U \sim \nu}{\mathbb{E}}\left[U P U^{\dagger}\right]=\frac{\Tr(O)}{d}I_n,
\end{align}
where in the last equality we have used the first-moment formula (Eq.~\eqref{eq:1momHaar}) on each of the qubits and used that the Pauli are trace-less.
\end{proof}
\noindent As a consequence of the Pauli mixing property, we have the following lemma.

\begin{lemma}[Second moments of single-qubit random gates layers]
\label{le:mixham}
Let $\nu$ be a distribution over the tensor product of single-qubit $2$-design gates, namely over unitaries $U$ of the form $U=\bigotimes^{n}_{i=1} u_i$, where $u_i$ is a single-qubit unitary acting on the $i$-th qubit. Let $B$ be any operator. Then we have

\begin{enumerate}
    \item Let $O \coloneqq \sum_{P\in \{I,X,Y,Z\}^{\otimes n}} a_P P$, with $a_P\in \mathbb{R}$ for any $P\in \{I,X,Y,Z\}^{\otimes n}$. We have 
    \begin{align}
       \underset{U \sim \nu}{\mathbb{E}}\left[\Tr(O UB U^{\dagger} )^2\right] = \sum_{P\in \{I,X,Y,Z\}^{\otimes n}} a_P^2  \underset{U \sim \nu}{\mathbb{E}}\left[\Tr(P UB U^{\dagger} )^2\right].
       \end{align}
    \item For any $P\in \{I,X,Y,Z\}^{\otimes n}$, we have 
\begin{align}
    \underset{U \sim \nu}{\mathbb{E}}\left[\Tr(P UB U^{\dagger})^2\right]=\frac{1}{3^{|P|}}\sum_{\substack{Q \in \{I,X,Y,Z\}^{\otimes n}:\\ \supp(Q)=\supp(P) }}\Tr(Q B)^2.
\end{align}
\end{enumerate}

\end{lemma}
\begin{proof}
We have
\begin{align}
     \underset{U \sim \nu}{\mathbb{E}}[\Tr(O UB U^{\dagger} )^2]&= \underset{U \sim \nu}{\mathbb{E}}[\Tr(O^{\otimes 2} U^{\otimes 2}B^{\otimes 2} U^{\dagger\otimes 2}  )]\\
     \nonumber
    &= \sum_{P,Q\in \{I,X,Y,Z\}^{\otimes n}} a_P a_Q \underset{U \sim \nu}{\mathbb{E}}[\Tr((P\otimes Q) U^{\otimes 2}B^{\otimes 2} U^{\dagger\otimes 2})]\\
     \nonumber
    &= \sum_{P,Q\in \{I,X,Y,Z\}^{\otimes n}} a_P a_Q  \underset{U \sim \nu}{\mathbb{E}}[\Tr(U^{\dagger\otimes 2}(P\otimes Q) U^{\otimes 2}B^{\otimes 2} )]\\
     \nonumber
    &= \sum_{P\in \{I,X,Y,Z\}^{\otimes n}} a_P^2  \underset{U \sim \nu}{\mathbb{E}}[\Tr(U^{\dagger\otimes 2}(P\otimes P) U^{\otimes 2}B^{\otimes 2} )]\\
     \nonumber
    &=\sum_{P\in \{I,X,Y,Z\}^{\otimes n}} a_P^2 \underset{U \sim \nu}{\mathbb{E}}[\Tr(P UB U^{\dagger})^2],
     \nonumber
\end{align}
where for the fourth equality we have used the fact that $U=\bigotimes^{n}_{i=1} u_i $ is a layer of single-qubit $2$-design unitaries, 
\begin{equation}
\ExU \left[f(U)\right]= \ExU \left[f(U^\dagger)\right]
\end{equation}
for any measurable function $f$, and the Pauli mixing property in Eq.~\eqref{eq:paulimixing} for each of the single-qubit unitaries to conclude that $\underset{u_i \sim \mu_H}{\mathbb{E}} [u_i^{\otimes 2} (P_1\otimes P_2) u_i^{\dagger \otimes 2}]=0$ for two different single-qubit Pauli $P_1$ and $P_2$.

Similarly, for $P\in\{I,X,Y,Z\}^{\otimes n}$ such that $P=P_1\otimes P_2\otimes\dotsb\otimes P_n$, we use the Pauli-mixing property in Eq.~\eqref{eq:paulimixing_2}, along with the fact that $U=\bigotimes_{i=1}^n u_i$ is a tensor product of single-qubit unitaries from a 2-design, to obtain
\begin{align}
    \underset{U \sim \nu}{\mathbb{E}}[\Tr(P UB U^{\dagger})^2]&=\underset{U \sim \nu}{\mathbb{E}}[\Tr(U^{\dagger\otimes 2} (P\otimes P) U^{\otimes 2}B^{\otimes 2})]\\
\nonumber&=\Tr\!\left(\underset{U\sim\nu}{\mathbb{E}}\left[\bigotimes_{i=1}^n u_i^{\otimes 2}(P_i\otimes P_i)u_i^{\dagger\otimes 2}\right]B^{\otimes 2}\right)\\
\nonumber    &=\Tr\!\left(\underset{U\sim\nu}{\mathbb{E}}\left[\bigotimes_{i\in\supp(P)} u_i^{\otimes 2}(P_i\otimes P_i)u_i^{\dagger\otimes 2}\right]B^{\otimes 2}\right)\\
    \nonumber&=\Tr\!\left(\left(\bigotimes_{i\in\supp(P)}\frac{1}{3}\sum_{Q_i\in\{X,Y,Z\}} Q_i\otimes Q_i \right)B^{\otimes 2}\right)\\
    \nonumber&=\frac{1}{3^{|P|}}\sum_{\substack{Q \in \{I,X,Y,Z\}^{\otimes n}:\\ \supp(Q)=\supp(P) }}\Tr(Q^{\otimes 2} B^{\otimes 2}),
    \nonumber
\end{align}
%where we have used that $U$ is a layer of single-qubit $2$-design gates $U=\otimes^{n}_{i=1} u_i $ and the Pauli mixing property in Eq.\eqref{eq:paulimixing_2} for each of the single-qubit gates to conclude that for $P\in \{I,X,Y,Z\}$:
%\begin{align}
%     \underset{u_i \sim \mu_H}{\mathbb{E}} [u_i^{\otimes 2} (P\otimes P) u_i^{\dagger \otimes 2}]=\frac{1}{3}\sum_{Q\in \{X,Y,Z\}} Q\otimes Q.
%\end{align}
which is the desired result, because $\Tr(Q^{\otimes 2} B^{\otimes 2})=\Tr(Q B)^2$.
\end{proof}

\subsection{Quantum channels}
\label{quantumchannels}
\noindent A \emph{quantum channel} $\mathcal{N}:\mathcal{L}(\mathbb{C}^d)\rightarrow\mathcal{L}(\mathbb{C}^d)$ is a linear, completely positive, and trace-preserving map. Completely positive means that for all positive operators $\sigma \in \mathcal{L}(\mathbb{C}^d\otimes \mathbb{C}^D)$, for any $D\in \mathbb{N}$, the operator $(\mathcal{N} \otimes \mathcal{I})(\sigma)$ is positive. The trace-preserving property means that $\Tr(\mathcal{N}(A))=\Tr(A)$ for any $A\in\MatC{d}$. Here, $\mathcal{I}:\MatC{D}\rightarrow \MatC{D}$ denotes the identity map. Any quantum channel $\mathcal{N}$ can be represented in terms of at most $d^2$ Kraus operators $\{K_i\}^{d^2}_{i=1}$, i.e., 
\begin{equation}\mathcal{N}\left(\cdot\right)=\sum^{d^2}_{i=1} K_i \left(\cdot\right)K^\dagger_i
, 
\end{equation}
with the condition $\sum^{d^2}_{i=1} K^\dagger_i K_i= I$ to satisfy trace-preservation.
Given a quantum channel $\mathcal{N}$, we say that $\mathcal{N}$ is \emph{unital} if and only if it maps the identity operator to the identity operator, i.e., $\mathcal{N}(I)=I$.
Otherwise, we say that $\mathcal{N}$ is 
\emph{non-unital}. 
%\je{Throughout the draft, it is probably important that non-unital noise is no detail, but very common for realistic noise models.}
Given a quantum channel $\mathcal{N}:\mathcal{L}(\mathbb{C}^d)\rightarrow\mathcal{L}(\mathbb{C}^d)$, its adjoint map $\mathcal{N}^{\ast}:\mathcal{L}(\mathbb{C}^d)\rightarrow\mathcal{L}(\mathbb{C}^d)$ is defined as the linear map such that 
\begin{equation}
\hs{\mathcal{N}^{\ast}(A)}{B}=\hs{A}{\mathcal{N}(B)}
\end{equation}
for any $A,B \in \MatC{d}$. If $\{K_i\}^{d^2}_{i=1}$ is a set of Kraus operators for $\mathcal{N}$, then the adjoint channel $\mathcal{N}^*$ can be expressed as 
\begin{equation}
\mathcal{N}^*(\cdot)=\sum^{d^2}_{i=1} K^\dagger_i (\cdot) K_i.
\end{equation}
Note that $\mathcal{N}^*$ is always unital, $\mathcal{N}^*(I)=I$, inherited from the property of the channel being trace-preserving.
However the adjoint is not necessarily trace-preserving: it holds that $\mathcal{N}^{*}$ is trace preserving if and only if $\mathcal{N}$ is unital.
If the Kraus operators of the quantum channel $\mathcal{N}$ are Hermitian, then the adjoint channel coincides with the quantum channel $\mathcal{N}^*=\mathcal{N}$.
If $\mathcal{N}_1$, $\mathcal{N}_2$ are two quantum channels, then $(\mathcal{N}_1 \circ \mathcal{N}_2)^*=\mathcal{N}_2^* \circ \mathcal{N}_1^*$.
Moreover $(\mathcal{N}_1 \otimes \mathcal{N}_2)^*=\mathcal{N}_1^* \otimes \mathcal{N}_2^*$.
For any Hermitian operator $O$, we have~\cite{Bhatia2007PositiveDM}
\begin{equation}
    \norm{\mathcal{N}^{*}(O)}_\infty \le \norm{O}_\infty.
\end{equation}

\subsubsection{Pauli transfer matrix representation of a quantum channel}

\noindent In this subsection, we introduce the Pauli transfer matrix representation of a single-qubit quantum channel.
Let $\mathcal{N}:\mathcal{L}(\mathbb{C}^2)\rightarrow\mathcal{L}(\mathbb{C}^2)$ be a linear map.
Any linear map can be expressed in terms of its action on the Pauli basis, i.e.
\begin{align}
    \mathcal{N}(P)=\sum_{Q\in \{I,X,Y,Z\}} T_{Q,P} Q,
\end{align}
where $T_{Q,P} \coloneqq \frac{1}{2}\Tr(Q\mathcal{N}(P))$. 
Assuming that $\mathcal{N}$ represents a quantum channel, it inherently preserves Hermiticity, implying that $T_{Q,P}\in \mathbb{R}$. Furthermore, by employing the \emph{Hölder inequality}, we establish that $|T_{Q,P}|\le 1$ for all $P,Q\in\{I,X,Y,Z\}$.
Given that a quantum channel is trace-preserving and the Pauli matrices are all trace-less except for the identity, we deduce that $T_{I,P}=\delta_{I,P}$. Consequently, we have
\begin{align}
    \mathcal{N}(I)&= I + T_{X,I}X + T_{Y,I}Y + T_{Z,I}Z,\label{eq:noisedefPauli}\\
    \mathcal{N}(X)&= T_{X,X}X + T_{Y,X}Y + T_{Z,X}Z,\label{eq:noisedefPauli_X}\\
    \mathcal{N}(Y)&= T_{X,Y}X + T_{Y,Y}Y + T_{Z,Y}Z,\label{eq:noisedefPauli_Y}\\
    \mathcal{N}(Z)&= T_{X,Z}X + T_{Y,Z}Y + T_{Z,Z}Z.\label{eq:noisedefPauli_Z}
\end{align}
From Eqs.~\eqref{eq:noisedefPauli}-\eqref{eq:noisedefPauli_Z}, we can see that considering a non-unital noise channel is equivalent to assuming that at least one of the parameters $T_{X,I}$, $T_{Y,I}$, or $T_{Z,I}$ must be non-zero.

We define the Pauli transfer matrix $\mathrm{T}(\mathcal{N})$ of the channel $\mathcal{N}$ as the matrix with components defined as $[\mathrm{T}(\mathcal{N})]_{Q,P} \coloneqq \frac{1}{2}\Tr(Q\mathcal{N}(P))=T_{Q,P}$ for all $Q,P \in \{I,X,Y,Z\}$, i.e.,
\begin{equation}\label{eq:PTM_PTM}
\mathrm{T}(\mathcal{N})=\begin{bmatrix}
1 & 0 & 0 & 0 \\
T_{X,I} & T_{X,X} & T_{X,Y} & T_{X,Z} \\
T_{Y,I} & T_{Y,X} & T_{Y,Y} & T_{Y,Z} \\
T_{Z,I} & T_{Z,X} & T_{Z,Y} & T_{Z,Z} \\
\end{bmatrix}.
\end{equation}
It is important to note that any single-qubit quantum channel can be expressed in such a form. However, not every linear map of this form represents a valid quantum channel.
Furthermore, utilizing the definition of the adjoint map, we can easily verify that the adjoint map $\mathcal{N}^{*}$ is given by $ \mathcal{N}^{*}(P)=\sum_{Q\in \{I,X,Y,Z\}} T_{P,Q} Q$.
This results in the fact that the Pauli transfer matrix of the adjoint channel is the transpose of the Pauli transfer matrix of the channel, i.e.,
\begin{align}
    \mathrm{T}(\mathcal{N}^{*})=\mathrm{T}(\mathcal{N})^T.
\end{align}

Given two quantum channels $\mathcal{N}^{(A)}$ and $\mathcal{N}^{(B)}$, we have that the Pauli transfer matrix associated to their composition is given by the multiplication of their Pauli transfer matrices:
\begin{align}
    \mathrm{T}(\mathcal{N}^{(A)}\circ\mathcal{N}^{(B)})=\mathrm{T}(\mathcal{N}^{(A)})\cdot \mathrm{T}(\mathcal{N}^{(B)}).
    %\mathcal{N}^{(A)}\circ\mathcal{N}^{(B)}(P)= \sum_{Q\in \mathcal{P}}\left(t^{(A)}t^{(B)}\right)_{P,Q} Q.
\end{align}

\subsubsection{Normal form representation of a quantum channel}
\label{normalchannel}

\noindent We now present a quantum channel representation~\cite{King2001,BETHRUSKAI2002159} that will be useful when dealing with noisy random circuits. In words, it says that the Pauli transfer matrix of a single-qubit noise channel, up to unitary rotations, is diagonal in the sub-block corresponding to the non-identity Pauli matrices. 
%\daniel{Here I would separate what is ours and new and what is from others. The bound in $c$ is ours.} 
%\daniel{why are we including the proof of the normal form?}
We include here the lemma and proof of this representation for easy reference.
\begin{lemma}[Normal form of a quantum channel~\cite{King2001,BETHRUSKAI2002159}]
\label{le:normal}
    Any single-qubit quantum channel $\mathcal{N}$ can be written in the so called `normal' form:
    \begin{align}
        \mathcal{N}(\cdot)=U\mathcal{N}^{\prime}(V^{\dag}(\cdot)V)U^{\dag},
    \end{align}
    where $U$, $V$ are unitaries and $\mathcal{N}^{\prime}(\cdot)$ is a quantum channel with Pauli transfer matrix 
    \begin{equation}\label{eq:PTM_normalFORM}
\mathrm{T}(\mathcal{N}^{\prime})=\begin{bmatrix}
1 & 0 & 0 & 0 \\
t_X & D_X & 0 & 0 \\
t_Y & 0 & D_Y & 0 \\
t_Z & 0 & 0 & D_Z \\
\end{bmatrix},
\end{equation}
where $\bold{t} \coloneqq (t_X,t_Y,t_Z)$ and $\bold{D} \coloneqq (D_X,D_Y,D_Z) \in \mathbb{R}^3$, such that the entries of $\bold{D}$ have all the same sign. 
\end{lemma}

\begin{proof}
Let us consider the Pauli transfer matrix of $\mathcal{N}$, which is characterized by the real $3\times 3$ matrix $B\in\mathcal{L}(\mathbb{R}^{3})$ and the vector $\boldsymbol{b}=(b_X,b_Y,b_Z)\in \mathbb{R}^3$:
\begin{align}
    \mathrm{T}(\mathcal{N})=\begin{bmatrix}
1 & 0 & 0 & 0 \\
b_X & B_{X,X} & B_{X,Y}  & B_{X,Z}  \\
b_Y & B_{Y,X} & B_{Y,Y}  & B_{Y,Z}  \\
b_Z & B_{Z,X}  & B_{Z,Y}  & B_{Z,Z}  \\
\end{bmatrix}.
\end{align}
Any single qubit quantum state $\rho$ can be written as $\rho=(I+ \bold{w}\cdot \boldsymbol{\sigma})/2$, where $\bold{w} \coloneqq (w_X,w_Y,w_Z) \in \mathbb{R}^3$ with $\norm{w}_2\le 1$ and $\boldsymbol{\sigma} \coloneqq (X,Y,Z)$. 
Then, we have
\begin{align}
    \mathcal{N}(\rho)= \mathcal{N}\!\left(\frac{I}{2}\right) + \frac{1}{2}\mathcal{N}\!\left(\bold{w}\cdot \boldsymbol{\sigma} \right)= \left(\frac{I}{2} + \frac{1}{2}\boldsymbol{b}\cdot \boldsymbol{\sigma}\right)  + \frac{1}{2}(B \bold{w})\cdot \boldsymbol{\sigma} = \frac{I}{2}   + \frac{1}{2}(\boldsymbol{b}+B \bold{w})\cdot \boldsymbol{\sigma}.
\end{align}
Next, because $B$ is real, we can perform a real singular value decomposition of $B$ and have $B=R_1 D R^T_2$, where $D$ is a diagonal matrix with the non-negative diagonal elements $(D_X,D_Y,D_Z) \in \mathbb{R}^3$ and $R_1, R_2$ are in general orthogonal $\mathrm{O}(3)$ matrices. Now, every orthogonal matrix has determinant equal to $\pm 1$. This fact, along with the fact that $\textrm{det}(-R)=(-1)^3\textrm{det}(R)=-\textrm{det}(R)$ for every $R\in\mathrm{O}(3)$, means that we can, without loss of generality, assume that $R_1$ and $R_2$ both have determinant equal to $1$. In other words, we can assume that $R_1$ and $R_2$ are both special-orthogonal matrices in $\textrm{SO}(3)$. The diagonal elements $(D_X,D_Y,D_Z)$ are then not necessarily non-negative, but they all have the same sign.
We now use the fact that for every special-orthogonal matrix $R \in \mathrm{SO}(3)$, there exists a unitary $U\in\mathrm{U}(2)$ such that~\cite{Landau1981Quantum}
\begin{align}
    (R \boldsymbol{v})\cdot \boldsymbol{\sigma} = U (\boldsymbol{v}\cdot \boldsymbol{\sigma})U^{\dagger},
\end{align}
for all $\boldsymbol{v}\in \mathbb{R}^3$. The previous identity can be easily verified by choosing $U \coloneqq \exp(-i\frac{\theta}{2}\hat{n}\cdot\boldsymbol{\sigma})$, where $\hat{n}$ and $\theta$ are, respectively, the unit-norm vector and the rotation angle which characterizes the special-orthogonal matrix $R \in \mathrm{SO}(3)$.
Thus, we have that 
\begin{align}
    \mathcal{N}(\rho)=U \left(\frac{I}{2}   + \frac{1}{2}(R^{T}_1\boldsymbol{b}+ DR^T_2\bold{w})\cdot \boldsymbol{\sigma}\right)U^{\dagger} =U \mathcal{N}^{\prime}\left(\frac{I+ (R^T_2\bold{w})\cdot \boldsymbol{\sigma}}{2}\right)U^{\dagger} =U \mathcal{N}^{\prime}\left( V^{\dag}\left(\frac{I+ \bold{w}\cdot \boldsymbol{\sigma}}{2} \right)V\right)U^{\dagger}, 
\end{align}
where $U$ and $V$ are the unitaries associated to the special-orthogonal matrices $R_1$ and $R_2$, and $\mathcal{N}^{\prime}$ is the linear map such that $\mathcal{N}^{\prime}(\frac{I+ \bold{w}\cdot \boldsymbol{\sigma}}{2} )=\frac{I+ (\bold{t}+ D \bold{w})\cdot \boldsymbol{\sigma}}{2}   $,
where $\bold{t} \coloneqq R^T_1\boldsymbol{b}$. 
Hence, we have shown that $\mathcal{N}(\rho)$ can be written as $\mathcal{N}(\rho)= U \mathcal{N}^{\prime}\left( V^{\dag}\rho V\right)U^{\dagger}$, where the Pauli transfer matrix of $\mathcal{N}^{\prime}$ is the one in Eq.~\eqref{eq:PTM_normalFORM}.
%Let us now prove Eq.\eqref{eq:ineqnormal}. 
\end{proof}

Thus, every single-qubit quantum channel $\mathcal{N}$ can be expressed as $\mathcal{N}(\cdot)=U\mathcal{N}^{\prime}(V^{\dag}(\cdot)V)U^{\dag}$, where $U$, $V$ are unitaries, and $\mathcal{N}^{\prime}$ is a quantum channel such that it acts on a quantum state written in its Bloch sphere representation as 
\begin{align}
    \mathcal{N}^\prime \!\left(\frac{I+ \bold{w}\cdot \boldsymbol{\sigma}}{2}\right)= \frac{I}{2}   + \frac{1}{2}(\bold{t}+D \bold{w})\cdot \boldsymbol{\sigma},
    \label{eq:normSS}
\end{align}
where $\bold{w}\in \mathbb{R}^3$ with $\|\bold{w}\|_2\le 1$, $\bold{t} \coloneqq (t_X,t_Y,t_Z)\in \mathbb{R}^3$ and $D \coloneqq \mathrm{diag}(\bold{D})$ with $\bold{D} \coloneqq (D_X,D_Y,D_Z) \in \mathbb{R}^3$.

We now prove that the parameters of the normal form representation satisfy a particular constrain. Such constrain will be crucial in our following discussion. 
\begin{lemma}[Contraction coefficient in terms of the normal form parameters]
\label{le:contractionnormal}
For any single-qubit quantum channel, the parameters $\bold{t}, \bold{D} \in \mathbb{R}^3$ of its normal form representation satisfy
\begin{align}
    c\coloneqq\frac{1}{3}(t^2_X + D^2_X + t^2_Y + D^2_Y + t^2_Z + D^2_Z)\le 1,
    \label{eq:ineqnormal}
\end{align}
and the equality is saturated if and only if the channel is unitary. \textcolor{black}{Furthermore, it also holds $\norm{\bold{t}}_2\le 1$ and $D_P \le 1$ for all $P\in \{X,Y,Z\}$.}
\end{lemma}
\begin{proof}
    Because of the previous Lemma, any single-qubit quantum channel $\mathcal{N}$ can be expressed as $\mathcal{N}(\cdot)=U\mathcal{N}^{\prime}(V^{\dag}(\cdot)V)U^{\dag}$, where $U$, $V$ are unitaries, and $\mathcal{N}^{\prime}$ such that it holds Eq.~\eqref{eq:normSS}.
    Let $\rho$ be an arbitrary qubit quantum state. 
    Noting that $\mathcal{N}^{\prime}$ is a quantum channel, on account of being a composition of quantum channels, it holds that $\norm{\mathcal{N}^{\prime}(\rho)}_{\infty}\leq 1$. If we let $\bold{w}\in\mathbb{R}^3$ be the Bloch vector corresponding to $\rho$, then because $\norm{\mathcal{N}^{\prime}(\rho)}_{\infty}$ is equal to the largest eigenvalue of $\mathcal{N}^{\prime}(\rho)$, we find that
\begin{align}
       1\geq \norm{\mathcal{N}^{\prime}(\rho)}_{\infty}= \frac{1}{2}(1+ \norm{\bold{t}+ D \bold{w}}_2).
\end{align}
Hence, we get
\begin{align}
  (t_X+ D_X w_x)^2+(t_Y+ D_Y w_y)^2+(t_Z+ D_Z w_z)^2=\norm{\bold{t}+ D \bold{w}}^2_2\le 1.
    \label{eq:blochCONSTR}
\end{align}
Now, recall that $\norm{\bold{w}}_2\le 1$. If $\bold{w}=0$, then we get $\norm{\bold{t}}_2\le 1$. In particular by choosing $w=(\pm 1,0,0)$, we get
\begin{align}
    (t_X \pm D_X)^2\le 1,
    \label{eq:pmtxDX}
\end{align}
and similarly for $Y$ and $Z$ (from which follows that $D_P \le 1$ for all $P\in \{X,Y,Z\}$). Now, assume that the entries of $\bold{D}$ are all non-negative (remember that they have the same sign). Together with the previous equation, this implies that
\begin{align}
     t_X^{ 2} + D^2_X + t_Y^{ 2} + D^2_Y +t_Z^{ 2 } + D^2_Z \le (t_X + \mathrm{sign}(t_X)D_X )^2+(t_Y + \mathrm{sign}(t_Y) D_Y )^2+(t_Z + \mathrm{sign}(t_Z)D_Z )^2\le 3.
\end{align}
Similarly, if the entries of $\bold{D}$ are all negative, we have 
\begin{align}
     t_X^{ 2} + D^2_X + t_Y^{ 2} + D^2_Y +t_Z^{ 2 } + D^2_Z \le (t_X - \mathrm{sign}(t_X)D_X )^2+(t_Y - \mathrm{sign}(t_Y) D_Y )^2+(t_Z - \mathrm{sign}(t_Z)D_Z )^2\le 3.
\end{align}
This proves Eq.~\eqref{eq:ineqnormal}.
%\armando{I think this can be improved. Pick $w = \sqrt{\frac{1}{3}}(sign(t_X), sign(t_Y), sign(t_Z))$. Then we have\[ \sum_{P\in \{X,Y,Z\}}  t_P^2 + \frac{1}{3} D_P^2 \le \sum_{P\in \{X,Y,Z\}} (t_P + \sqrt{\frac{1}{3}}sign(t_P) D_P)^2 \le 1 \] \[ \implies \sum_{P\in \{X,Y,Z\}}  3 t_P^2 +  D_P^2 \le 3 \]}

Finally, we show that Eq.~\eqref{eq:ineqnormal} is saturated if and only if $\mathcal{N}$ is unitary. If $\mathcal{N}$ is unitary, then also $\mathcal{N}^{\prime}$ is unitary. This implies that $\bold{t}=0$, because unitary channels are also unital. Moreover, it also implies that the purity of any state must remain the same, so the diagonal matrix $D\coloneqq\mathrm{diag}(\bold{D})$ must be norm-preserving, hence orthogonal. Therefore, we have $D=\pm \mathrm{diag}(1,1,1)$. This saturates inequality~\eqref{eq:ineqnormal}.
Now, let us assume that 
\begin{align}
    \frac{1}{3}(t^2_X + D^2_X + t^2_Y + D^2_Y + t^2_Z + D^2_Z)= 1.
    \label{eq:sat}
\end{align}
From the inequality~\eqref{eq:pmtxDX}, we get also that $t_X^{2} + D^2_X\le 1$, and the same for $Y$ and $Z$. Hence, Eq.~\eqref{eq:sat} implies $t_X^{2} + D^2_X= 1$, and the same for $Y$ and $Z$. Using this with Eq.~\eqref{eq:pmtxDX}, we have that the possible values for $t_X^{2}$ and $D_X^{2}$ are, respectively, $1$ and $0$, or vice versa. Similarly for $Y$ and $Z$.  From Eq.~\eqref{eq:blochCONSTR}, we deduce that the only possibility is that $\bold{t}=0$ and that $D_X^{2}=D_Y^{2}=D_Z^{2}=1$. Hence, we have that the Pauli transfer matrix of $\mathcal{N}^{\prime}$ is equal, up to a possible minus sign factor, to the identity matrix. This implies that $\mathcal{N}^{\prime}$ must be the identity channel and that $\mathcal{N}$ is unitary.
\end{proof}

Here, we give examples of normal form parameters $\bold{t} =(t_X,t_Y,t_Z)$ and $\bold{D} = (D_X,D_Y,D_Z)$ for standard noise channels.
The single-qubit depolarizing channel with parameter $p\in [0,1]$ can be defined as
\begin{align}
    \mathcal{N}_p^{(\mathrm{\mathrm{dep}})}(\sigma)=(1-p)\sigma + p\Tr(\sigma) \frac{I}{2}.
    \label{eq:depnoise}
\end{align}
Its normal form parameters are $\bold{t} = (0,0,0)$ and $\bold{D} = (1-p,1-p,1-p)$.
The amplitude damping channel $\mathcal{N}_q^{(\text{amp})}$, parameterized by $q \in [0,1]$, is given in the computational basis as
\begin{equation}
    \mathcal{N}_q^{(\text{amp})}(\sigma) = \begin{pmatrix}
        \sigma_{0,0} + q\sigma_{1,1} & \sqrt{1-q}\sigma_{0,1} \\
        \sqrt{1-q}\sigma_{1,0} & (1-q)\sigma_{1,1} \\
    \end{pmatrix},
    \label{eq:ampdampnoise}
\end{equation}
where $\sigma_{i,j} \coloneqq \bra{i}\sigma\ket{j}$. Here, $\bold{t} = (0, 0, q)$ and $\bold{D} = (\sqrt{1-q}, \sqrt{1-q}, 1-q)$.
The single-qubit dephasing channel $\mathcal{N}_p^{(\text{deph})}$ with $p \in [0,1]$ can be defined as
\begin{equation}
    \mathcal{N}_p^{(\text{deph})}(\sigma) = \begin{pmatrix}
        \sigma_{0,0} & (1-p)\sigma_{0,1} \\
        (1-p)\sigma_{1,0} & \sigma_{1,1} \\
    \end{pmatrix}.
\end{equation}
Its normal form parameters are $\bold{t} = (0, 0, 0)$ and $\bold{D} = (1-p, 1-p, 1)$.

\subsection{Circuit and noise model}
\label{sub:circuitmodel}

\noindent In our work, we examine $n$-qubit quantum circuits $\Phi$ formed by layers of two-qubit random unitary gates interleaved by local noise, with a final layer of random single qubit gates. For example, the standard brickwork circuit architecture is within our model (Figure~\ref{Fig:circ}). Mathematically,
\begin{align}\label{eq:randcirc}
\Phi\coloneqq \mathcal{V}^{\mathrm{single}}\circ\mathcal{N}^{\otimes n} \circ \mathcal{U}_{L} \circ \cdots  \circ \mathcal{N}^{\otimes n} \circ \mathcal{U}_{1}  ,
\end{align}
where $\mathcal{V}^{\mathrm{single}} \coloneqq V(\cdot)V^{\dagger}$ with $V \coloneqq \bigotimes^{n}_{i=1} u_i $ is a layer of single-qubit gates, $\mathcal{U}_{i} \coloneqq U_i(\cdot)U^{\dag}_i$ corresponds to the $n$-qubit unitary channel associated with the $i$-th unitary layer $U_i$ for $i\in [L]$, and $\mathcal{N}$ is a single-qubit quantum channel. \textcolor{black}{We point out that the single-qubit layer $\mathcal{V}^{\mathrm{single}}$ included in our model is not essential for our main results and can be omitted with minor modifications.}%, as will be explained in more detail.}

\begin{figure}[h]
\centering
\includegraphics[width=0.65\textwidth]{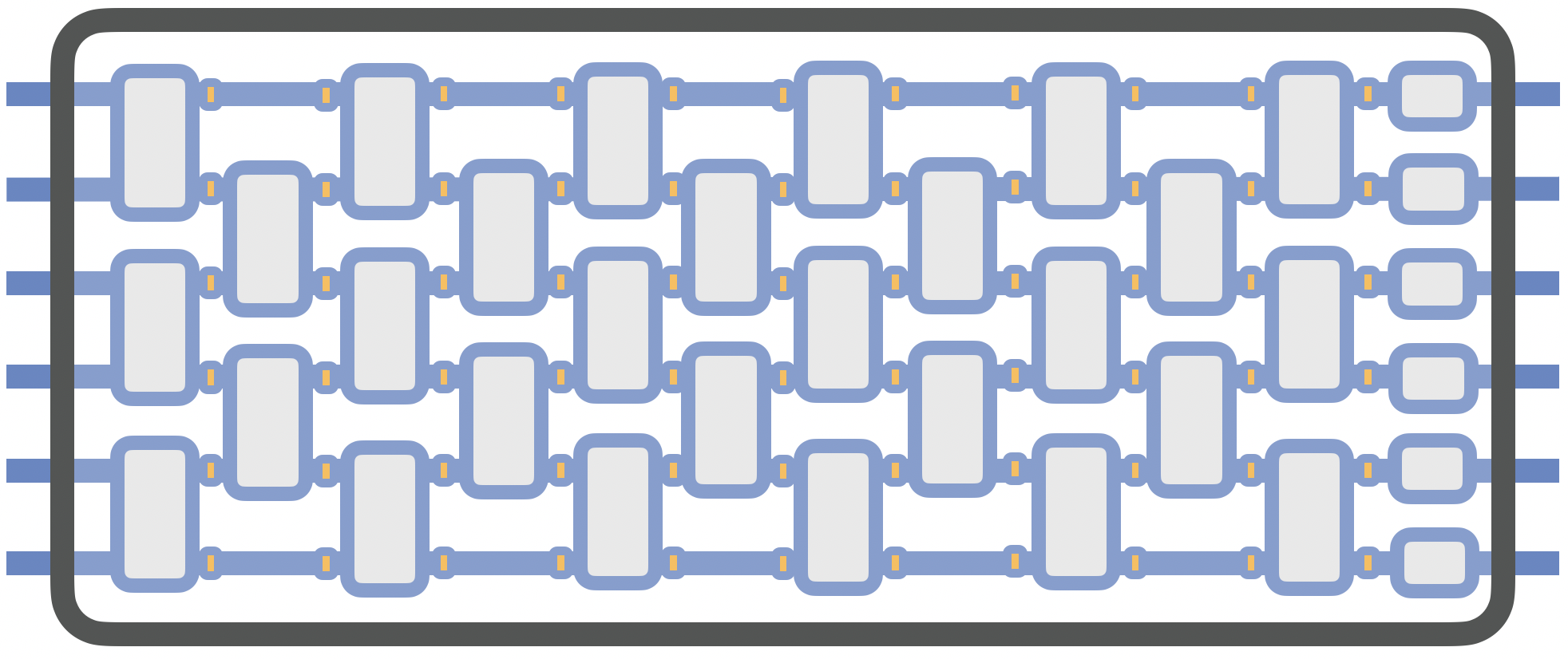}
\caption{Example of the architecture that our model encompasses: A brickwork circuit composed of two-qubit gates followed by local noise (depicted with yellow circles).}
\label{Fig:circ}
\end{figure}

\subsubsection{Assumption on the circuit distribution}
Firstly, we assume that each single-qubit gate in the layer $\bigotimes^{n}_{i=1} u_i $ is distributed according to a single-qubit $2$-design (e.g., Haar random). 
Moreover, we assume that each unitary layer $\mathcal{U}_{i} \coloneqq U_i(\cdot)U^{\dagger}_i$, for $i\in [L]$, consists of two-local qubit gates, each forming a local $2$-design. More precisely, we assume that each $U_i$ for $i\in [L]$ is distributed according to a $2$-local $2$-design layer distribution, defined as follows:

\begin{definition}[$2$-local $2$-design layer distribution]
\label{def:localdesignlayer}
We say that $\nu$ is a $2$-local $2$-design layer distribution if and only if it is a probability distribution over quantum circuits formed by local $2$-qubit gates, where each of them is distributed accordingly to a local $2$-design and each qubit is acted on by at least one of the gates.
\end{definition}

Moreover, we point out that we consider an arbitrary circuit geometry/architecture, i.e., we do not make any particular assumptions on the geometric dimensionality of our circuit, except when explicitly mentioned. Note that our model is in stark contrast to works~\cite{schumann2023emergence,ErrorMitigationObstructions} 
where the unitary layers are chosen as global $n$-qubit $2$-designs, and we expect that the model that we consider is more realistic.

\subsubsection{Noise model}
\label{sub:noisemodel}

Moreover, since before and after any noise channel $\mathcal{N}$ there is a gate that is distributed according to a $2$-design and in our work we consider up to second moment quantities, because of the normal form representation of the channel and unitary invariance, we can restrict the noise channels $\mathcal{N}$ to have a sparse Pauli transfer matrix of the form of Eq.~\eqref{eq:PTM_normalFORM}, characterized by two real vectors $\bold{t} \coloneqq (t_X,t_Y,t_Z)$ and $\bold{D} \coloneqq (D_X,D_Y,D_Z)$. 
In particular, the adjoint channel $\mathcal{N}^{*}$ acts on $Q\in \{X,Y,Z\}$ as 
\begin{align}
\label{eq:adjointeq}
    \mathcal{N}^{*}(Q)=t_{Q} I + D_{Q}Q=\sum_{a\in\{0,1\}} D_{Q}^{a}t_{Q}^{1-a}Q^a. 
\end{align}
Since we work with at most second-moment quantities, we often single out `for free' from each $2$-local $2$-design unitary layer $\{\mathcal{U}_{i}\}^L_{i=1}$ layers of single qubit Haar random gates, due to the invariance of the Haar measure and because each qubit is a acted on by at least one of the $2$-local $2$-design gates. Specifically, without loss of generality, we can consider equivalently circuits of the form
\begin{align}
    \Phi=(\mathcal{V}_{L}^{\,\mathrm{single}}\circ\mathcal{N}^{\otimes n} \circ \mathcal{U}_L) \circ \cdots   \circ (\mathcal{V}_{1}^{\,\mathrm{single}}\circ\mathcal{N}^{\otimes n} \circ \mathcal{U}_1)  ,
    \label{eq:randcirc23ss}
\end{align}
where $\{\mathcal{V}_k^{\mathrm{single}}\}^L_{k=1}$ are layers of single-qubit gates distributed according a single-qubit $2$-design. \textcolor{black}{However, these single-qubit layers do not play a fundamental role in our model and could be omitted. (In particular, the last single-qubit layer might also be omitted with minor modifications: if the circuit ends with a layer of noise rather than a single-qubit layer, this noise layer and the preceding unitary layer can be absorbed into the observable in the Heisenberg picture. Thus, we can reduce the scenario to one where the circuit terminates with a layer of single-qubit gates, due to unitary invariance, and with an observable of comparable locality.)}

We will often denote circuits derived from $\Phi$ by removing the last layer of single-qubit gates and the last layer of noise.  In this case, we use the notation
\begin{align}
    \Phi^{\prime} & \coloneqq (\mathcal{N}^{\otimes n} \circ \mathcal{U}_L) \circ \cdots   \circ (\mathcal{V}_{1}^{\,\mathrm{single}}\circ\mathcal{N}^{\otimes n} \circ \mathcal{U}_1), \\
\Phi^{\prime\prime} & \coloneqq  \mathcal{U}_L \circ \cdots   \circ (\mathcal{V}_{1}^{\,\mathrm{single}}\circ\mathcal{N}^{\otimes n} \circ \mathcal{U}_1)  .
\end{align}
Here, $\Phi^{\prime}$ denotes the circuit without the final layer of single-qubit gates, while $\Phi^{\prime \prime}$ denotes the circuit without the final layer of single-qubit gates layer and also without the final layer of noise.

We also often need to denote circuits derived from $\Phi$ by retaining certain layers from the start or from the end. In these cases, we use subscripts to indicate the relevant layers. That is, for $a\le b \in [L]$,
\begin{align}
\Phi_{[a,b]} & \coloneqq  (\mathcal{V}_b^{\mathrm{single}} \circ \mathcal{N}^{\otimes n} \circ \mathcal{U}_b) \circ \cdots \circ (\mathcal{V}_a^{\mathrm{single}} \circ\mathcal{N}^{\otimes n} \circ \mathcal{U}_{a}). %, \\
%\Phi_{[1,k]} & \coloneqq   (\mathcal{V}_{k-1}^{\mathrm{single}}\circ \mathcal{N}^{\otimes n} \circ \mathcal{U}_{k-1}) \circ \cdots \circ (\mathcal{V}_1^{\mathrm{single}}\circ \mathcal{N}^{\otimes n} \circ \mathcal{U}_1).
\end{align}
%
%
%\begin{align}
%\Phi_{[k,L]} & \coloneqq  (\mathcal{V}_{L}^{\,\mathrm{single}}\circ\mathcal{N}^{\otimes n} \circ \mathcal{U}_L) \circ \cdots   \circ (\mathcal{V}_{k}^{\,\mathrm{single}}\circ\mathcal{N}^{\otimes n} \circ \mathcal{U}_k) , \\
%\Phi_{[1,k]} & \coloneqq   (\mathcal{V}_{k-1}^{\mathrm{single}}\circ \mathcal{N}^{\otimes n} \circ \mathcal{U}_{k-1}) \circ \cdots \circ (\mathcal{V}_1^{\mathrm{single}}\circ \mathcal{N}^{\otimes n} \circ \mathcal{U}_1).
%\end{align}
%Similarly, for $k \in [L]$, when retaining $k$ layers from the end, we will use the following notation,
%\begin{align}
%\Phi_{[1,k]} & \coloneqq  \mathcal{V}^{\,\mathrm{single}}_{k+1}\circ(\mathcal{N}^{\otimes n} \circ \mathcal{U}_k) \circ \cdots   \circ (\mathcal{N}^{\otimes n} \circ \mathcal{U}_1) , 
%\end{align}
%where here $\mathcal{V}^{\,\mathrm{single}}_{k+1}$ is a layer of single qubit Haar random gates that we singled-out ``for free" from the $2$-design unitary layer $\mathcal{U}_k$.
%For example, $\Phi_{[L-1,L]}=\mathcal{V}^{\mathrm{single}}\circ\mathcal{N}^{\otimes n} \circ \mathcal{U}_{L}\circ\mathcal{N}^{\otimes n} \circ \mathcal{U}_{L-1}$.
Whenever we write an expectation value, $\Ex[\cdot]$ or $\Var[\cdot]$, without explicitly specifying the underlying distribution, we consider the probability distribution over the defined random circuit.
%\newpage

\newpage

\section{Observable expectation values of noisy random quantum circuits}
\label{exp_values}
In this section, we analyze expectation values of random quantum circuits under possibly non-unital noise. We make here a summary of our results that we analyze in detail in their respective subsections.
We consider a circuit architecture $\Phi$ as we described in subsection~\ref{sub:circuitmodel}, where the local noise channels are characterized by the parameters of their normal form representation $\bold{t} \coloneqq (t_X,t_Y,t_Z)$ and $\bold{D} \coloneqq (D_X,D_Y,D_Z)$, which we assume to be constants with respect to the number of qubits. 
Our first main theorem is the following.

\begin{theorem}[Variance of expectation values of random circuits with non-unital noise]
\label{th:variance}
Let $H \coloneqq \sum_{P\in \{I,X,Y,Z\}^{\otimes n}} a_P P$, with $a_P\in \mathbb{R}$ for $P\in \{I,X,Y,Z\}^{\otimes n}$, be an arbitrary Hamiltonian. Let $\rho$ be a quantum state. We assume that the noise is non-unital, specifically $\norm{\bold{t}}_2=\Theta(1)$. Then, at any depth of the noisy circuit $\Phi$, as defined in Eq.~\eqref{eq:randcirc}, we have
\begin{align}
       \Var[\Tr(H \Phi(\rho)  )] = \sum_{P\in \{I,X,Y,Z\}^{\otimes n}\setminus I^{\otimes n}} a_P^2 \exp(-\Theta(|P|)).
\end{align}  
\end{theorem}

To prove Theorem~\ref{th:variance}, we make use of results we show in subsection~\ref{sub:local} and subsection~\ref{sub:global}, where we respectively show a lower bound on the variance (Proposition~\ref{prop:lbvar}) and a matching upper bound (Proposition~\ref{prop:upvar}). 

Theorem~\ref{th:variance} directly implies that the variance of expectation value of local observables (i.e., \emph{local expectation values}) can be significantly large, e.g., $\Var[\Tr(Z_1 \Phi(\rho)  )]=\Omega(1)$. This means that local expectation values can deviate significantly from their mean value $\Ex[\Tr(H \Phi(\rho)  )]$. This is in stark contrast to the behaviour of noiseless random quantum circuits or circuits with unital noise~\cite{napp2022quantifying,nibp}. Theorem~\ref{th:variance} also implies that the variance of expectation value of global observables (i.e., \emph{global expectation values}) are exponentially concentrated to their mean value, e.g., $\Var[\Tr(Z^{\otimes n} \Phi(\rho)  )]=\exp(-\Omega(n))$.

\textcolor{black}{Moreover, based on the results shown in subsection~\ref{sub:effective} and subsection~\ref{sub:tracedistance}, we establish the following theorem.}

\begin{theorem}[Average distance between two quantum states]
\label{th:effective}
Let \(P \in \{I, X, Y, Z\}^{\otimes n}\). Let \(\rho\) and \(\sigma\) be any quantum states. Consider \(\Phi\) as any noisy random quantum circuit of depth \(L\), as defined in Eq.~\eqref{eq:randcirc}. Assume that the noise is not a unitary channel. Then, it holds that 
\begin{align}
\label{eq:paulidecSUM}
    \Ex_{\Phi}[|\Tr(P \Phi(\rho)) - \Tr(P \Phi(\sigma))|] \leq \exp\!\left(-\Omega\!\left(L + |P|\right)\right),
\end{align}  
and for all observables \(O\), we have
\begin{align}
\label{eq:OdecSUM}
    \Ex_{\Phi}[|\Tr(O \Phi(\rho)) - \Tr(O \Phi(\sigma))|] \leq \norm{O}_{\infty} \exp\!\left(-\Omega\!\left(L\right)\right).
\end{align}  
This further implies that for \(L = \Omega(n)\), we have
\begin{align}
    \Ex_{\Phi}[\norm{\Phi(\rho) - \Phi(\sigma)}_{1}] \leq \exp(-\Omega(n)),
\end{align}  
where the expected value is taken over the $2$-design distribution of the two-qubits gates that compose the circuit  \(\Phi\).
\end{theorem}

%We remark that the previous bound holds not only for all Pauli $P$, but for all observables $O$ with bounded operator norm $\norm{O}_{\infty}$ (or more generally with bounded normalized $2$-norm $\norm{O}^2_{2}/2^n$).

The average trace distance upper bound is proven in subsection~\ref{sub:tracedistance}, where we prove also a worst-case trace distance upper bound (i.e., without the expected value) that holds in certain high-noise regime. Note that we cannot hope to prove a worst-case upper bound on the trace distance that is valid for every noise regime. This is because there are quantum error correction methods, such as the so-called \emph{quantum refrigerator} construction~\cite{refrigerator}, which can leverage non-unital noise to perform fault-tolerant quantum computation in a model similar to ours, up to depths that are exponential in the number of qubits. Thus, for these special classes of circuits, the trace distance remains of constant order. Moreover, it is known that this result is tight~\cite{fawzi2022lower}, as bounds on the worst-case convergence are known in the regime where the depth is exponential in the number of qubits.

\textcolor{black}{From a direct application of Eq.~\eqref{eq:OdecSUM}, it follows that with high probability over the choice of the random circuit, considering only the last \(O(\log(\varepsilon^{-1}))\) layers suffices for the estimation of expectation values with \(\varepsilon\) precision. In particular, we get the following:}
\begin{proposition}[Effective depth]
\label{prop:classim}
Let \(\varepsilon, \delta > 0\). Let \(O\) be any observable, and let \(\rho_0\) be any initial state. Consider a noisy quantum circuit \(\Phi\) of depth \(L\), as defined in Eq.~\eqref{eq:randcirc}. With probability at least \(1-\delta\) over the choice of the random circuit, we have
\begin{align}
    |\Tr(O \Phi(\rho_0)) - \Tr(O \Phi_{[L-\ell,L]}(\sigma_0))| \leq \varepsilon,
\end{align}
where \(\sigma_0\) is any preferred initial state (e.g., \(\sigma_0 \coloneqq \ketbra{0^n}{0^n}\)). Here, \(\Phi_{[L-\ell,L]}\) denotes the channel \(\Phi\) restricted to the last \(\ell\) layers, where \(\ell \coloneqq O(\log(\norm{O}_{\infty}/(\delta \varepsilon^2)))\).
\end{proposition}

Note that if \(O\) is local and the desired accuracy \(\varepsilon\) is constant in the number of qubits, then \(\Tr(O \Phi_{[L-\ell,L]}(\rho_0))\) can be computed efficiently classically via light-cone arguments. 
Moreover, if \(\Phi^{*}_{[L-\ell,L]}(O)\) is close to something proportional to the identity (which can be verified efficiently classically), then we can certify that our algorithm has succeeded. If \(O\) is a global Pauli observable, then we can just output zero for estimating the expectation value with inverse-polynomial precision (because of Eq.~\eqref{eq:paulidecSUM}). 
Collectively, these insights underpin a classical simulation algorithm capable of estimating Pauli expectation values of random quantum circuits affected by—possibly non-unital—noise, as explained in Subsection~\ref{sub:simu}. Its runtime depends polynomially on the inverse of the precision for one-dimensional architectures and quasipolynomially for higher-dimensional ones.

%From a direct application of Eq.~\eqref{eq:paulidecSUM}, it follows that we can compute classically expectation values for most of the circuit instances, as explained in Subsection~\ref{sub:simu}.

\subsection{Variance lower bound: Local expectation values with non-unital noise are not exponentially concentrated}
%\subsection{Variance Lower Bound: Local Expectation Values typically deviate significantly by their mean}
%\subsection{Variance lower bound: Local Expectation Values are not typically exponentially concentrated at any depth}
\label{sub:local}

In this subsection, we show that local expectation values of average quantum circuits with non-unital noise can be far from zero, in contrast to what happens with unital noise or in the noiseless case. Let us consider a circuit $\Phi$ as described in Subsection~\ref{sub:circuitmodel}, where the local noise channel is characterized by the parameters of its normal form representation $\bold{t} \coloneqq (t_X,t_Y,t_Z)$ and $\bold{D} \coloneqq (D_X,D_Y,D_Z)$, and we consider any circuit depth $L\ge 1$. 

\begin{proposition}[Lower bound on the variance]
\label{prop:lbvar}
Let $H \coloneqq \sum_{P\in \{I,X,Y,Z\}^{\otimes n}} a_P P$, with $a_P\in \mathbb{R}$ for all $P\in \{I,X,Y,Z\}^{\otimes n}$, be an arbitrary Hamiltonian. Let $\rho$ be a quantum state. Then, for any depth of the noisy circuit $\Phi$, we have
    \begin{align}
       \Var[\Tr(H \Phi(\rho)  )] \ge \sum_{P\in \{I,X,Y,Z\}^{\otimes n}\setminus I^{\otimes n}} a_P^2 \left(\frac{\|\bold{t}\|^2_2}{3}\right)^{|P|},
       \label{eq:lbpaulipurity}
       \end{align}
       where we note that $\norm{\bold{t}}_2$ is non-zero if the noise channel is non-unital.
\end{proposition}

\begin{proof}
Because our circuit ends with a layer of random single qubit gates $\otimes^{n}_{i=1} u_i $, it holds that $\Ex[\Tr(P \Phi(\rho)  )]=0$ for any $P\in \{I,X,Y,Z\}^{\otimes n} \setminus I^{\otimes n}$, which follows from Lemma~\ref{le:qubitonedesign}. We, therefore, 
have that
\begin{align}
    \Ex[\Tr(H \Phi(\rho)  )]=a_{I^{\otimes n}}.
\end{align}
We now focus on $\Ex[\Tr(H \Phi(\rho) )^2]$. First of all, using point 1 of Lemma~\ref{le:mixham}, we have
\begin{align}
 \Ex[\Tr(H \Phi(\rho) )^2]=\sum_{P\in \{I,X,Y,Z\}^{\otimes n}} a_P^2 \Ex[\Tr(P \Phi(\rho))^2] = a^2_{I^{\otimes n}} +\sum_{P\in \{I,X,Y,Z\}^{\otimes n}\setminus I^{\otimes n}} a_P^2 \Ex[\Tr(P \Phi(\rho))^2].
   % \Ex[\Tr(H \Phi(\rho) )^2]=\sum_{P\in \{I,X,Y,Z\}^{\otimes n}} a_P^2 \Ex[\Tr(P \Phi(\rho))^2].
\end{align}
%where in the second equality we have used the fact that our circuit ends with a layer of single qubits $2$-design gates $\otimes^{n}_{i=1} u_i $ and the Pauli mixing property, specifically point $1$ of Lemma~\ref{le:mixham}.
Let us now analyze each term $\Ex[\Tr(P \Phi(\rho))^2]$ in the sum above separately. Using point 2 of Lemma~\ref{le:mixham}, we obtain
\begin{align}
    \Ex[\Tr(P \Phi(\rho) )^2 ]=\frac{1}{3^{|P|}}\sum_{\substack{Q \in \{I,X,Y,Z\}^{\otimes n}:\\ \supp(Q)=\supp(P) } }\Ex[\Tr(Q \Phi^{\prime}(\rho) )^2],
\end{align}
which corresponds to `removing' the last layer of single-qubit gates and using the Pauli mixing property. Recall that $\Phi^{\prime}$ denotes the noisy circuit channel without the last layer of single qubit gates, while $\Phi^{\prime\prime}$ denotes the noisy circuit channel without the last layer of single qubit gates and last layer of noise, i.e., $\Phi'=\mathcal{N}^{\otimes n}\circ\Phi''$. Taking the adjoint of the noise, and using the fact that $\mathcal{N}^{\ast}$ is a unital channel, we obtain
%\je{[?????]}
%Next, we have 
\begin{align}
\nonumber
    \Ex[\Tr(Q \Phi^{\prime}(\rho) )^2]&= \Ex\left[\Tr(\mathcal{N}^{* \otimes n}(Q)\Phi^{\prime\prime}(\rho) )^2\right]\\
      \nonumber
      &=\Ex\left[\Tr(\left(\bigotimes_{j \in \supp(Q)} \mathcal{N}^{*}(Q_j)\right) \Phi^{\prime\prime}(\rho) )^2\right]\\
    \nonumber&=\Ex\left[\Tr(\left(\bigotimes_{j \in \supp(Q)} (t_{Q_j }I_j+ D_{Q_j}Q_j)\right) \Phi^{\prime\prime}(\rho) )^2\right]\\& =\Ex\left[\Tr(\left(\sum_{a\in \{0,1\}^{|Q|}}\bigotimes_{j \in \supp(Q)} (t^{a_j}_{Q_j }  D^{1-a_j}_{Q_j}Q^{1-a_j}_j)\right) \Phi^{\prime\prime}(\rho) )^2\right] 
    \nonumber 
    \\
    & =\sum_{a\in \{0,1\}^{|Q|}}\prod_{j \in \supp(Q)} (t^{a_j}_{Q_j } D^{1-a_j}_{Q_j})^2\Ex\left[\Tr(\left(\bigotimes_{k \in \supp(Q)}  Q^{1-a_k}_k\right) \Phi^{\prime\prime}(\rho) )^2\right]
    \label{eq:upstep}
\end{align}
where in the %first step we took the adjoint of the last layer of noise, in the second step we have used that any adjoint channel is unital,
third step we have used the normal-form parametrization of the channel, specifically, Eq.~\eqref{eq:adjointeq}. The fifth step follows by observing that we can apply point $1$ of Lemma~\ref{le:mixham}, which we can do because $\Phi^{\prime\prime}$ ends with a $2$-local $2$-design unitary layer, hence we can single-out from it a layer of single qubit Haar random gates, due to the invariance of the Haar measure and because each qubit is a acted on by at least one of the $2$-qubit $2$-design gate. 

Now, we are left with a sum of positive terms and from such sum we can keep only the term corresponding to identity term, and we lower bound the remaining terms with zero. This implies that
\begin{align}
\nonumber
  &\sum_{a\in \{0,1\}^{|Q|}}\!\prod_{j \in \supp(Q)} \!t^{2a_j}_{Q_j } D^{2(1-a_j)}_{Q_j}\Ex\!\left[\Tr(\left(\bigotimes_{k \in \supp(Q)}  Q^{1-a_k}_k\right) \Phi^{\prime\prime}(\rho) )^2\right]
    \\
    \nonumber &\ge  |t_X|^{2|Q|_X}|t_Y|^{2|Q|_Y}|t_Z|^{2|Q|_Z}\Ex\left[\Tr(I_n \Phi^{\prime\prime}(\rho))^2\right]\\
    &=|t_X|^{2|Q|_X}|t_Y|^{2|Q|_Y}|t_Z|^{2|Q|_Z}, \label{eq:lbpauli}
\end{align}
where, in the last step, we have used simply that density matrices have unit trace.
Substituting, we get 
\begin{align}
    \Ex\left[\Tr(P \Phi(\rho))^2\right] \ge \frac{1}{3^{|P|}} \sum_{\substack{Q \in \{I,X,Y,Z\}^{\otimes n}:\\ \supp(Q)=\supp(P) }} |t_X|^{2|Q|_X}|t_Y|^{2|Q|_Y}|t_Z|^{2|Q|_Z} = \frac{1}{3^{|P|}} \left(|t_X|^2+|t_Y|^2+|t_Z|^2\right)^{|P|} 
\end{align}
where we have used the multinomial theorem in the last equality. By putting everything together and using the definition of variance, we can conclude.
\end{proof}
Note that if the noise is unital, i.e.,  $\|\bold{t}\|_2=0$, the previous lower bound becomes vacuous.
As an immediate corollary of the previous inequality we have the following.
\begin{corollary}[Local expectation values are not exponentially concentrated on average]
Let $P\in \{I,X,Y,Z\}^{\otimes n}$ be a Pauli operator with weight $|P|=\Theta(1)$. Let us assume that the noise is non-unital, specifically that  $\norm{\bold{t}}_2=\Theta(1)$. Then, we have
\begin{align}
    \Var[\Tr(P \Phi(\rho) )]= \Theta(1).
\end{align}
\end{corollary}
We can easily translate the fact that the variance is large into the fact that the probability of deviating from the mean is large. For example, let $C \coloneqq \Tr(P \Phi(\rho) )$, with $\sup(|C|)\le 1$. By using the following probability inequality (see Lemma~\ref{le:probstat} in the last \emph{miscellaneous} section), we find that
%\antonio{is it necessary to use this homemade inequality?}
\begin{align}
    \mathrm{Prob}\left(\left| C  - \mathbb{E}[C] \right| > \sqrt{\frac{\mathrm{Var}[C]}{2}} \right) \ge \frac{1}{8}\mathrm{Var}[C].
\end{align}
Note that for an inverse polynomially small non-unital noise rate, i.e., $\|\bold{t}\|_2 = \Omega\left(\frac{1}{\mathrm{poly}(n)}\right)$, we would get that 
\begin{equation}
\Var[\Tr(P \Phi(\rho))] = \Omega\left(\frac{1}{\mathrm{poly}(n)}\right), 
\end{equation}
which implies a lack of exponential concentration of local expectation values even for such small non-unital noise regime.

From a more technical perspective, we have shown that \emph{random quantum circuits} with non-unital noise have local expectation values that are not exponentially concentrated.
This is in stark contrast with the behavior of random quantum circuits in the noiseless regime or with local depolarizing noise~\cite{napp2022quantifying}, as summarized in Table~\ref{tab:costconcentration}. 
\begin{table}[h]
    \centering
    \caption{\textbf{Concentration of local expectation values for $\Omega(n)$-depth circuits}}
    \begin{tabular}{lcc}  \hline
    \textbf{Noise model}
    &\quad\quad\quad & \textbf{$\Var\!\left[\Tr\!\left(Z_1 \rho \right)\right]$}
    \\ \hline
    Noiseless~\cite{McClean_2018,napp2022quantifying} &\quad\quad\quad& $\exp(-\Theta(n))$ \\
        Unital noise~\cite{nibp}  &\quad\quad\quad & $\exp(-\Theta(n))$ \\
        
        Non-unital noise~[This work] 
        &\quad\quad\quad& $\Theta(1)$ \\ \hline
    \end{tabular}
    \label{tab:costconcentration}
    \begin{flushleft}
        Table~\ref{tab:costconcentration} illustrates that if a state is prepared by a non-unital noisy-random quantum circuit, the expectation value of local observables never exhibits exponential concentration at any depth around a fixed value. This stands in stark contrast to the noiseless and unital noise regimes.
    \end{flushleft}
\end{table}

\subsection{Variance upper bound: Global expectation values are exponentially concentrated}
\label{sub:global}

In this section, we show that expectation values of global observables are typically exponentially concentrated around their mean value.
As we have previously done, we consider a circuit model as described in Subsection~\ref{sub:circuitmodel}. We let
\begin{equation}\label{eq-noise_value}
   c \coloneqq \frac{1}{3}\left(\|\bold{t}\|^2_2+\|\bold{D}\|^2_2\right),
\end{equation}
where we recall that $\bold{t} \coloneqq (t_X,t_Y,t_Z)$ and $\bold{D} \coloneqq (D_X,D_Y,D_Z)$ are the local noise channel parameters of its normal form representation. We consider any circuit depth $L\ge 1$. From Lemma~$\ref{le:normal}$, we have that $c<1$ if and only if the channel is not unitary.

\begin{proposition}[Variance upper bound]
\label{prop:upvar}
Let $H \coloneqq \sum_{P\in \{I,X,Y,Z\}^{\otimes n}} a_P P$, with $a_P\in \mathbb{R}$ for any $P\in \{I,X,Y,Z\}^{\otimes n}$ be an arbitrary Hamiltonian. Let $\rho$ be any quantum state. Then, at any depth of the noisy circuit $\Phi$, as defined in Eq.~\eqref{eq:randcirc}, we have
    \begin{align}
       \Var[\Tr(H \Phi(\rho)  )] \le \sum_{P\in \{I,X,Y,Z\}^{\otimes n}\setminus I^{\otimes n}} a^2_P c^{|P|},
       \end{align}
       where the parameter $c$ is defined in
       Eq.~\eqref{eq-noise_value}.
\end{proposition}
\begin{proof}
The proof follows the same initial steps of the proof of Proposition~\ref{prop:lbvar}. In particular, the mean is $\Ex[\Tr(H \Phi(\rho)  )]=a_{I^{\otimes n}}$, and we have
\begin{align}
    \Ex[\Tr(H \Phi(\rho) )^2]= a^2_{I^{\otimes n}} +\sum_{P\in \{I,X,Y,Z\}^{\otimes n}\setminus I^{\otimes n}} a_P^2 \Ex[\Tr(P \Phi(\rho))^2],
\end{align}
with
\begin{align}
    \Ex[\Tr(P \Phi(\rho)  )^2]=\frac{1}{3^{|P|}}\sum_{\substack{Q \in \{I,X,Y,Z\}^{\otimes n}:\\ \supp(Q)=\supp(P) } }\Ex[\Tr(Q^{\otimes 2} \Phi^{\prime}(\rho)^{\otimes 2} )].
\end{align}
Therefore, by taking the adjoint of the last layer of noise, as done in the proof of Proposition~\ref{prop:lbvar} (specifically, Eq.~\eqref{eq:upstep}, we have 
\begin{align}
    \Ex\left[\Tr(Q \Phi^{\prime}(\rho) )^2\right]&=\sum_{a\in \{0,1\}^{|Q|}}\prod_{j \in \supp(Q)} (t^{a_j}_{Q_j } D^{1-a_j}_{Q_j})^2\Ex\left[\Tr(\left(\bigotimes_{j \in \supp(Q)}  Q^{1-a_j}_j\right) \Phi^{\prime\prime}(\rho) )^2\right].
\end{align}
Now, using the fact that $\Tr(P\sigma)\le 1$ for every Pauli $P$ and state $\sigma$, we have
\begin{align}
       \sum_{a\in \{0,1\}^{|Q|}}\prod_{j \in \supp(Q)} (t^{a_j}_{Q_j } D^{1-a_j}_{Q_j})^2\Ex\left[\Tr(\left(\bigotimes_{j \in \supp(Q)}  Q^{1-a_j}_j\right) \Phi^{\prime\prime}(\rho) )^2\right] &\le  \sum_{a\in \{0,1\}^{|Q|}}\prod_{j \in \supp(Q)}t^{2a_j}_{Q_j } D^{2(1-a_j)}_{Q_j}
    \\&=\prod_{j \in \supp(Q)} (t^2_{Q_j }+ D^{2}_{Q_j}).
    \nonumber
\end{align}
Thus, by substituting, we find
\begin{align}
    \Ex[\Tr(P \Phi(\rho) )^2 ]&\le\frac{1}{3^{|P|}}\sum_{\substack{Q \in \{I,X,Y,Z\}^{\otimes n}:\\ \supp(Q)=\supp(P) } }\prod_{j \in \supp(Q)} (t^2_{Q_j }+ D^{2}_{Q_j}) \\
    \nonumber
    &=\frac{1}{3^{|P|}} (t^2_{ X }+ D^{2}_{ X}+t^2_{ Y }+ D^{2}_{ Y}+t^2_{Z }+ D^{2}_{ Z})^{|P|} \label{eq:ub-pauli-purity}\\
    \nonumber
    &=c^{|P|},
\end{align}
where we have used the multinomial theorem.
\end{proof}
As an immediate corollary of the previous inequality, we find the following corollary.

\begin{corollary}[Global expectation values are exponentially concentrated on average]
\label{global expectation}
Let $P\in \{I,X,Y,Z\}^{\otimes n}$ be a Pauli operator with weight $|P|=\Theta(n)$. Then, at any depth of the noisy circuit $\Phi$ defined in Eq.~\eqref{eq:randcirc}, and for any constant noise parameters, we have
\begin{align}
    \Var[\Tr(P \Phi(\rho) )]= \exp(-\Theta(n)).
\end{align}
\end{corollary}
By combining the lower bound on the variance derived in the previous section (Proposition~\ref{prop:lbvar}) with the matching upper bound derived in this section (Proposition~\ref{prop:upvar}), we get a proof of Theorem~\ref{th:variance}.

\subsection{Effective shallow circuits}%: %Expectation values are typically influenced by only the last $O(\log(n))$ layers}
\label{sub:effective}

In the previous Subsection~\ref{sub:local}, we have shown that local expectation values can have a large variance in the presence of non-unital noise. In this subsection, we identify even more compelling consequences of this feature. We show that such large variance can only be due to the last few layers of the circuit. Specifically, we prove that the layers preceding the last $\Theta(\log(n))$ do not significantly affect observable expectation values. Let us start with proving the following proposition valid for Pauli expectation values.

%\begin{figure}[h]
%\centering
%\includegraphics[width=0.45\textwidth]{figures/veryeffective.pdf}
%\includegraphics[width=0.6\textwidth]{figures/Fig1Nonunital.png}
%\caption{For most quantum circuits with non-unital noise, only the last $O(\log(n))$ can influence expectation values significantly, which we formally state and prove in Proposition~\ref{prop:expdecayP}.}
%\label{Fig:circ2}
%\end{figure}

%Remember that we denote $c \coloneqq \frac{1}{3}\left(\|\bold{t}\|^2_2+\|\bold{D}\|^2_2\right)$, where $\bold{t}$ and $\bold{D}$ are the noise channel parameters, and we have that $c<1$ if the channel is not unitary.
\begin{proposition}
\label{prop:expdecayP}
Let $P\in \{I,X,Y,Z\}^{\otimes n}$, let $\rho$ and $\sigma$ be quantum states, and let $L$ be depth of the noisy circuit $\Phi$ defined in Eq.~\eqref{eq:randcirc}. Then, we have
    \begin{align}
       \Ex[\Tr(P\Phi(\rho-\sigma))^2] \le 4 c^{|P|+L-1},
       \end{align}
       where the parameter $c$ is defined in Eq.~\eqref{eq-noise_value}.
\end{proposition}
\begin{proof}
By removing the last layer of single qubits gates and using Lemma~\ref{le:mixham}, we have
\begin{align}
    \Ex[\Tr(P \Phi(\rho-\sigma)  )^2]=\frac{1}{3^{|P|}}\sum_{\substack{Q \in \{I,X,Y,Z\}^{\otimes n}:\\ \supp(Q)=\supp(P) } }\Ex[\Tr(Q \Phi^{\prime}(\rho-\sigma))^2].
\end{align}
Now using the exactly the same argument used in Eq.~\eqref{eq:upstep} in the proof of Proposition~\ref{prop:lbvar}, we obtain
\begin{align}
    \Ex\left[\Tr(Q \Phi^{\prime}(\rho-\sigma) )^2\right] & =\sum_{a\in \{0,1\}^{|Q|}}\prod_{j \in \supp(Q)} (t^{a_j}_{Q_j } D^{1-a_j}_{Q_j})^2\Ex\left[\Tr(\left(\bigotimes_{j \in \supp(Q)}  Q^{1-a_j}_j\right) \Phi^{\prime\prime}(\rho-\sigma) )^2\right].
\end{align}

Note that the expected value on the right-hand side can be bounded from above by %$\max_{Q \in \{I,X,Y,Z\}^{\otimes n}}\Ex[\Tr(Q^{\otimes 2} \Phi^{\prime\prime}(\rho-\sigma)^{\otimes 2} )]=
$\max_{Q \in \{I,X,Y,Z\}^{\otimes n}}\Ex[\Tr(Q \Phi^{\prime\prime}(\rho-\sigma))^2]$.
Thus, we have
\begin{align}
     \Ex[\Tr(P \Phi(\rho-\sigma)  )^2]&\le \frac{1}{3^{|P|}}\sum_{\substack{Q \in \{I,X,Y,Z\}^{\otimes n}:\\ \supp(Q)=\supp(P) } }\sum_{a\in \{0,1\}^{|Q|}}\prod_{j \in \supp(Q)} (t^{a_j}_{Q_j } D^{1-a_j}_{Q_j})^2\max_{Q \in \{I,X,Y,Z\}^{\otimes n}}\Ex[\Tr(Q \Phi^{\prime\prime}(\rho-\sigma))^2]\\
     \nonumber
     &= \frac{1}{3^{|P|}}\sum_{\substack{Q \in \{I,X,Y,Z\}^{\otimes n}:\\ \supp(Q)=\supp(P) } }\prod_{j \in \supp(Q)} (t^2_{Q_j } + D^2_{Q_j}) \max_{Q \in \{I,X,Y,Z\}^{\otimes n}}\Ex[\Tr(Q \Phi^{\prime\prime}(\rho-\sigma))^2]\\
       \nonumber
     &= \frac{1}{3^{|P|}}(\|\bold{D}\|^2_2+\|\bold{t}\|^2_2)^{|P|} \max_{Q \in \{I,X,Y,Z\}^{\otimes n}}\Ex[\Tr(Q \Phi^{\prime\prime}(\rho-\sigma))^2]\\
       \nonumber
     &=c^{|P|} \max_{Q \in \{I,X,Y,Z\}^{\otimes n}}\Ex[\Tr(Q \Phi^{\prime\prime}(\rho-\sigma))^2],
       \nonumber
\end{align}
where we have used the multinomial theorem.
Moreover, we can assume that the maximum over the Pauli operators is not achieved by the identity, otherwise the right-hand side of the above inequality would be zero, because $\Phi^{\prime\prime}$ is trace preserving and $\rho-\sigma$ is traceless. Thus, we have
\begin{align}
     \Ex[\Tr(P \Phi(\rho-\sigma)  )^2]&\le c^{|P|} \max_{Q \in \{I,X,Y,Z\}^{\otimes n}\setminus I_n}\Ex[\Tr(Q \Phi^{\prime\prime}(\rho-\sigma))^2].
     \label{eq:effniceeq}
\end{align}
We can assume now that all the two-qubit gates in the circuit are Clifford (see Definition~\ref{def:Clifford}), as we are computing a second moment and the Cliffords form a $2$-design (Lemma~\ref{le:Clifford3design}). Thus, the two qubit gates of the circuit will also map Paulis to Paulis.
Moreover, we assume without loss of generality that before each layer of noise there is a layer of single-qubit $2$-design unitaries, as we are computing a second moment over $2$-design quantities and we can use the invariance of the Haar measure to do so.
Therefore, the Pauli $Q\neq I_n$ above will be mapped by the two-qubits Clifford to another Pauli still different from the identity. Since now we have a circuit that ends with a layer of single qubits $2$-design unitaries, which are preceded by a noise layer and a layer of two-qubits $2$-design gates, we are in the same situation we faced at the beginning of the proof. So reiterating the argument to the next layer, we have
\begin{align}
    \Ex[\Tr(Q \Phi^{\prime\prime}(\rho-\sigma))^2]&\le   \max_{R \in \{I,X,Y,Z\}^{\otimes n}\setminus I_n}c^{|R|} \Ex\left[\Tr(R \Phi_{[1,L-1]}^{\prime\prime}(\rho-\sigma))^2\right]\\
    \nonumber
    &\le   c \max_{R \in \{I,X,Y,Z\}^{\otimes n}\setminus I_n}\Ex\left[\Tr(R \Phi_{[1,L-1]}^{\prime\prime}(\rho-\sigma))^2\right],
    \nonumber
\end{align}
where we have used the notation $\Phi_{[1,k]}^{\prime\prime}  \coloneqq  \mathcal{U}_{k} \circ \cdots  \circ \mathcal{N}^{\otimes n} \circ \mathcal{U}_{1}$ and used the fact that the Pauli weight of $R$ is at least one. 

Recursively applying the above reasoning to all of the remaining layers of the circuit, and using the fact that for any Pauli operator $P$ we have $|\Tr(P(\rho-\sigma))|\le \norm{\rho-\sigma}_{1}$ (because of the H\"older inequality), we obtain 
\begin{align}
    \Ex[\Tr(Q \Phi^{\prime\prime}(\rho-\sigma))^2]&\le c^{L-1} \norm{\rho-\sigma}^2_{1}\le 4 c^{L-1},
\end{align}
where in the last step we have used using triangle inequality and the fact that quantum states have one-norm equal to one.
Substituting back in Eq.~\eqref{eq:effniceeq}, we conclude the proof.
\end{proof} 

Proposition~\ref{prop:expdecayP} implies that if $\rho$ and $\sigma$ are states created by the `first' part of the same circuit architecture with different parameters of the gates (e.g., the red part of the circuit in Fig.~\ref{Fig:circ2}), then if we implement on them a noisy quantum circuit of depth $L=\omega(\log(n))$, on average we have that the influence on the expectation value of the different gates in the first part of the circuit will be super-polynomially small.

We now show that the previous bounds for Pauli expectation values imply a bound for any observable \( O \).

\begin{proposition}
\label{prop:expdecayO}
Let \( O \) be an observable, \( \rho \) and \( \sigma \) be quantum states, and let \( L \) be the depth of the noisy circuit \( \Phi \). Then, we have
\begin{align}
    \Ex[|\Tr(O \Phi(\rho)) - \Tr(O \Phi(\sigma))|^2] \le 4 \left(\frac{\|O\|^2_2}{2^n} \right) c^L \le 4 \|O\|^2_\infty c^L,
\end{align}
where the parameter \( c \) is defined as in Eq.~\eqref{eq-noise_value}.
\end{proposition}

\begin{proof}
Let \( O = \sum_{P \in \{I, X, Y, Z\}^{\otimes n}} c_P P \) be the Pauli decomposition of \( O \). Then, we find
\begin{align}
    \Ex[\Tr(O \Phi(\rho - \sigma))^2] &= \Ex \left[ \left( \sum_{P \in \{I, X, Y, Z\}^{\otimes n}} c_P \Tr(P \Phi(\rho - \sigma)) \right)^2 \right] \\
    &= \sum_{P \in \{I, X, Y, Z\}^{\otimes n}} |c_P|^2 \Ex[\Tr(P \Phi(\rho - \sigma))^2],\nonumber
\end{align}
where in the second step we have used the fact that the cross-terms vanish due to Lemma~\ref{le:mixham}.

Using the bound from Proposition~\ref{prop:expdecayP}, we have
\begin{align}
    \Ex[\Tr(O \Phi(\rho - \sigma))^2] &\le \sum_{P \in \{I, X, Y, Z\}^{\otimes n}} |c_P|^2 \cdot 4 c^L \\
    &= \frac{\|O\|^2_2}{2^n} 4 c^L,
    \nonumber
\end{align}
where we have used the fact that \(\sum_{P \in \{I, X, Y, Z\}^{\otimes n}} |c_P|^2 = \frac{\|O\|^2_2}{2^n}\).
Finally, using the inequality \(\|O\|_2 \le \sqrt{2^n} \|O\|_\infty\), we get
\begin{align}
    \Ex[\Tr(O \Phi(\rho - \sigma))^2] \le 4 \|O\|^2_\infty c^L.
\end{align}
\end{proof}

\begin{corollary}
Let \( O \) be an observable, and let \( \rho \) and \( \sigma \) be quantum states. Consider a noisy circuit \( \Phi \) with depth \( L \). Then, the  bound 
\begin{align}
    \mathbb{E}_{\Phi}\left[\left\lvert \Tr(O \Phi(\rho)) - \Tr(O \Phi(\sigma)) \right\rvert\right] \leq 2 \left(\frac{\|O\|_2}{\sqrt{2^n}}\right) c^{L/2} \leq 2 \|O\|_\infty c^{L/2}
\end{align}
holds, where the parameter \( c \) is defined as in Eq.~\eqref{eq-noise_value}.
\end{corollary}
This corollary immediately proves Theorem~\ref{th:lastlogMAIN} in the main text through a straightforward relabeling of terms.

\textcolor{black}{Notably, from the proof of the previous proposition, it is evident that the assumption that the circuit must terminate with a layer of single-qubit random gates is not essential. The circuit could instead conclude with a layer of noise, as this noise can be effectively absorbed into the observable within the Heisenberg picture.} %Consequently, the previous corollary would still apply if the circuit concludes with a layer of noise.}

\subsection{Indistinguishability of quantum states affected by noisy quantum circuits}
\label{sub:tracedistance}

We now translate the results in the previous section in terms of the trace distance.

\subsubsection{Indistinguishability 
in terms of the trace distance}

\begin{proposition}[Average distance between 
two quantum states]
Let $\Phi$ be a noisy random quantum circuit with of depth $L$, as defined in Eq.~\eqref{eq:randcirc}. Then, the average trace distance between $\Phi(\rho)$ and $\Phi(\sigma)$, where $\rho$ and $\sigma$ are two arbitrary quantum states, decays exponentially in $L$ as
 \begin{align}
      \Ex[\norm{\Phi(\rho) - \Phi(\sigma)}_1] \le  2^{n+1} c^{\frac{L-1}{2}},
 \end{align}
 where we recall the definition of the parameter $c$ in Eq.~\eqref{eq-noise_value}. Thus, for any $\varepsilon>0$, assuming that $L\ge\frac{1}{\log(c^{-1})}\Omega\!\left(n+\log(\frac{1}{\varepsilon})\right)$, we have that $\Ex[\norm{\Phi(\rho) - \Phi(\sigma)}_1]\le \varepsilon$.
\end{proposition}
\begin{proof}
We have
\begin{align}
       (\Ex[\norm{\Phi(\rho) - \Phi(\sigma)}_1])^2&\le \Ex[\norm{\Phi(\rho) - \Phi(\sigma)}^2_1]\\
       \nonumber
       &\le 2^n\, \Ex[\norm{\Phi(\rho) - \Phi(\sigma)}^2_2] \\
       \nonumber
       &=\sum_{P\in \{I,X,Y,Z\}^{\otimes n}}  \Ex[\Tr(P(\Phi(\rho)-\Phi(\sigma)))^2]\\
       \nonumber
       &\le 4 \sum_{P\in \{I,X,Y,Z\}^{\otimes n}}   c^{|P|+L-1}\\
       \nonumber
      % &=4 c^{L-1}\sum^{n}_{k=0} {n\choose k} 3^k c^k\\
       &=4(1+3c)^{n}c^{L-1} .
       \nonumber
\end{align}
In the first step, we have used Jensen's inequality, in the second step we have used the fact that $\norm{\cdot}_1\leq 2^{\frac{n}{2}}\norm{\cdot}_2$, and in the third step we expressed $\Phi(\rho)-\Phi(\sigma)$ in the Pauli basis and used the fact that $\norm{A}^2_2=\Tr(A^{\dagger} A)$ for any matrix $A$. This, in particular, implies that 
\begin{equation}\norm{A}_2^2=\frac{1}{2^n}\sum_{P\in\{I,X,Y,Z\}^{\otimes n}}\Tr(PA)^2 
\end{equation}
for every Hermitian matrix $A$. Then, in the fourth step, we have used Proposition~\ref{prop:expdecayP}, and in the final step the binomial theorem. Finally, from the fact that $1+3c\leq 4$, on account of the fact that $c\leq 1$, we obtain
\begin{align}
    \Ex[\norm{\Phi(\rho) - \Phi(\sigma)}_1]\le 2(1+3c)^{\frac{n}{2}} c^{\frac{L-1}{2}}\le  2^{n+1} c^{\frac{L-1}{2}} .
\end{align}
The right-most expression in the above chain of inequalities is bounded from above by $\varepsilon$ if
\begin{align}
    L\ge \frac{1}{\log(c^{-1})}\left(2n+2\log(\frac{1}{\varepsilon})+3\right),
\end{align}
which implies the desired result.
\end{proof}

The previous proposition implies that for most of noisy circuit of depth larger than $L\ge \Omega(n)$, the trace distance between the two output states is bounded from above by $\exp(-\Theta(n))$, which means that the two output states cannot be distinguished between each other efficiently by performing arbitrary measurements on polynomially many copies of the state, because of the Holevo-Helstrom theorem~\cite{wilde_2013}.
\textcolor{black}{It seems plausible that the previous assumption on the depth \( L \geq \Omega(n) \) is an artifact of our proof technique, and we conjecture that this assumption might be relaxed to smaller depths with a more fine-grained analysis; we leave this for future work.}

One may wonder if it is possible to prove an upper bound on the worst-case trace distance (i.e., with no expected value) that decreases exponentially with the number of layers. \textcolor{black}{Although this is possible by making strong structural assumption on the noise, namely being unital and with the maximally mixed state being their unique fixed point (e.g., depolarizing noise)~\cite{aharonov1996limitations,DanielPaper}; this is not possible for arbitrary noise regime in general.
In fact, as an easy example, if the noise is dephasing, the circuit is made only by Toffoli gates, and the input states $\rho$ and $\sigma$ are two different computational basis states, then the trace distance between the two output states must remain constant for any depth (since Toffoli maps computational basis states in computational basis states and dephasing noise acts trivially on computational basis states).}
More interestingly, the so-called \emph{quantum refrigerator} construction~\cite{refrigerator} shows surprisingly how non-unital noise can be exploited to perform fault-tolerant quantum computations in a model similar to ours, up to exponential depth. Therefore, for these special classes of circuits, the trace distance remains of constant order. %(assuming the noiseless version of the circuit simply implements the identity).  
However, we show below that in a certain high noise regime, we can find a worst-case upper bound on the trace distance that decays exponentially in the number of qubits.% as we show below. 

\subsubsection{Worst-case upper bound on the trace distance}
In this section we give a worst-case bound for the trace distance that holds whenever the noise parameters exceed certain thresholds.
We first introduce some technical tools before proving our worst-case trace distance upper bound. Our argument is based on the contraction coefficients of the quantum Wasserstein distance of order 1 ($W_1$ distance)~\cite{de2021quantum}.

\begin{comment}
Following~\cite{de2021quantum}, we define the quantum Wasserstein distance of order 1 (or quantum $W_1$ distance) between the $n$-qubit quantum states $\rho$ and $\sigma$ as
\begin{equation}
    W_1(\rho,\sigma) = \min \left( \sum_{i=1}^n c_i : c_i \geq 0, \rho -\sigma =\sum_{i=1}^n c_i \left(\rho^{(i)} - \sigma^{(i)}\right), \rho^{(i)},\sigma^{(i)}   \in \mathcal{S}\left(\mathbb{C}^{2^n}\right)\right).
\end{equation}
The $W_1$ distance comes with a dual norm, referred as the quantum Lipschitz constant. For an observable $H$, we define 
\begin{align}
    \|H\|_L:&=\max_{i \in [n]}\left(\max \left( \Tr[H(\rho - \sigma)]: \rho ,\sigma \in \mathcal{S}\left(\mathbb{C}^{2^n}\right), \Tr_i \rho = \Tr_i \sigma\right)\right)
    \\
    \nonumber &= \max \left(\Tr[H(\rho - \sigma)]: \rho ,\sigma \in \mathcal{S}\left(\mathbb{C}^{2^n}\right), W_1(\rho,\sigma)\right).
\end{align}
\end{comment}

Let $O^T_n \subset \mathcal{L}(\mathbb{C}^{2^n})$ be the subset of traceless self-adjoint linear operators. The $W_1$ distance is induced by the quantum $W_1$ norm, which is defined as follows~\cite{de2021quantum}:
\begin{equation}
    \|X\|_{W_1} = \frac{1}{2}\min\left\{\sum_{i=1}^n \|X^{(i)}\|_1 : X^{(i)} \in O^T_n , \Tr_i X^{(i)} = 0, X= \sum_{i=1}^n X^{(i)}\right\}.
\end{equation}
Hence, for two arbitrary states $\rho,\sigma$, the $W_1$ distance is defined as
\begin{equation}
   W_1(\rho,\sigma) \coloneqq  \|\rho-\sigma\|_{W_1}.
\end{equation}
The quantum $W_1$ norm and the trace norm are always within a factor of $n$, since
\begin{equation}
   \frac{1}{2} \|X\|_{1} \leq \|X\|_{W_1} \leq \frac{n}{2}\|X\|_{1}.
\end{equation}
%The $W_1$ distance comes with a dual norm, referred to as the quantum Lipschitz constant. For an observable $H$, we define 
%\begin{align}
 %   \|H\|_L:&=\max_{i \in [n]}\left(\max \left( \Tr[H(\rho - \sigma)]: \rho ,\sigma \in \mathcal{S}\left(\mathbb{C}^{2^n}\right), \Tr_i \rho = \Tr_i \sigma\right)\right)
  %  \\
   % \nonumber &= \max \left(\Tr[H(\rho - \sigma)]: \rho ,\sigma \in \mathcal{S}\left(\mathbb{C}^{2^n}\right), \|\rho - \sigma\|_{W_1}\leq 1\right).
%\end{align}
%\antonio{Is this above useful later?}
We will employ the \emph{contraction coefficient} of a channel $\Phi$ with respect to the quantum $W_1$ distance, defined as
\begin{equation}
    \|\Phi\|_{W_1\rightarrow W_1} \coloneqq \max_{\rho \neq\sigma \in \mathcal{S}\left(\mathbb{C}^{2^n}\right)} \frac{\|\Phi(\rho) - \Phi(\sigma)\|_{W_1}}{\|\rho-\sigma\|_{W_1}} = \max_{\substack{X\in O^T_n,\\\|X\|_{W_1}=1}} \|\Phi(X)\|_{W_1}.
\end{equation}
The contraction coefficient is not in general bounded by 1, as the $W_1$ does not satisfy a data-processing inequality for all channels. Importantly, as showed in Ref.~\cite{de2021quantum}, if $\Phi$ is a layer of $k$-qubit gates, the contraction coefficient of $\Phi$ can be bounded by light-cone argument as follows
\begin{equation}
\label{eq:contraction}
    \|\Phi\|_{W_1 \rightarrow W_1} \leq \begin{cases}
         1 &\text{ if  $k = 1$}, \\
         \frac{3}{2}k &\text{ if $k>1$ } \text{ (\cite{de2021quantum}, Proposition 13)}.
    \end{cases}
\end{equation}
And thus a layer of two qubit gates has contraction coefficient at most $3$. 
%Moreover, $\|\mathcal{N}^{(\mathrm{dep})\otimes n}_p \|_{W_1 \rightarrow W_1} = 1-p$, which readily implies (\cite{hirche2023quantum}, Proposition IV.8).
If $\mathcal{N}$ is a single-qubit channel, the contraction coefficient of the tensor power channel $\mathcal{N}^{\otimes n}$ can be upper bounded by the diamond distance between $\mathcal{N}$ and a suitable 1-qubit channel $\mathcal{E}$~\cite{de2021quantum}, as follows. %In particular, we recall a result given in Ref.~\cite{de2021quantum}.
\begin{proposition}[Proposition 11,~\cite{de2021quantum}]
\label{prop:depalma}
Let $\Phi$ be a single qubit quantum channel with fixed point a quantum state $\tau$ and let $\mathcal{E}$ the single-qubit quantum channel that replaces any state with $\tau$. Then,
\begin{equation}
   \frac{1}{2}\|\Phi-\mathcal{E}\|_{1 \rightarrow 1} \leq \|\Phi^{\otimes n}\|_{W_1 \rightarrow W_1} \leq \|\Phi - \mathcal{E}\|_\diamond \leq 2 \|\Phi-\mathcal{E}\|_{1 \rightarrow 1},
\end{equation}
where we recall that for any single-qubit Hermitian-preserving linear map $\mathcal{F}$,
\begin{align}
    &\|\mathcal{F}\|_{1 \rightarrow 1} = \max_{\rho\in\mathcal{S}(\mathbb{C}^2)}\|\mathcal{F}(\rho)\|_1,\\
    &\|\mathcal{F}\|_\diamond = \max_{\rho\in\mathcal{S}(\mathbb{C}^2\otimes\mathbb{C}^2)}\|\mathcal{F}\otimes \mathcal{I}(\rho)\|_1.
\end{align}
\end{proposition}
By exploiting the above result, we give an explicit upper bound of the contraction coefficient in terms of the parameters of the noise channel $\mathcal{N}$ expressed in the normal form. 
We remark that the adoption of the normal form comes without loss of generality: as discussed in Section\ \ref{normalchannel}, we can always write a single-qubit channel $\mathcal{M}$ as $\mathcal{M}(\cdot) = U\mathcal{N}(V^\dag(\cdot)V)U^\dag$, where $U,V$ are suitable single-qubit unitaries. By Eq.~\eqref{eq:contraction}, $U$ and $V$ do not alter the $W_1$ norm, thus they can be neglected in our analysis.

\begin{lemma}
\label{lem:contraction-coeff}
Let $\mathcal{N}$ be a single-qubit channel in normal form,
\begin{align}
\mathcal{N}\!\left(\frac{I+\boldsymbol{w}\cdot\boldsymbol{\sigma}}{2}\right)
=
\frac{I+(\boldsymbol{t}+D\boldsymbol{w})\cdot\boldsymbol{\sigma}}{2},
\qquad
D=\operatorname{diag}(D_X,D_Y,D_Z),
\end{align}
and assume that \(D_P\neq 1\) for every \(P\in\{X,Y,Z\}\). Then \(\mathcal N\) has the unique fixed point
\begin{align}
\tau=\frac I2+\sum_{P\in\{X,Y,Z\}}\frac{t_P}{2(1-D_P)}P .
\end{align}
Let \(\mathcal E_\tau\) be the replacer channel \(\mathcal E_\tau(X):=\operatorname{Tr}(X)\tau\), and define
\begin{align}
M_{\mathcal N}:=\max_{P\in\{X,Y,Z\}}\left|\frac{D_P}{1-D_P}\right|.
\end{align}
Then
\begin{align}
\|\mathcal N^{\otimes n}\|_{W_1\to W_1}
&\le
\|\mathcal N-\mathcal E_\tau\|_\diamond
\nonumber\\
&\le
2M_{\mathcal N}\|\mathcal N-\mathcal I\|_{1\to 1}
\le
4M_{\mathcal N}.
\label{eq:w1-contraction-coeff-normal-form}
\end{align}
Here \(\|\cdot\|_{1\to 1}\) denotes the state-restricted single-qubit norm used in Proposition~\ref{prop:depalma}, namely
\(\|\mathcal F\|_{1\to 1}:=\max_{\rho\in\mathcal S(\mathbb C^2)}\|\mathcal F(\rho)\|_1\).
\end{lemma}

\begin{proof}
A fixed point \(\tau=(I+\boldsymbol{\omega}\cdot\boldsymbol{\sigma})/2\) must satisfy
\(\boldsymbol{t}+D\boldsymbol{\omega}=\boldsymbol{\omega}\), i.e.,
\(t_P+D_P\omega_P=\omega_P\) for every \(P\in\{X,Y,Z\}\). Since \(D_P\neq 1\) for every \(P\), this affine fixed-point equation has the unique solution
\begin{align}
\omega_P=\frac{t_P}{1-D_P},
\end{align}
which gives
\begin{align}
\tau=\frac I2+\sum_{P\in\{X,Y,Z\}}\frac{t_P}{2(1-D_P)}P .
\end{align}
Existence of a fixed point follows because \(\mathcal N\) maps the compact convex Bloch ball into itself. Hence the above solution is a valid quantum state and is the unique fixed point of \(\mathcal N\).

By Proposition~\ref{prop:depalma}, applied to the fixed point \(\tau\), we have
\begin{align}
\|\mathcal N^{\otimes n}\|_{W_1\to W_1}
\le
\|\mathcal N-\mathcal E_\tau\|_\diamond .
\label{eq:depalma-application}
\end{align}
It remains to bound \(\|\mathcal N-\mathcal E_\tau\|_\diamond\) in terms of the normal-form coefficients. Let
\(\rho=I/2+\frac12\sum_{P\in\{X,Y,Z\}}w_PP\), with \(\|\boldsymbol w\|_2\le 1\). For every traceless Hermitian single-qubit operator \(A=(\boldsymbol a\cdot\boldsymbol\sigma)/2\), one has \(\|A\|_1=\|\boldsymbol a\|_2\). Therefore,
\begin{align}
\|\mathcal N(\rho)-\tau\|_1
&=
\left[
\sum_{P\in\{X,Y,Z\}}
\left(t_P+D_Pw_P-\frac{t_P}{1-D_P}\right)^2
\right]^{1/2}
\nonumber\\
&=
\left[
\sum_{P\in\{X,Y,Z\}}
\left(\frac{D_P}{1-D_P}\right)^2
\left(w_P-(t_P+D_Pw_P)\right)^2
\right]^{1/2}
\nonumber\\
&\le
M_{\mathcal N}
\left[
\sum_{P\in\{X,Y,Z\}}
\left(w_P-(t_P+D_Pw_P)\right)^2
\right]^{1/2}
\nonumber\\
&=
M_{\mathcal N}\|\rho-\mathcal N(\rho)\|_1 .
\label{eq:normal-form-to-id-distance}
\end{align}
Taking the maximum over single-qubit states gives
\begin{align}
\|\mathcal N-\mathcal E_\tau\|_{1\to 1}
\le
M_{\mathcal N}\|\mathcal N-\mathcal I\|_{1\to 1}.
\label{eq:one-to-one-bound-replacer}
\end{align}
Using again Proposition~\ref{prop:depalma},
\begin{align}
\|\mathcal N-\mathcal E_\tau\|_\diamond
&\le
2\|\mathcal N-\mathcal E_\tau\|_{1\to 1}
\nonumber\\
&\le
2M_{\mathcal N}\|\mathcal N-\mathcal I\|_{1\to 1}.
\label{eq:diamond-bound-replacer}
\end{align}
Finally, \(\|\mathcal N(\rho)-\rho\|_1\le 2\) for every state \(\rho\), and hence
\begin{align}
\|\mathcal N-\mathcal I\|_{1\to 1}\le 2 .
\end{align}
Combining Eqs.~\eqref{eq:depalma-application}--\eqref{eq:diamond-bound-replacer} proves the claim.
\end{proof}

The above bound applies to any local noise channel expressed in normal form and satisfying \(D_P\neq 1\) for all \(P\in\{X,Y,Z\}\). It should not be applied to channels with a conserved Bloch direction, such as pure dephasing, for which the fixed point is not unique. This worst-case contraction statement complements the average-case effective-depth result of the previous section. The following proposition extends to the non-unital case a result, Proposition IV.8 of Ref.~\cite{hirche2023quantum}, proven there for local depolarizing noise.

\begin{proposition}
\label{prop:worst-case-w1-contraction}
Let \(\mathcal N\) be a single-qubit channel in normal form,
\begin{align}
\mathcal{N}\!\left(\frac{I+\boldsymbol{w}\cdot\boldsymbol{\sigma}}{2}\right)
=
\frac{I+(\boldsymbol{t}+D\boldsymbol{w})\cdot\boldsymbol{\sigma}}{2},
\end{align}
and assume that \(D_P\neq 1\) for every \(P\in\{X,Y,Z\}\). Let \(\tau\) be the unique fixed point of \(\mathcal N\), let \(\mathcal E_\tau(X):=\operatorname{Tr}(X)\tau\), and define
\begin{align}
b_{\mathcal N}:=3\|\mathcal N-\mathcal E_\tau\|_\diamond .
\end{align}
Let \(\Phi=\mathcal L_L\circ\cdots\circ\mathcal L_1\) be a noisy quantum circuit of depth \(L\), where each noisy layer \(\mathcal L_j\) consists of a two-qubit unitary layer, a local noise layer \(\mathcal N^{\otimes n}\), and possibly a single-qubit unitary layer. Then, for any two quantum states \(\rho,\sigma\),
\begin{align}
\|\Phi(\rho)-\Phi(\sigma)\|_{W_1}
&\le
b_{\mathcal N}^L\|\rho-\sigma\|_{W_1},
\label{eq:contraction-w1}
\\
\|\Phi(\rho)-\Phi(\sigma)\|_1
&\le
n\,b_{\mathcal N}^L\|\rho-\sigma\|_1 .
\label{eq:contraction-tr}
\end{align}
Thus, if \(b_{\mathcal N}<1\), then
\begin{align}
L\ge \frac{\log(2n/\varepsilon)}{\log(b_{\mathcal N}^{-1})}
\qquad\Longrightarrow\qquad
\|\Phi(\rho)-\Phi(\sigma)\|_1\le \varepsilon
\end{align}
for all quantum states \(\rho,\sigma\).

Moreover, a sharper sufficient condition for \(b_{\mathcal N}<1\) is
\begin{align}
6M_{\mathcal N}\|\mathcal N-\mathcal I\|_{1\to 1}<1,
\qquad
M_{\mathcal N}:=\max_{P\in\{X,Y,Z\}}\left|\frac{D_P}{1-D_P}\right|.
\end{align}
Since \(\|\mathcal N-\mathcal I\|_{1\to 1}\le 2\), the purely coefficient-only condition
\begin{align}
12M_{\mathcal N}<1
\end{align}
also suffices. If \(0\le D_P<1\) for all \(P\in\{X,Y,Z\}\), this simpler condition is equivalent to
\begin{align}
D_P<\frac1{13}
\qquad\text{for all }P\in\{X,Y,Z\}.
\end{align}
\end{proposition}

\begin{proof}
Let \(\mathcal U=U(\cdot)U^\dagger\) be a layer of two-qubit unitaries. By Eq.~\eqref{eq:contraction}, equivalently by Proposition~13 of Ref.~\cite{de2021quantum}, we have
\begin{align}
\|\mathcal U\|_{W_1\to W_1}\le 3.
\end{align}
A layer of single-qubit unitaries has \(W_1\)-contraction coefficient at most \(1\). Hence, for one noisy layer
\(\mathcal L=\mathcal V^{\rm single}\circ\mathcal N^{\otimes n}\circ\mathcal U\), submultiplicativity and Lemma~\ref{lem:contraction-coeff} give
\begin{align}
\|\mathcal L\|_{W_1\to W_1}
&\le
\|\mathcal V^{\rm single}\|_{W_1\to W_1}
\|\mathcal N^{\otimes n}\|_{W_1\to W_1}
\|\mathcal U\|_{W_1\to W_1}
\nonumber\\
&\le
1\cdot \|\mathcal N-\mathcal E_\tau\|_\diamond\cdot 3
=
b_{\mathcal N}.
\label{eq:single-layer-contraction}
\end{align}
The same bound holds for any convention in which the single-qubit unitary layer, the local noise layer, and the two-qubit unitary layer are ordered differently, since the same contraction coefficients appear in the product. Iterating over the \(L\) noisy layers gives
\begin{align}
\|\Phi\|_{W_1\to W_1}\le b_{\mathcal N}^L,
\end{align}
which proves Eq.~\eqref{eq:contraction-w1}.

To pass to trace distance, we use
\begin{align}
\frac12\|X\|_1\le \|X\|_{W_1}\le \frac n2\|X\|_1 .
\end{align}
Therefore,
\begin{align}
\|\Phi(\rho)-\Phi(\sigma)\|_1
&\le
2\|\Phi(\rho)-\Phi(\sigma)\|_{W_1}
\nonumber\\
&\le
2b_{\mathcal N}^L\|\rho-\sigma\|_{W_1}
\nonumber\\
&\le
n b_{\mathcal N}^L\|\rho-\sigma\|_1,
\end{align}
which proves Eq.~\eqref{eq:contraction-tr}. Since \(\|\rho-\sigma\|_1\le 2\), the choice
\begin{align}
L\ge \frac{\log(2n/\varepsilon)}{\log(b_{\mathcal N}^{-1})}
\end{align}
implies \(\|\Phi(\rho)-\Phi(\sigma)\|_1\le\varepsilon\).

It remains to verify the sufficient conditions. Lemma~\ref{lem:contraction-coeff} gives
\begin{align}
\|\mathcal N-\mathcal E_\tau\|_\diamond
\le
2M_{\mathcal N}\|\mathcal N-\mathcal I\|_{1\to 1},
\end{align}
so
\begin{align}
b_{\mathcal N}
=
3\|\mathcal N-\mathcal E_\tau\|_\diamond
\le
6M_{\mathcal N}\|\mathcal N-\mathcal I\|_{1\to 1}.
\end{align}
Thus \(6M_{\mathcal N}\|\mathcal N-\mathcal I\|_{1\to 1}<1\) implies \(b_{\mathcal N}<1\). Since \(\|\mathcal N-\mathcal I\|_{1\to 1}\le 2\), the simpler condition \(12M_{\mathcal N}<1\) also suffices. Finally, if \(0\le D_P<1\), then \(D_P/(1-D_P)<1/12\) is equivalent to \(D_P<1/13\). This completes the proof.
\end{proof}

In particular, if \(b_{\mathcal N}<1\) is a constant independent of \(n\), then
\begin{align}
\|\Phi(\rho)-\Phi(\sigma)\|_1
\le
n\,2^{-\Omega(L)}\|\rho-\sigma\|_1 .
\end{align}
Since \(\|\rho-\sigma\|_1\le 2\) for quantum states, this also gives
\begin{align}
\|\Phi(\rho)-\Phi(\sigma)\|_1
\le
n\,2^{-\Omega(L)}
\end{align}
up to an irrelevant constant factor.

\subsection{Classical simulation of Pauli expectation values of noisy random quantum circuits}
\label{sub:simu}

\textcolor{black}{In this section, we address the problem of estimating Pauli expectation values of noisy random quantum circuits under arbitrary noise. Specifically, given an instance of a noisy circuit $\Phi$ in which the two-qubit gates are sampled uniformly at random within a fixed architecture, our goal is to estimate $\Tr(P\Phi(\rho_0))$ to accuracy $\varepsilon$ with high probability over the choice of the random circuit, where $P$ is a prescribed Pauli operator and $\rho_0$ is an initial state. Without loss of generality, we focus on the case where the Pauli operator is local, since contributions from large Pauli weights are exponentially suppressed and, therefore, negligible. }

We have seen that the presence of any non-unitary noise in the circuit renders the circuit effectively shallow for the purpose of estimating expectation values. In particular, for any inverse polynomial precision, the last logarithmically-many layers suffice. Specifically, a direct consequence of Proposition~\ref{prop:expdecayP} implies the following:

\begin{corollary}[Effective-depth picture]
\label{cor:lastlog}
    \textcolor{black}{Let $P\in \{I,X,Y,Z\}^{\otimes n}$, and $\rho_0$ any input state. Let $L$ be depth of the noisy circuit $\Phi$. Then, we have}
    \begin{align}
       \Ex_{\Phi_{[L-m,L]}}[|\Tr(P\Phi(\rho_0))-\Tr(P\Phi_{[L-m,L]}(\sigma_0))|^2] \le 4 c^{|P|+m-1},
       \end{align} 
   where $\sigma_0$ is any preferred state (e.g., $\sigma_0 \coloneqq \ketbra{0^n}{0^n}$).
   Here, $\Phi_{[L-m,L]}(\cdot)$ refers to the noisy circuit where only the last $m$ layers are considered.
\end{corollary}
%So far, we did not make any assumptions on the locality of the circuits. However, now we assume that the circuit architecture is geometrically local, with constant dimension $\mathrm{D}$.
 \textcolor{black}{This yields the following simple algorithm for estimating local expectation values: work in the Heisenberg picture and `propagate' the local Pauli $P$ only a few number of layers backwards; compute classically the matrix $P_m \coloneqq \Phi^{*}_{[L-m,L]}(P)$, and then evaluate $\Tr(P_m \sigma_0)$. Due to standard light-cone arguments, for any product state $\sigma_0$ (e.g., $\sigma_0 \coloneqq \ketbra{0^n}{0^n}$), the time-complexity of this algorithm is exponential in the number of qubits over which $P_m$ is supported. If the circuit is a $\mathrm{D}$-dimensional geometrical local circuit, $P_m$ is supported on at most $|P| (2m)^{\mathrm{D}}$. If the circuit architecture instead does not possess any geometrical locality, i.e., it has all-to-all connectivity, then $P_m$ is supported on at most $|P| 2^m$. Thus, for $\mathrm{D}$-dimensional geometrical local circuit architectures, this algorithm incurs a total time complexity bounded from above by $\exp(O(|P| m^{\mathrm{D}}))$, while for all-to-all connected architectures, the time complexity is $\exp(O(|P| 2^m ))$. We refer to Sec.~\ref{sec-definitions} for a formal definition of the light-cone of an observable and geometrical locality. }

\begin{proposition}[Average classical simulation of local expectation values]
\label{prop:classim}
 \textcolor{black}{Let $\varepsilon, \delta > 0$. Consider a Pauli operator $P$ and any initial state $\rho_0$. For a noisy quantum circuit $\Phi$ of depth $L$, sampled according to the described circuit distribution, there exists a classical algorithm that outputs a value $\hat{C}$ satisfying}
\begin{align}
	| \hat{C} - \Tr(P \Phi(\rho_0)) | \leq \varepsilon
\end{align}
with success probability at least $1 - \delta$ over the choice of the random circuit. Specifically, the classical algorithm involves computing $\hat{C} \coloneqq \Tr(P\Phi_{[L-m,L]}(\ketbra{0^n}{0^n}))$ with
\begin{align}
	m \coloneqq \left\lceil \frac{1}{\log(c^{-1})} \log\left(\frac{4}{\delta \varepsilon^2}\right)\right\rceil,
\end{align}
where $c$ is the noise parameter defined in Lemma~\ref{le:normal}.

The time complexity of this algorithm is given by:
\begin{align}
	\text{Runtime} &\leq 
	\begin{cases}
		\exp\!\left(O\!\left(|P| m^{\mathrm{D}}\right)\right) = \exp\!\left(O\!\left(\log^{\mathrm{D}}(\varepsilon^{-1})\right)\right), & \text{for $\mathrm{D}$-geometrically-local architectures}, \\[1ex]
		\exp\!\left(|P| \exp(O(m))\right) = \exp\!\left(\mathrm{poly}(\varepsilon^{-1})\right), & \text{for all-to-all connected architectures}.
	\end{cases}
	\label{eq:runt}
\end{align}
where in the last equation we assumed constant noise rate $c$, constant failure probability $\delta$ and $|P|=O(1)$.
\end{proposition}

\begin{proof}
Because of the Markov inequality, we have
\begin{align}
        \mathrm{Prob}\left(|\Tr(P\Phi(\rho_0))-\Tr(P\Phi_{[L-m,L]}(\rho_0))| > \varepsilon \right)%&= \mathrm{Prob}\left(|\Tr(P\Phi(\rho_0))-\Tr(P\Phi_{[L-m,L]}(\rho_0))|^2 > \varepsilon^2 \right)\\
        &\le \frac{1}{\varepsilon^2}\Ex[|\Tr(P\Phi(\rho_0))-\Tr(P\Phi_{[L-m,L]}(\rho_0))|^2] \\
        \nonumber
        &\le \frac{4}{\varepsilon^2} c^{m},
\end{align}
where we have used Corollary~\ref{cor:lastlog} with $|P|\ge 1$. The right-hand side of this inequality is at most $ \delta$ if we choose 
\begin{equation}
m = \lceil \frac{1}{\log(c^{-1})}\log(\frac{4}{\delta \varepsilon^2})\rceil.
\end{equation} The algorithm consists of computing classically the matrix $P_m\coloneqq \Phi^{*}_{[L-m,L]}(P)$ and then evaluate $\Tr(P_m \rho_0)$. As we have previously described, this can be done via standard light-cone arguments, with a time complexity exponential in the number of qubits over which $P_m$ is supported. 
\end{proof}
%Specifically, 
%where we explicitly wrote the dependence on the noise parameter $c$.

\begin{figure}
\centering
\includegraphics[width=0.63\textwidth]{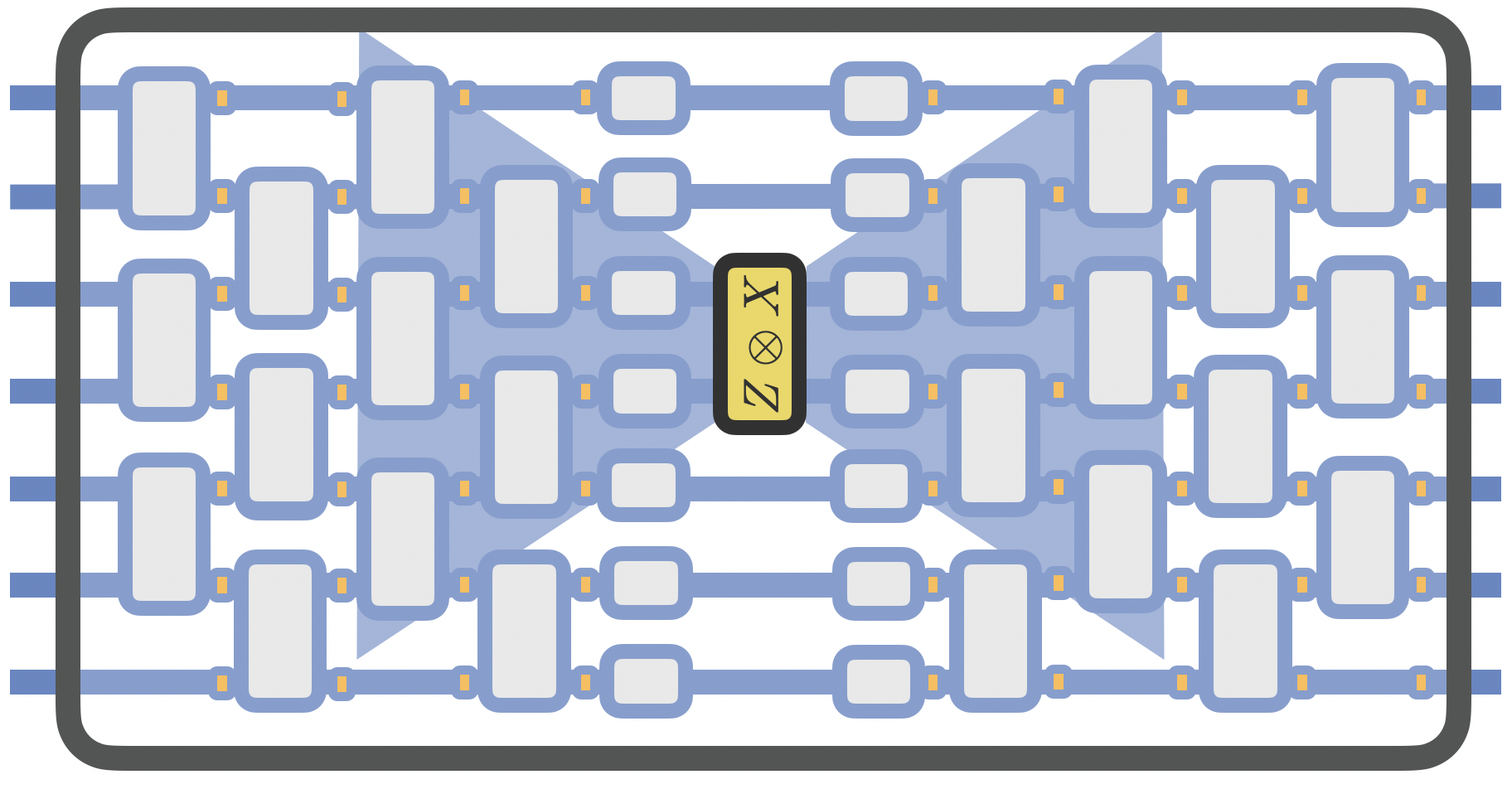}
\caption{A graphical representation of $\Phi^*_{[L-2,L]}(P)$ with respect to the local Pauli observable $P$ represented by the blue shaded area. 
The (noisy) gates outside the blue shaded area are contracted trivially due to the fact that the adjoint of every channel is unital, and thus cannot influence the expectation value of the Pauli. Even if the qubits in the system are $n$, the computation of $\Phi^*_{[L-2,L]}(P)$ is restricted to only a constant number of qubits.}
\label{Fig:circ33}
\end{figure}
\textcolor{black}{Note that for one-dimensional circuits $\mathrm{D}=1$ and local Pauli, the runtime depends only polynomially by $\varepsilon^{-1}$, while for higher dimension the time complexity depends quasi-polynomially by $\varepsilon^{-1}$.}

\textcolor{black}{Moreover, if $\Phi^{*}_{[L-m,L]}(P)$ is close to something proportional to the identity (which can be verified classically with the same time-complexity of above), then we can certify that our algorithm has succeeded. The intuition about this is that if $\Phi^{*}_{[L-m,L]}(P)$ were proportional to the identity, then keeping adding (adjoint) layers does not change the matrix because of the unitality of the adjoint channel.}
Specifically, at the end of the previous algorithm, we can check (efficiently in the effective dimension of the propagated observable) if the condition 
\begin{align}
    E \coloneqq \min_{q\in \mathbb{R}}\norm{\Phi^{*}_{[L-m,L]}(P)-q I}_{\infty}\le \varepsilon/2
\end{align}
is satisfied. If it is, then the previous algorithm succeeded with unit probability, as we are going to show in the next observation. 
\begin{observation}[Verification guarantees]
After running the algorithm described in Proposition~\ref{prop:classim}, if such condition is true:
\begin{align}
    E \coloneqq \min_{q\in \mathbb{R}}\norm{P_m-q I}_{\infty}\le \varepsilon/2,
\end{align}
where $P_l \coloneqq \Phi^{*}_{[L-m,L]}(P)$, then we can conclude that the algorithm in Proposition~\ref{prop:classim} succeeded with unit probability.
\textcolor{black}{Furthermore, such condition can be verified with the same time complexity of the algorithm in Proposition~\ref{prop:classim}. In particular, it holds that}
\begin{align}
    E= \frac{1}{2}\left(\lambda_{\max}(P_m)-\lambda_{\min}(P_m)\right),
\end{align}
where $\lambda_{\min}(P_m)$ and $\lambda_{\max}(P_m))$ are respectively the minimum and the maximum eigenvalue of $P_m$.
\end{observation}
\begin{proof}
Assuming that $E \coloneqq \min_{q\in \mathbb{R}}\norm{P_m-q I}_{\infty}\le \varepsilon/2$, we have
\begin{align}
    |\Tr(P\Phi(\rho_0))-\Tr(P\Phi_{[L-m,L]}(\rho_0))|&\le |\Tr(P\Phi(\rho_0))-q | + |q  -\Tr(P\Phi_{[L-m,L]}(\rho_0))| \\
    \nonumber
    &= |\Tr(P_l\Phi_{[1,L]}(\rho_0))-q | + |\Tr(P_m \rho_0)-q | \\
    \nonumber
    &=|\Tr((P_m-q I)\Phi_{[1,L]}(\rho_0))| + |\Tr((P_m -q I)\rho_0)| \\
    \nonumber
    &\le \norm{P_m-q I}_{\infty} + \norm{P_m -q I}_{\infty} \\
    &\le \varepsilon,
    \nonumber
\end{align}
where in the second to last step we have used H\"older inequality. This shows that if $E\le \varepsilon/2$, then the algorithm in Proposition~\ref{prop:classim} succeeds with unit probability.

Moreover, $E$ can be computed in polynomial time in the effective dimension of the observable $P_m$.
This can be seen for example by noting that $E$ depends only by the eigenvalues of $P_l$, since it suffices to compute its spectrum. In particular, we have
\begin{align}
    \min_{q\in\mathbb{R}}\norm{P_m-q I}_{\infty}= \min_{q\in\mathbb{R}} \max(|\lambda_{\min}(P_m)-q|,|\lambda_{\max}(P_m)-q|)=\frac{1}{2}\left(\lambda_{\max}(P_m)-\lambda_{\min}(P_m)\right).
\end{align}
This completes the proof.
\end{proof}
We note that such verification step can be also inserted at each step in which we take the adjoint of each of the last unitary layer. 
We refer to Algorithm~\ref{alg:sim} for a summary of the described steps.
\begin{algorithm}[H]
\textbf{Parameters}: $\varepsilon, \delta > 0$ (desired precision and success probability). \\
\textbf{Input}: Classical descriptions of $\mathcal{C}$ (noiseless circuit), noise channel $\N$ with parameter $c$ and observable of interest $P$.
%\vspace{5pt}
\begin{algorithmic}[1]
\State \textbf{Initialize}: $P_0 = P$.
%\State $P_l\gets \Phi^{*}_{[L-l,L]}(P)$ with $l \coloneqq \lceil \frac{1}{\log(c^{-1})}\log(\frac{4}{\delta \varepsilon^2})\rceil$. 
\State $l \coloneqq \lceil \frac{1}{\log(c^{-1})}\log(\frac{4}{\delta \varepsilon^2})\rceil$. 

\For{ $t = 1 \text{ to } m$, }
    \State $P_t \gets \Phi^*_{[L-t,L]}(P_{t-1})$. 
    \State $E_t \gets \frac{1}{2}\left(|\lambda_{\max}(P_t)-\lambda_{\min}(P_t)| \right)$ \Comment{Check early-break condition}
    \If{$ 2 E_t \le \varepsilon$}
    \State Output $\Tr(P_t \ketbra{0^n}{0^n}).$
    \State Break
    \EndIf
\EndFor

\State Output $\Tr(P_m \ketbra{0^n}{0^n})$
 
\end{algorithmic}
\caption{Computing local expectation values on noisy circuit \label{alg:sim}}
\end{algorithm}

The above algorithm is efficient if $|P|=O(\log(n))$ (i.e.,  its time complexity runs polynomially in the number of qubits), and it is no longer efficient if $|P|=\omega(\log(n))$.
However, in this high Pauli-weight regime, we do not need to run any algorithm, since we can just output zero and this succeeds with high probability over the choice of the circuit and with an inverse-polynomial accuracy, due to the following observation. 

\begin{observation}[Output zero if the Pauli is Global]
Let $\varepsilon=\Theta(1/\mathrm{poly}(n))$, $P\in \{I,X,Y,Z\}^{\otimes n}$, $\rho_0$ an arbitrary initial state, and let $\Phi$ be a noisy quantum circuit of any depth sampled according to the described circuit distribution. If $|P|=\omega(\log(n))$, then the probability that expectation value $\Tr(P\Phi(\rho_0))$ is larger than $\varepsilon$ is negligible: %\daniel{strange notation with the negl. }
\begin{align}
    \mathrm{Prob}\left(|\Tr(P\Phi(\rho_0))| \ge \varepsilon \right)\le \mathrm{negl}(n),
\end{align}
where $\mathrm{negl}(n)$ denotes a negligible function, i.e.,
a function that grows more slowly than any inverse polynomial in the number of qubits $n$.
\end{observation}
\begin{proof}
%\daniel{we can just say here that the argument is the same as before.}
Recalling that $\Ex[\Tr(P\Phi(\rho_0))]=0$, the Chebyshev inequality implies that
\begin{align}
    \mathrm{Prob}\left(|\Tr(P\Phi(\rho_0))| \ge \varepsilon \right) \le \frac{1}{\varepsilon^2}\Var\left[\Tr(P\Phi(\rho_0))\right] \le \frac{c^{|P|}}{\varepsilon^2} ,
\end{align}
where we have used Proposition~\ref{prop:upvar}. If $\varepsilon$ is at most inverse-polynomially small and $c^{|P|}$ with $|P|=\omega(\log(n))$ is super-polynomially small, then $c^{|P|}/(\varepsilon^2)$ will be negligible.
\end{proof}
Taken together, the results of this subsection give a classical simulation algorithm for estimating Pauli expectation values of (possibly non-unital) noisy random quantum circuits. If the required precision is constant in the number of qubits $n$, then the running time of the algorithm is efficient for any circuit architecture. If the required precision scales inverse-polynomially, then the algorithm runs in polynomial time for $1$-D architectures, while in quasi-polynomial time in higher constant dimensionality (e.g., $2$-D). 

\subsection{Classical simulation via depth- and Pauli-weight-truncated light cones}
\label{app:improved-simulation}
In the previous section we introduced an algorithm that, for any fixed architecture in constant dimension \(D\) and any initial state \(\rho_0\), has a runtime of \(\exp(O(\log(\varepsilon^{-1}))^{D})\), where \(\varepsilon\) is the desired accuracy of the system, whereas the runtime for all-to-all architectures was \(\exp(O(\varepsilon^{-1}))\). Thus, this algorithm is always efficient if the required precision is constant with respect to the number of qubits \(n\). However, if the desired precision scales inverse-polynomially with \(n\), the algorithm has a runtime of \(\exp(O(\log(n))^{D}) = n^{O(\log(n))^{D-1}}\) for any constant dimension \(D\), and exponential time for all-to-all connected architectures.

Here, we present an algorithm that improves upon these runtimes. Specifically, the runtime of this algorithm is \(n^{O((D-1)\log\!\log(n))}\) for geometrically local architectures in any dimension \(D\), and \(n^{O(\log(n))}\) for all-to-all connected architectures, independent of the depth and initial state. Since \(\log\!\log n\) is extremely small for practically relevant \(n\), the first runtime can be effectively regarded as polynomial in practice. The idea is simple: in addition to retaining only the last few layers and applying a light-cone argument as in the previous section, we can further restrict attention to Pauli strings of small weight inside the light cone. As we have already shown that high-weight Pauli terms contribute exponentially little to the expectation value, they can be safely discarded. Specifically, we combine the theorems proven in the previous section regarding effective depth and exponential suppression of Pauli weight with ideas from~\cite{angrisani2024classicallyestimatingobservablesnoiseless}, originally developed for the noiseless case, and generalize them to the possibly noisy setting.
We first introduce some useful definitions and lemmas.

\begin{definition}[Truncation map]
\label{def:truncation-map}
We define the truncation map \(T^{(>k)}: \mathcal{B}(\mathcal{H}_n) \to \mathcal{B}(\mathcal{H}_n)\) as the linear map that acts on any Pauli basis element \(P \in \{I, X, Y, Z\}^{\otimes n}\) as follows:
\begin{equation}
T^{(>k)}(P) \coloneqq 
\begin{cases} 
P & \text{if } |P| \le k, \\
0 & \text{if } |P| > k,
\end{cases}
\end{equation}
where \(|P|\) denotes the Pauli weight of \(P\) and \(k \in \mathbb{R}\) is a specified threshold. This definition extends to possibly non-Pauli operators by linearity.
\end{definition}

It is important to observe a few properties of the truncation map.

\begin{lemma}[Basic properties of the truncation map]
\label{le:trunmapprop}
We have
\begin{enumerate}
    \item \(\|T^{(>k)}(O)\|_2 \leq \|O\|_2\), for any observable \(O\).
    \item The truncation map commutes with any layer of single-qubit Clifford gates, i.e.,
    \begin{equation}
    \mathcal{C}^{\mathrm{single}} \circ T^{(>k)} = T^{(>k)} \circ \mathcal{C}^{\mathrm{single}},
    \end{equation}
    where \(\mathcal{C}^{\mathrm{single}} \coloneqq C(\cdot)C^{\dagger}\), with \(C \coloneqq \bigotimes_{i=1}^{n} C_i\), is a layer of single-qubit Clifford gates.
\end{enumerate}
\end{lemma}

\begin{proof}
The inequality \(\|T^{(>k)}(O)\|_2 \leq \|O\|_2\) can be shown by expanding \(O\) in the Pauli basis.
The commutativity property can be verified by explicitly applying the linear maps on both sides to an operator expanded in the Pauli basis, and using that single qubits Clifford layers do not change the Pauli weight of a given Pauli.
\end{proof}

We now present a Lemma, which is shown in Ref.~\cite{huang2023learningpredictarbitraryquantum}. 
\begin{lemma}[Low-degree approximation, Corollary 13 in Ref.~\cite{huang2023learningpredictarbitraryquantum} (Restated)]
\label{le:robert}
Let \(O\) be an operator and \(k \in \mathbb{R}\). Let \(\mathcal{D}\) be a distribution over quantum states that is invariant under single-qubit Clifford gates. Then, we have
\begin{align}
    \mathbb{E}_{\rho \sim \mathcal{D}} \left\lvert \Tr(O \rho) - \Tr(O^{(k)} \rho) \right\rvert^2 \le \left(\frac{2}{3}\right)^k \frac{\|O\|^2_{2}}{2^n},
\end{align}
where \(O^{(k)} \coloneqq T^{(>k)}(O)\).
\end{lemma}

\begin{lemma}[Linear map 2-norm bound]
\label{le:avg-inv-2norm}
%Let \(\nu\) be a distribution over the tensor product of single-qubit 2-design gates, specifically unitaries \(U\) of the form \(U = \bigotimes_{i=1}^n u_i\), where each \(u_i\) is a single-qubit unitary. 
Let $\nu$ be a distribution over the unitary group, satisfying
\begin{align}
    \forall P,Q \in \{I,X,Y,Z\}^{\otimes n} \text{ such that } P\neq Q : \bbE_{U\sim \nu} [U^{\dag \otimes 2} (P \otimes Q) U^{\otimes 2}] = 0.
\end{align}
For any linear map \(\Phi\) and any Hermitian operator \(O\), we have
\begin{align}
    \mathbb{E}_{U \sim \nu} \| \Phi^{\dagger}(U O U^{\dagger}) \|_2^2 \leq \left(\max_{s \in \calP_n} \| \Phi^\dag(s) \|_2^2\right) \|O\|_2^2.
\end{align}
In particular, if $\Phi$ is a quantum channel, we have
\begin{align}
    \mathbb{E}_{U \sim \nu} \| \Phi^{\dagger}(U O U^{\dagger}) \|_2^2 \leq \|O\|_2^2.
\end{align}
\end{lemma}

\begin{proof}
Expanding \(O\) and \(\Phi^{\dagger}(U O U^{\dagger})\) in the Pauli basis \(\mathcal{P}_n \coloneqq \left\{\frac{I}{\sqrt{2}}, \frac{X}{\sqrt{2}}, \frac{Y}{\sqrt{2}}, \frac{Z}{\sqrt{2}}\right\}^{\otimes n}\), we find
\begin{align}
    \mathbb{E}_{U \sim \nu} \| \Phi^{\dagger}(U O U^{\dagger}) \|_2^2 &= \sum_{t \in \mathcal{P}_n} \mathbb{E}_{U \sim \nu} \left[\Tr(t \Phi^{\dagger}(U O U^{\dagger}))\right]^2 \\
    \nonumber
    &= \sum_{t \in \mathcal{P}_n} \mathbb{E}_{U \sim \nu} \left[\Tr(\Phi(t) U O U^{\dagger})\right]^2 \\
    \nonumber
    &= \sum_{s, t \in \mathcal{P}_n} \Tr(s \Phi(t))^2 \mathbb{E}_{U \sim \nu} \left[\Tr(s U O U^{\dagger})\right]^2 \\
     \nonumber
     &= \sum_{s \in \mathcal{P}_n}  \left( \sum_{t \in \mathcal{P}_n} \Tr(s \Phi(t))^2\right) \mathbb{E}_{U \sim \nu} \left[\Tr(s U O U^{\dagger})\right]^2 \\
     \nonumber
    &\le  \sum_{s \in \mathcal{P}_n}  \left( \max_{s' \in \mathcal{P}_n}  \sum_{t \in \mathcal{P}_n} \Tr(s' \Phi(t))^2\right) \mathbb{E}_{U \sim \nu} \left[\Tr(s U O U^{\dagger})\right]^2 \\
     \nonumber
    &\le  \left( \max_{s' \in \mathcal{P}_n}  \sum_{t \in \mathcal{P}_n} \Tr(\Phi^{\dagger}(s') t)^2\right) \sum_{s \in \mathcal{P}_n}  \mathbb{E}_{U \sim \nu} \left[\Tr(s U O U^{\dagger})\right]^2 \\
     \nonumber
    &= \left(\max_{s' \in \mathcal{P}_n} \| \Phi^{\dagger}(s') \|_2^2\right) \mathbb{E}_{U \sim \nu} \| U O U^{\dagger} \|_2^2 \\
     \nonumber
    &= \left(\max_{s \in \mathcal{P}_n} \| \Phi^{\dagger}(s) \|_2^2\right) \mathbb{E}_{U \sim \nu} \| O \|_2^2.
     \nonumber
\end{align}
Here, we have used  Lemma~\ref{le:mixham} in the third step and the unitary invariance of the 2-norm in the last step. Therefore,
\begin{align}
    \mathbb{E}_{U \sim \nu} \| \Phi^{\dagger}(U O U^{\dagger}) \|_2^2 &= \left(\max_{s \in \mathcal{P}_n} \| \Phi^{\dagger}(s) \|_2^2\right) \| O \|_2^2 \\
    \nonumber
    &\leq \left(\max_{s \in \mathcal{P}_n} 2^n \| \Phi^{\dagger}(s) \|_\infty^2\right) \| O \|_2^2 \\
    \nonumber
    &= \left(\max_{P \in \left\{I, X, Y, Z\right\}^{\otimes n}} \| \Phi^{\dagger}(P) \|_\infty^2\right) \| O \|_2^2 \\
    \nonumber
    &\leq \left(\max_{P \in \left\{I, X, Y, Z\right\}^{\otimes n}} \| P \|_\infty^2\right) \| O \|_2^2 \\
    \nonumber
    &= \| O \|_2^2,
\end{align}
where in the second step we have used  the fact that \(\|A\|_2 \le \sqrt{\mathrm{rank}(A)} \|A\|_\infty\), in the fourth step we applied the Russo-Dye Theorem~\cite{Bhatia2007PositiveDM} (which states \(\|\Phi^{\dagger}(A)\|_\infty \le \|A\|_\infty\) for any quantum channel \(\Phi\)), and in the final step we noted that the operator norm of any Pauli matrix is 1.
\end{proof}
Note that if we remove the single-qubit random layers, the previous inequality becomes false, i.e., \(\| \Phi^{\dagger}( O) \|_2 \leq \|O\|_2\) does not hold in general, as can be verified by taking the single-qubit channel \(\Phi(\cdot) \coloneqq \Tr(\cdot) \ketbra{0}{0}\) and $O\coloneqq \ketbra{0}{0}$. 
\footnote{More generally, the same counterexample demonstrates that for \(1 \le p \le \infty\), the inequality \(\| \Phi^{\dagger}(O) \|_p \leq \|O\|_p\) can hold for all quantum channels \(\Phi\) and observables \(O\) only if \(p = \infty\). Specifically, for \(p = \infty\), the inequality \(\| \Phi^{\dagger}(O) \|_\infty \leq \|O\|_\infty\) is satisfied for all channels and observables (Russo-Dye Theorem~\cite{Bhatia2007PositiveDM}). Furthermore, it can be shown that for \(1 \le p < \infty\), the inequality \(\| \Phi^{\dagger}(O) \|_p \leq \|O\|_p\) holds for all observables \(O\) if and only if the channel \(\Phi\) is unital. This result follows by: 1)  for \(p > 1\), the inequality \(\| \Phi(\rho) \|_p \leq \|\rho\|_p\) holds for all matrices \(\rho\) if and only if the channel \(\Phi\) is unital, 2) from a duality argument, that is, let \(p, q \in \mathbb{R}\) be such that \(p^{-1} + q^{-1} = 1\), then \(\Phi\) is \(p\)-norm contractive if and only if \(\Phi^{\dagger}\) is \(q\)-norm contractive.}

We now recall our definition of noisy random circuit.
\begin{definition}[Noisy Circuit Model]
\label{def:circuitmod}
We consider $n$-qubit quantum circuits \(\Phi\) consisting of layers of two-qubit gates interleaved by local (single-qubit) noise, with a final layer of single-qubit gates. All gates are assumed to be drawn from a $2$-design, and we make no assumptions about geometric locality, except where explicitly mentioned.
We express our circuits as
\begin{align}
\Phi \coloneqq \mathcal{V}^{\mathrm{single}} \circ \mathcal{N}^{\otimes n} \circ \mathcal{U}_{L} \circ \cdots \circ \mathcal{N}^{\otimes n} \circ \mathcal{U}_{1},
\label{eq:randcirc_main}
\end{align}
where \(\mathcal{V}^{\mathrm{single}} \coloneqq V(\cdot)V^{\dagger}\), with \(V \coloneqq \bigotimes_{i=1}^{n} U_i\), is a layer of single-qubit gates, \(L\) represents the number of layers (also referred to as circuit depth), \(\mathcal{U}_{i}\) corresponds to the channel associated with the \(i\)-th unitary circuit layer for \(i \in [L] \coloneqq \{1, 2, \dotsc, L\}\), and \(\mathcal{N}\) is a single-qubit quantum channel.
\end{definition}

\begin{definition}[Weight and Depth Truncated Adjoint Circuit]
\label{def:truncwd}
Let \(\Phi \coloneqq \Phi_L \circ \cdots \circ \Phi_1\) be a noisy quantum circuit, where \(\{\Phi_j\}_{j=1}^L\) represents the sequence of noisy circuit layers, as for the previous definition. Let \(\Phi_{[L-m,L]}(\cdot)\) be the noisy circuit where only the last \(m\) layers are considered.

We define the \(k\)-weight truncated adjoint circuit restricted to the last \(m\)-layers \((\Phi_{[L-m,L]})_{k\mathrm{-trunc}}^{\dagger}\) as
\begin{equation}
(\Phi_{[L-m,L]})_{k\mathrm{-trunc}}^{\dagger} \coloneqq (T^{(>k)} \circ \Phi_{L-m}^{\dagger}) \circ \cdots \circ (T^{(>k)} \circ \Phi_L^{\dagger}),
\end{equation}
where \(T^{(>k)}\) is the truncation map defined in Definition~\ref{def:truncation-map}.
\end{definition}

We are now ready to state the main results of this section. We begin with the following proposition, which combines the effective depth of noisy random circuits with Pauli-weight truncation. 
While the $m$ factor multiplying the second term on the RHS can be removed using a refinement of the techniques in Ref.~\cite{angrisani2024classicallyestimatingobservablesnoiseless}, this improvement does not change the asymptotic scaling of our classical simulation result and is therefore omitted here for simplicity.

\begin{proposition}[Truncation in depth and weight approximation]
\label{th:mainth}
Let \(\rho_0\) be an initial state, \(O\) an observable, and \(\Phi \coloneqq \Phi_L \circ \cdots \circ \Phi_1\) a noisy quantum circuit, where \(\{\Phi_j\}_{j=1}^L\) represents the sequence of noisy circuit layers as in Definition~\ref{def:circuitmod}.
Define \(O_{\Phi}^{(k,m)} \coloneqq (\Phi_{[L-m,L]})_{k\mathrm{-trunc}}^{\dagger}(O)\), which represents the Heisenberg-evolved observable \(O\) with the last \(m\) noisy circuit layers $k$-weight truncated (layer-by-layer), as in Definition~\ref{def:truncwd}.
We have
\begin{align}
\label{eq:uppboundTRUNC}
\mathbb{E}_{\Phi}\left[ \left| \Tr(O\Phi(\rho_0)) - \Tr(O_{\Phi}^{(k,m)} \sigma_0) \right| \right] \le 2 \norm{O}_{\infty} \exp(-\alpha m) + m \left(\frac{2}{3}\right)^{k/2} 
\|O\|_{\infty},
\end{align}
where $\sigma_0$ is any preferred initial state. 
Here, the average \(\mathbb{E}_{\Phi}\) is taken with respect to the $2$-design distribution of every two-qubit gate that composes the circuit, and \(\alpha > 0\) is a quantity that depends only on the noise parameters.
\end{proposition}

\begin{proof}
We have 
\begin{align}
\mathbb{E}_{\Phi}\left[ \left| \Tr(O\Phi(\rho_0)) -  \Tr(O_{\Phi}^{(k,m)} \sigma_0) \right| \right] &\le 
\mathbb{E}_{\Phi}\left[ \left| \Tr(O\Phi(\rho_0)) - \Tr(O\Phi_{[L-m]}(\rho_0)) \right| \right] + \mathbb{E}_{\Phi}\left[ \left| \Tr(O\Phi_{[L-m]}(\rho_0)) - \Tr(O_{\Phi}^{(k,m)} \sigma_0) \right| \right]
\nonumber
\\
&\le  2 \norm{O}_{\infty} \exp(-\alpha m) +  \mathbb{E}_{\Phi}\left[ \left| \Tr(O\Phi_{[L-m]}(\rho_0)) - \Tr(O_{\Phi}^{(k,m)} \sigma_0) \right| \right],
\end{align}
where we have used  triangle inequality and Proposition~\ref{prop:expdecayO}.
We now focus on the second term. 
We have
\begin{align}
\label{eq:1pr}
     \mathbb{E}_{\Phi}\left[ \left| \Tr(O\Phi_{[L-m]}(\rho_0)) - \Tr(O_{\Phi}^{(k,m)} \sigma_0) \right| \right]&=  \mathbb{E}_{\Phi}\left[ \left| \Tr(O\Phi_{[L-m]}(\rho_0)) - \Tr((\Phi_{[L-m,L]})_{k\mathrm{-trunc}}^{\dagger}(O) \rho_0) \right| \right]\\
     \nonumber
     &=  \mathbb{E}_{\Phi}\left[ \left| \Tr(\Phi^{\dagger}_{[L-m]}(O)\rho_0) - \Tr((\Phi_{[L-m,L]})_{k\mathrm{-trunc}}^{\dagger}(O) \rho_0) \right| \right]\\
      \nonumber
     &= \mathbb{E}_{\Phi}\left[ \left| \Tr\!\left(\left(\Phi^{\dagger}_{[L-m]}(O) - (\Phi_{[L-m,L]})_{k\mathrm{-trunc}}^{\dagger}(O)\right) \rho_0\right) \right| \right].
      \nonumber
\end{align}
We notice that we can rewrite \(\Phi^{\dagger}_{[L-m]}(O) - (\Phi_{[L-m,L]})_{k\mathrm{-trunc}}^{\dagger}(O)\) using a telescopic sum as 
\begin{align}
    &\Phi^{\dagger}_{[L-m]}(O) - (\Phi_{[L-m,L]})_{k\mathrm{-trunc}}^{\dagger}(O) 
    \\
    \nonumber
    &= \Phi_{L-m}^{\dagger} \circ \cdots \circ \Phi_L^{\dagger}(O) - (T^{(>k)} \circ \Phi_{L-m}^{\dagger}) \circ \cdots \circ (T^{(>k)} \circ \Phi_L^{\dagger})(O) \\
     \nonumber
    &= \bigcirc_{a=L-m}^{L} \Phi_{a}^{\dagger}(O) - \bigcirc_{b=L-m}^{L} (T^{(>k)} \circ \Phi_{b}^{\dagger})(O) \\
     \nonumber
    &= \sum_{j=0}^{m-1} \left( \Phi_{[L-m,L-j]}^{\dagger} \circ \bigcirc_{b=L-j+1}^{L} (T^{(>k)} \circ \Phi_{b}^{\dagger})(O) - \Phi_{[L-m,L-j-1]}^{\dagger} \circ \bigcirc_{b=L-j}^{L} (T^{(>k)} \circ \Phi_{b}^{\dagger})(O) \right).
     \nonumber
\end{align}
Here, the notation \(\bigcirc_{a=L-m}^{L}\) denotes the composition of maps from \(a = L-m\) to \(L\) in the forward direction, while the notation \(\Phi_{[a,b]}^{\dagger}\) with \(a \le b\) means \(\Phi_a^{\dagger} \circ \cdots \circ \Phi_b^{\dagger}\).
Substituting the telescoping sum into Eq.~\eqref{eq:1pr} and applying the triangle inequality, we get
\begin{align}
&\mathbb{E}_{\Phi}\left[ \left| \Tr(O\Phi_{[L-m]}(\rho_0)) - \Tr(O_{\Phi}^{(k,m)} \sigma_0) \right| \right] 
\\
\nonumber
&= \mathbb{E}_{\Phi}\left[ \left| \Tr\!\left(\left(\Phi^{\dagger}_{[L-m]}(O) - (\Phi_{[L-m,L]})_{k\mathrm{-trunc}}^{\dagger}(O)\right) \rho_0\right) \right| \right] \\
 \nonumber
&= \mathbb{E}_{\Phi}\left[ \left| \sum_{j=0}^{m}\Tr\!\left( \left( \Phi_{[L-m,L-j]}^{\dagger} \circ \bigcirc_{b=L-j+1}^{L} (T^{(>k)} \circ \Phi_{b}^{\dagger})(O) - \Phi_{[L-m,L-j-1]}^{\dagger} \circ \bigcirc_{b=L-j}^{L} (T^{(>k)} \circ \Phi_{b}^{\dagger})(O) \right) \rho_0\right) \right| \right] \\
 \nonumber
&\le \sum_{j=0}^{m-1} \mathbb{E}_{\Phi}\left[ \left| \Tr\!\left(\left( \Phi_{[L-m,L-j]}^{\dagger} \circ \bigcirc_{b=L-j+1}^{L} (T^{(>k)} \circ \Phi_{b}^{\dagger})(O) - \Phi_{[L-m,L-j-1]}^{\dagger} \circ \bigcirc_{b=L-j}^{L} (T^{(>k)} \circ \Phi_{b}^{\dagger})(O) \right) \rho_0\right) \right| \right] 
\\
 \nonumber
&= \sum_{j=0}^{m-1} \mathbb{E}_{\Phi}\left[ \left| \Tr\!\left(\left( \Phi_{L-j}^{\dagger} \circ \bigcirc_{b=L-j+1}^{L} (T^{(>k)} \circ \Phi_{b}^{\dagger})(O) - (T^{(>k)} \circ \Phi_{L-j}^{\dagger}) \bigcirc_{b=L-j+1}^{L} (T^{(>k)} \circ \Phi_{b}^{\dagger})(O) \right) \Phi_{[L-m,L-j-1]}(\rho_0)\right) \right| \right] \\
 \nonumber
&= \sum_{j=0}^{m-1} \mathbb{E}_{\Phi}\left[ \left| \Tr\!\left(\left( \Phi_{L-j}^{\dagger} (O_{\Phi}^{(k,L-j+1)}) - (T^{(>k)} \circ \Phi_{L-j}^{\dagger})(O_{\Phi}^{(k,L-j+1)}) \right) \Phi_{[L-m,L-j-1]}(\rho_0)\right) \right| \right] \\
&\le \sum_{j=0}^{m-1} \left( \mathbb{E}_{\Phi}\left[ \left| \Tr\!\left(\left( \Phi_{L-j}^{\dagger} (O_{\Phi}^{(k,L-j+1)}) - (T^{(>k)} \circ \Phi_{L-j}^{\dagger})(O_{\Phi}^{(k,L-j+1)}) \right) \Phi_{[L-m,L-j-1]}(\rho_0)\right) \right|^2 \right] \right)^{1/2},
 \nonumber
\end{align}
where in the last step we have used  Jensen's inequality.

We now focus on each term in the sum. Since we are dealing with second moment quantities for each unitary layer, we can simplify the expression by adding independent single-qubit Clifford gates layers before and after each layer. Moreover, by Lemma~\ref{le:trunmapprop} (point 2), the truncation map \( T^{(>k)} \) commutes with any single-qubit Clifford layer. Thus, we the distribution associated to the state \(\Phi_{[L-m,L-j-1]}(\rho_0)\) is invariant under single-qubit Clifford rotations. Applying Lemma~\ref{le:robert}, we obtain
\begin{align}
     \mathbb{E}_{\Phi}\left[ \left| \Tr\!\left(\left( \Phi_{L-j}^{\dagger} (O_{\Phi}^{(k,L-j+1)}) - (T^{(>k)} \circ \Phi_{L-j}^{\dagger})(O_{\Phi}^{(k,L-j+1)}) \right) \Phi_{[L-m,L-j-1]}(\rho_0)\right) \right|^2 \right] \le \left(\frac{2}{3}\right)^k \frac{1}{2^n} \mathbb{E}_{\Phi} \left\| O_{\Phi}^{(k,L-j+1)} \right\|_2^2.
\end{align}
Here, the expected value is taken only over the layers that appear in \(O_{\Phi}^{(k,L-j+1)}\). We recall the definition of this term. For each \(j \in \{1, \ldots, m-1\}\), we have
\begin{align}
    O_{\Phi}^{(k,L-j+1)} &= (T^{(>k)} \circ \Phi_{L-j}^{\dagger}) \circ \cdots \circ (T^{(>k)} \circ \Phi_L^{\dagger})(O) \\
    \nonumber
    &= T^{(>k)} \circ \Phi_{L-j}^{\dagger}(O_{\Phi}^{(k,L-j+2)}),
\end{align}
where \(O_{\Phi}^{(k,L+1)} \coloneqq O\). Thus, we have
\begin{align}
    \mathbb{E}_{\Phi} \left\| O_{\Phi}^{(k,L-j+1)} \right\|_2^2 &= \mathbb{E}_{\Phi} \left\| T^{(>k)} \circ \Phi_{L-j}^{\dagger}(O_{\Phi}^{(k,L-j+1)}) \right\|_2^2 \\
    \nonumber
    &\le \mathbb{E}_{\Phi} \left\| \Phi_{L-j}^{\dagger}(O_{\Phi}^{(k,L-j-1)}) \right\|_2^2 \\
     \nonumber
    &\le \mathbb{E}_{\Phi} \left\| O_{\Phi}^{(k,L-j-2)} \right\|_2^2 \\
     \nonumber
    &\le \left\| O \right\|_2^2,
\end{align}
where in the second step we have used  that the truncation map always contracts the Frobenius norm (Lemma~\ref{le:trunmapprop}), in the third step we have used  the invariance of the circuit layer distribution under single-qubit Clifford layers and used Lemma~\ref{le:avg-inv-2norm}, and in the last step we applied this argument for every layer.
Thus, we obtain
\begin{align}
     \mathbb{E}_{\Phi}\left[ \left| \Tr\!\left(\left( \Phi_{L-j}^{\dagger} (O_{\Phi}^{(k,L-j+1)}) - (T^{(>k)} \circ \Phi_{L-j}^{\dagger})(O_{\Phi}^{(k,L-j+1)}) \right) \Phi_{[L-m,L-j-1]}(\rho_0)\right) \right|^2 \right] &\le \left(\frac{2}{3}\right)^k \frac{1}{2^n} \left\| O \right\|_2^2.
\end{align}
Therefore, we get
\begin{align}
\mathbb{E}_{\Phi}\left[ \left| \Tr(O \Phi_{[L-m]}(\rho_0)) - \Tr(O_{\Phi}^{(k,m)} \sigma_0) \right| \right] &\le \sum_{j=0}^{m-1} \left(\frac{2}{3}\right)^{k/2} \frac{1}{2^{n/2}} \left\| O \right\|_2 \\
 \nonumber
&\le m \left(\frac{2}{3}\right)^{k/2} \left\| O \right\|_\infty,
 \nonumber
\end{align}
where in the last step we have used  the fact that \(\|A\|_2 \le \sqrt{\text{rank}(A)} \|A\|_\infty\).
Thus, we can conclude the proof.
\end{proof}

\begin{corollary}[Error scaling]
\label{cor:errscal}
Let \(\varepsilon_A, \varepsilon_B > 0\) be accuracy parameters and let \(\delta \in (0,1)\) be the failure probability.
Consider an initial state \(\rho_0\), an observable \(O\), and a noisy quantum circuit architecture \(\Phi \coloneqq \Phi_L \circ \cdots \circ \Phi_1\), where \(\{\Phi_j\}_{j=1}^L\) represents the sequence of noisy circuit layers. Here, \(\alpha > 0\) is a quantity that depends only on the noise parameters (see Proposition~\ref{prop:expdecayO} for details).
Define \(O_{\Phi}^{(k,m)} \coloneqq (\Phi_{[L-m,L]})_{k\mathrm{-trunc}}^{\dagger}(O)\), which represents the Heisenberg-evolved observable \(O\) with the last \(m\) noisy circuit layers \(k\)-weight truncated (layer-by-layer), as described in Definition~\ref{def:truncwd}.
Assume 
\begin{equation}
    m \coloneqq \left\lceil \alpha^{-1} \log\left(\frac{4\norm{O}_{\infty}}{\delta \varepsilon_A}\right) \right\rceil ,\,\, k \coloneqq \left\lceil \frac{2}{\log(2/3)} \log\left(\frac{m \norm{O}_\infty}{\delta \varepsilon_B}\right) \right\rceil,
    \end{equation}where \(m\) and \(k\) satisfy \(L \ge m\) and \(n \ge k\) (i.e., the truncated depth is less than the entire circuit depth and the truncated weight is less than the number of qubits). Let $\sigma_0$ be any preferred initial state. Given a randomly chosen instance of the circuit \(\Phi\) where the two-qubit gates in each layer are distributed according to a \(2\)-design distribution, we can guarantee that the quantity \( \Tr(O_{\Phi}^{(k,m)} \sigma_0)\) is close to the expectation value \(\Tr(O\Phi(\sigma_0))\) by \(\varepsilon_A + \varepsilon_B\), i.e.,
\begin{equation}
\left| \Tr(O\Phi(\rho_0)) - \Tr(O_{\Phi}^{(k,m)} \sigma_0) \right| \le \varepsilon_A + \varepsilon_B,
\end{equation}
with probability at least \(1-\delta\) over the random choice of \(\Phi\).
\end{corollary}

\begin{proof}
By the previous theorem, with \(k\) and \(m\) chosen as defined, we have
\begin{equation}
\mathbb{E}_{\Phi}\left[ \left| \Tr(O\Phi(\rho_0)) - \Tr(O_{\Phi}^{(k,m)} \sigma_0) \right| \right] \le (\varepsilon_A + \varepsilon_B)\delta.
\end{equation}
Applying Markov's inequality, we obtain
\begin{equation}
\mathrm{Prob}\left( \left| \Tr(O\Phi(\rho_0)) - \Tr(O_{\Phi}^{(k,m)} \sigma_0) \right| > \varepsilon_A + \varepsilon_B \right) \le \frac{(\varepsilon_A + \varepsilon_B)\delta}{\varepsilon_A + \varepsilon_B} = \delta.
\end{equation}
Therefore, with probability at least \(1 - \delta\), the inequality 
\begin{equation}
\left| \Tr(O\Phi(\rho_0)) - \Tr(O_{\Phi}^{(k,m)} \sigma_0) \right| \le \varepsilon_A + \varepsilon_B
\end{equation}
holds. This completes the proof.
\end{proof}

\begin{remark}
If we set the accuracy \(\varepsilon_A = \mathrm{poly}(n^{-1})\), \(\varepsilon_B = \mathrm{poly}(n^{-1})\) and the failure probability \(\delta = \mathrm{poly}(n^{-1})\), then the parameters \(m\) and \(k\) grow logarithmically with the number of qubits \(n\), i.e., \(m = O(\log(n))\) and \(k = O(\log(n))\).
\end{remark}
\begin{theorem}[Classical Estimation of Expectation Values]
Consider a quantum circuit \(\Phi\) sampled uniformly at random with respect to a fixed architecture, which may include any type of local (non-unitary) noise and any number of layers. Given a Pauli operator \(P \in \{I, X, Y, Z\}^{\otimes n}\) with Pauli weight \(|P| = O(1)\), and an initial state \(\rho_0\), we aim to estimate the expectation value \(\Tr(P \Phi(\rho_0))\) with accuracy \(\varepsilon\), where \(\varepsilon\) is an inverse-polynomial in \(n\), specifically \(\varepsilon = \mathrm{poly}(n^{-1})\).

There exists a classical algorithm that accomplishes this estimation with a success probability of at least \(1 - \delta\) over the choice of the random circuit \(\Phi\), where \(\delta > 0\) is a constant or an inverse-polynomial in \(n\). The runtime of this algorithm is:
\begin{itemize}
    \item \(n^{O((D-1)\log\!\log(n))}\mathrm{poly}(n)\) for \(D\)-dimensional geometrically local architectures.
    \item \(n^{O(\log(n))}\) for all-to-all connected architectures.
\end{itemize}
\end{theorem}

\begin{proof}
To apply Corollary~\ref{cor:errscal}, we need to determine the time complexity of computing the estimator \(\Tr(O_{\Phi}^{(k,m)} \sigma_0)\), where we choose \(\sigma_0\) to be the zero state \(\ketbra{0^n}{0^n}\). Specifically, we first want to find the time complexity of evaluating \(O_{\Phi}^{(k,m)}\), which is the Pauli \(P\) (with \(|P| = O(1)\)) Heisenberg-evolved with the last \(m\) layers of the circuit and \(k\)-truncated,
\begin{align}
    P_{\Phi}^{(k,m)} = (\Phi_{[L-m,L]})_{k\text{-trunc}}^{\dagger}(P) = (T^{(>k)} \circ \Phi_{L-m}^{\dagger}) \circ \cdots \circ (T^{(>k)} \circ \Phi_L^{\dagger})(P) = T^{(>k)} \circ \Phi_{L-m}^{\dagger}(P_{\Phi}^{(k,L-m+1)}).
\end{align}

First, note that such Heisenberg-evolved local Pauli is supported on at most \(n_{\mathrm{eff}} \coloneqq |P|m^D\) qubits due to light-cone arguments, where \(D\) is the dimension of the circuit (for all-to-all connectivity, this would be \(n_{\mathrm{eff}} \coloneqq |P|2^m\)). The operator \(P_{\Phi}^{(k,m)}\) is supported only on Pauli operators in the light cone with Pauli weight at most \(k\) (this is true not only for \(P_{\Phi}^{(k,m)}\) but also for \(P_{\Phi}^{(k,1)}, \dots, P_{\Phi}^{(k,m-1)}\)). Let \(\mathcal{P}^{(\le k)}_{\mathrm{light}}\) denote this set of Pauli operators. 
The number of such Pauli operators is upper bounded by
\begin{align}
    |\mathcal{P}^{(\le k)}_{\mathrm{light}}| &= \sum_{w=0}^k \binom{n_{\mathrm{eff}}}{w} 3^w \le \sum_{w=0}^k \frac{n_{\mathrm{eff}}^w}{w!} 3^w \le n_{\mathrm{eff}}^k \sum_{w=0}^k \frac{3^w}{w!} \le n_{\mathrm{eff}}^k \sum_{w=0}^\infty \frac{3^w}{w!} = e^3 n_{\mathrm{eff}}^k,
\end{align}
where \(e\) is the base of the natural logarithm.

The operator \(T^{(>k)} \circ \Phi_{L}^{\dagger}(P) = T^{(>k)}(\Phi_{L}^{\dagger}(P))\) is supported in the light cone and, in particular, on Pauli operators with weight at most \(k\). Thus, it can be expressed as
\begin{align}
   P_{\Phi}^{(k,1)} = T^{(>k)} (\Phi_{L}^{\dagger}(P)) = \frac{1}{d} \sum_{Q \in \{I, X, Y, Z\}^{\otimes n_{\mathrm{eff}}}} \Tr(Q \Phi_{L}^{\dagger}(P)) T^{(>k)}(Q) = \frac{1}{d} \sum_{Q \in \mathcal{P}^{(\le k)}_{\mathrm{light}}} \Tr(Q \Phi_{L}^{\dagger}(P)) Q.
\end{align}
Since \(\Phi_L\) is a circuit layer where each qubit is acted upon by at most one gate in \(\Phi_L\), the term \(\Tr(Q \Phi_{L}^{\dagger}(P))\) can be computed in \(O(n)\) time, leveraging the tensor product structure of the Pauli matrices and the locality of the noise/gates. Therefore, computing all such coefficients \(\Tr(Q \Phi_{L}^{\dagger}(P))\) takes \(O(n |\mathcal{P}^{(\le k)}_{\mathrm{light}}|)\) time.
We then evaluate
\begin{align}
    P_{\Phi}^{(k,2)} = \frac{1}{d} \sum_{Q \in \mathcal{P}^{(\le k)}_{\mathrm{light}}} \Tr(Q \Phi_{L}^{\dagger}(P)) T^{(>k)} (\Phi_{L-1}^{\dagger}(Q)).
\end{align}
For each term \(T^{(>k)} (\Phi_{L-1}^{\dagger}(Q))\), we repeat the previous process. Once computed, these terms can be rearranged in the Pauli basis of \(\mathcal{P}^{(\le k)}_{\mathrm{light}}\) to find coefficients \(c_Q\) such that:
\begin{align}
    P_{\Phi}^{(k,2)} = \sum_{Q \in \mathcal{P}^{(\le k)}_{\mathrm{light}}} c_Q Q.
\end{align}
This procedure takes \(O(n |\mathcal{P}^{(\le k)}_{\mathrm{light}}|^2)\) time.

Repeating this procedure for all \(m\) layers results in an overall time complexity of \(O(n m |\mathcal{P}^{(\le k)}_{\mathrm{light}}|^2)\).
Now we can estimate for each of this Pauli the expectation value with the state \(\ketbra{0^n}{0^n}\), which takes additional $ |\mathcal{P}^{(\le k)}_{\mathrm{light}}|$ time complexity.

For \(m = O(\log(n))\) and \(k = O(\log(n))\) (which is the interesting regime since the desired accuracy is inverse-polynomially small), we have
\begin{align}
    |\mathcal{P}^{(\le k)}_{\mathrm{light}}| \le O(n_{\mathrm{eff}}^k) = O(m^{D k}) = O(\log(n)^{D \log(n)}) = n^{O(D \log(\log(n)))}.
\end{align}
The \(D = 1\) case can be analyzed separately, yielding \( |\mathcal{P}^{(\le k)}_{\mathrm{light}}| = \mathrm{poly}(n)\).
For all-to-all connectivity, we have
\begin{align}
    |\mathcal{P}^{(\le k)}_{\mathrm{light}}| = O(n_{\mathrm{eff}}^k) = O(2^{m k}) \le n^{O(\log n)}.
\end{align}
The time complexity is dominated by $ |\mathcal{P}^{(\le k)}_{\mathrm{light}}|$, so we can conclude the proof.
\end{proof}

\newpage
\section{Quantum machine learning under non-unital noise: Barren plateaus}
In this section, we rigorously show that non-unital noise induces absence of barren plateaus for local cost functions, in contrast to the unital scenario~\cite{nibp}. Specifically, in Subsection~\ref{sub:absenceBP} we establish that the gates in the last $\Theta(\log(n))$ layers are trainable, whereas those preceding them are not.
This complements the results presented in the previous section by rigorously showing the significance of the last $\Theta(\log(n))$ layers.
Moreover, we establish that global cost functions exhibit barren plateaus.
In Subsection~\ref{sub:unital}, we also present an improved upper bound on the onset of barren plateaus in the unital noise scenario compared to the one shown in Ref.~\cite{nibp}.

The results we show in this section are in stark contrast with the behavior of quantum circuits in the noiseless regime or with local depolarizing noise~\cite{napp2022quantifying,alex2021random}, as summarized in Table~\ref{tab:sensi}. 

\begin{table}[h]
    \centering
    %\label{tab:costconcentrationBP}
    \caption{\textbf{Trainability w.r.t.~the last $g(n)$-layers }}
    \begin{tabular}{lccc}  \hline
 \textbf{Noise model}
    & $ g(n)=\omega(\log(n))$ & $ \,\,\Theta(\log(n))$ \,\,& $\Theta(1)$
    \\ \hline
    Noiseless~\cite{McClean_2018,napp2022quantifying} & \textcolor{black}{\xmark}& \textcolor{black}{\xmark}&  \textcolor{black}{\xmark}\\
        Unital noise~\cite{nibp}  & \textcolor{black}{\xmark}& \textcolor{black}{\xmark}&  \textcolor{black}{\xmark}\\
        
        Non-unital noise~[This work] 
        & \textcolor{black}{\xmark} & \textcolor{teal}{\cmark} &  \textcolor{teal}{\cmark} \\ \hline
    \end{tabular}
    \label{tab:sensi}
    \begin{flushleft}
        Table~\ref{tab:sensi} shows that the last $\Theta(\log(n))$ layers of a non-unital noise circuit are the only trainable layers. This behavior is notably absent in the unital and noiseless noise regime for circuits with depth $\omega(\log(n))$: in these cases the gates in all the layers are not trainable.
    \end{flushleft}
    %\text{}
\end{table}

\subsection{Preliminaries on barren plateaus}
In this section, we introduce concepts that will be crucial for our discussion.
We use an analogous circuit model described in subsection~\ref{sub:circuitmodel}, namely we consider $n$-qubit quantum circuits $\Phi$ of the form
\begin{align}
    \Phi=(\mathcal{V}_{L}^{\,\mathrm{single}}\circ\mathcal{N}^{\otimes n} \circ \mathcal{U}_L) \circ \cdots   \circ (\mathcal{V}_{1}^{\,\mathrm{single}}\circ\mathcal{N}^{\otimes n} \circ \mathcal{U}_1)  ,
    \label{eq:randcirc23}
\end{align}
where $L$ represents the number of layers, also referred to as circuit depth, $\{\mathcal{V}_k^{\mathrm{single}}\}^L_{k=1}$ are layers of single-qubit gates distributed according a single-qubit $2$-design, $\mathcal{U}_{i} \coloneqq U_i(\cdot)U^{\dag}_i$ corresponds to the $n$-qubit unitary channel associated with the unitary layer $U_i$ for $i\in [L]$ which is formed by two-qubits gates, and $\mathcal{N}$ is a single-qubit quantum channel. 
Recall that we assume that the two-qubit gates in the circuit are distributed according to a two-qubit $2$-design (see Definition~\ref{def:localdesignlayer}). 
For example, our model encompasses the brickwork architecture in Fig.~\ref{Fig:circ}. Remember that, because of the unitary invariance of the two-qubit $2$-design layers, one can add `for free' layers of single-qubit Haar random gates, since we will be considering only up to second moment quantities. Thus, in the above equation, the layer of single-qubit gates (apart from the last one, $\mathcal{V}_{L}^{\,\mathrm{single}}$) can be removed.  

We now assume that the circuit is also also dependent on variational parameters $\boldsymbol{\theta}  \coloneqq  (\theta_1, \dots, \theta_m) \in \mathbb{R}^{m}$, which parameterize some of the two-qubit gates, which come from the set $\{\exp(-i \theta_{\mu}H_\mu)\}^m_{\mu=1}$, where $H_\mu$ are two-local Hermitian operators with $\norm{H_\mu}_{\infty}\le 1$. Specifically, we assume for simplicity that these parameterized gates are positioned at the start of the unitary layer $\mathcal{U}_i$ for $i\in [L]$. It is important to note that while we introduce these parameterized gates, they do not impact our model due to left-right invariance of the two-qubit $2$-design layers we consider, and so they can be considered part of one of the unitary layers $\mathcal{U}_i$ for $i\in [L]$; their introduction is primarily to facilitate the discussion on partial derivatives and barren plateaus.

In quantum machine learning jargon, the term \textit{cost function} is usually referred to as an expectation value of an Hermitian operator over a `parameterized' quantum state.
\begin{definition}[Cost function]
Let \(H\) be an Hermitian operator. Let $\rho_0$ be a quantum state and $\Phi$ be a noisy quantum circuit as defined previously. We define the cost function \(C\!\left(\boldsymbol{\theta}\right)\) associated with \(H\) and $\Phi(\rho_0)$ as
\begin{align}
    C\!\left(\boldsymbol{\theta}\right) \coloneqq  \Tr(H\Phi(\rho_0)).
\end{align}
\end{definition}
We will often omit the $\boldsymbol{\theta}$-dependence and write simply $C$ instead of $C\!\left(\boldsymbol{\theta}\right)$. 
As usual, when we write expected values or variances, it will always be with respect to the distribution from which we sample the gates that compose our quantum circuit.

Next, we introduce the notion of lack of barren plateaus.
\begin{definition}[Lack of barren plateaus]
\label{def:BP}
We say a cost function \(C\) lacks barren plateaus if and only if
\begin{align}
    \Ex\left[\|\nabla_{\boldsymbol{\theta}}C\|^2_2\right]=\Omega\left(\frac{1}{\mathrm{poly}(n)}\right),
\end{align}
where $\nabla_{\boldsymbol{\theta}}C  \coloneqq (\frac{\partial C }{\partial \theta_1} ,\cdots, \frac{\partial C}{\partial \theta_m} )$ is the gradient of the cost function.
\end{definition}
Hence, we assert that a cost function has barren plateaus if and only if the variance of the $2$-norm of the gradient is at least super-polynomially small.
We define now the notion of \emph{trainability} of a parametrized gate, which is useful to identify which gate in the circuit influences significantly the cost function (on average).
\begin{definition}[Trainability of the cost function with respect to a parameter]
    We say that a cost function $C$ is trainable with respect to the parameter $\theta_\mu$ if and only if
    \begin{align}
        \Var\!\left[\partial_\mu C\right] =\Omega\left(\frac{1}{\mathrm{poly}(n)}\right),
    \end{align}
    where we have denoted $\partial_\mu C \coloneqq  \frac{\partial C}{\partial \theta_\mu} $.
\end{definition}
We point out that partial derivatives of expectation values are not only important for the consideration of barren plateaus, but also to understand which gate in the circuit has significant influence on the expectation value.

\subsubsection{Review of previous results}
%In this section, we revisit prior findings related to barren plateaus. 
In the noiseless scenario, initial observations by McClean \textit{et al.}~\cite{McClean_2018} pointed out that if the parameter distribution underlying the parametrized quantum circuit forms a global $2$-design with respect to the Haar measure of $n$-qubit unitaries, then any associated cost function exhibits barren plateaus.
Furthermore, when modeling a noiseless parametrized quantum circuit (often referred to as an \emph{ansatz}) as a `local random quantum circuit,' composed of geometrically local two-qubit gates, where each gate is distributed according to the Haar measure, barren plateaus start to manifest at $O(n)$ depth. This is because studies by Brandão \textit{et al.}~\cite{Brand_o_2016} have demonstrated that at linear $O(n)$ depth in one-dimensional architectures, the distribution over such circuits becomes `approximately' a $2$-design.
Similar results have been extended to higher-dimensional quantum circuit architectures. Specifically, it has been shown by Harrow \textit{et al.}~\cite{Harrow_2023} that the `approximate' $2$-design property emerges at $O(n^{1/\mathrm{D}})$, where $\mathrm{D}$ represents the dimension of the lattice of the circuit.
Discussions concerning the relationship between barren plateaus and approximate notions of $2$-design can be found in Ref.~\cite{Holmes_2022}.
The influence of the locality of observables on the onset of barren plateaus has 
been explored in Ref.~\cite{Cerezo_2021}. 

In one-dimensional architectures, it has been observed that while cost functions associated with $O(1)$-local observables do not exhibit barren plateaus at logarithmic depth, cost functions associated with global observables manifest barren plateaus even at constant depth.
Furthermore, these results have been generalized in 
Ref.~\cite{napp2022quantifying}, where it has been noted that the gradient of the cost function decays exponentially with respect to the circuit depth and the Hamiltonian locality. These findings were established under the assumption that the $2$-qubit gates composing the circuits are distributed according to a unitary $2$-design.
Methods for avoiding or mitigating barren plateaus in noiseless scenarios have been proposed, primarily relying on specific heuristic-based initialization strategies~\cite{ Grant_2019,Sack_2022,Mele_2022,PhysRevResearch.5.L032040,rudolph2023synergy, Jain_2022, Holmes_2022,shi2024avoiding}, as well as by constraining the expressibility of the ansatz~\cite{schatzki2022theoretical, zhang2020trainability, Volkoff_2021, Pesah_2021,liu2022mitigating, Meyer_2023, park2023hamiltonian,park2024hardwareefficient,Zhang_2024}. This constraint can be achieved, for instance, through the utilization of symmetries~\cite{Larocca_2022,Meyer_2023}---from an intuitive perspective, these strategies aim to limit the expressiveness of the ansatz, rendering it less akin to a global $2$-design with respect to the Haar measure over the full $n$-qubit unitary group. Furthermore, it has been argued/conjectured that if one can prove absence of barren plateaus, then one should also be able to classically simulate the ansatz class~\cite{cerezo2023does}, either with purely classical resources or after an initial data acquisition phase, which
may require a quantum computer.

In the context of noisy scenarios, an important observation has been pointed out 
in Ref.~\cite{nibp}, revealing that both expectation values and gradients experience exponential decay in the circuit depth. Consequently, at linear depth, the expectation values and gradients of cost functions decay exponentially with respect to the number of qubits. This phenomenon has been dubbed `noise-induced barren plateaus'.
The results hold even without using randomness of the gates, i.e., for any fixed circuit.
Significantly, this latter study assumed the presence of a local depolarizing noise model, which is unital in nature. Strikingly, even when employing error mitigation strategies, it appears challenging to effectively counteract the emergence of noise-induced barren plateaus, as argued in Ref.~\cite{wang2021error}.

\textcolor{black}{A concurrent and independent work~\cite{singkanipa2024unitalnoisevariationalquantum} has studied the impact of noise beyond unital on the barren plateaus phenomenon. However, this study is restricted to the so-called HS-contractive noise, whereas, to the best of our knowledge, our work is the first to rigorously address a general kind of possibly non-unital local noise.}

As we will prove, the last \( O(\log(n)) \) layers of the circuit are trainable, meaning they do not suffer from (sub-)exponentially vanishing partial derivatives. This results in the norm of the gradient of the cost function being sufficiently large. However, we will also show that the partial derivatives taken before these logarithmic many layers in the circuit are negligible, which is why we refer to this as an `effective log-depth circuit'.

%\subsubsection{Noiseless scenario}
%\subsubsection{Unital noise scenario}

\subsection{Gradients: useful lemmas}
\label{gradients1}

We now give a formula to compute directly the partial derivative, which can be useful to handle calculations. However, one might also take it as an equivalent definition of partial derivative. 

\begin{lemma}[Partial derivative]
\label{le:deriv}
Let $\mu \in [m]$. Consider a parameterized $2$-qubit gate $ \exp(-i \theta_{\mu}H_\mu)$, positioned at the start of unitary layer $\mathcal{U}_{k}$, where $k\in[L]$ is the index of the layer where the gate is positioned in the circuit. We have
\begin{align}
\partial_\mu C= i\Tr\!\left(\Phi_{[1,k-1]}(\rho_0) \left[H_\mu ,\Phi^{*}_{[k,L]}(H)\right]\right),
\label{eq:der}
\end{align}
where we have denoted
\begin{align}
\Phi_{[a,b]} & \coloneqq  (\mathcal{V}_b^{\mathrm{single}} \circ \mathcal{N}^{\otimes n} \circ \mathcal{U}_b) \circ \cdots \circ (\mathcal{V}_a^{\mathrm{single}} \circ\mathcal{N}^{\otimes n} \circ \mathcal{U}_{a}), %, \\
%\Phi_{[1,k]} & \coloneqq   (\mathcal{V}_{k-1}^{\mathrm{single}}\circ \mathcal{N}^{\otimes n} \circ \mathcal{U}_{k-1}) \circ \cdots \circ (\mathcal{V}_1^{\mathrm{single}}\circ \mathcal{N}^{\otimes n} \circ \mathcal{U}_1).
\end{align}
for $a\le b \in [L]$.
\end{lemma}
\begin{proof}
We can write the cost function as
\begin{align}
C= \Tr\!\left(\Phi(\rho_0) H\right)=\Tr\!\left(\Phi_{[1,k]}\circ \Phi_{[k,L]}(\rho_0) H\right)=\Tr\!\left(\Phi_{[k,L]}(\rho_0) \Phi^*_{[1,k]}(H)\right).
\end{align}
%Here, the adjoint of $\Phi_{[k,L]}(H)$ can be expressed as
%\begin{align}
%\Phi^*_{[k,L]}(H)&=W(\theta_\mu)^{\dagger} \left((\mathcal{U}^{\mathrm{after} *}_{k}\circ \mathcal{N}^{* \otimes n})  \circ \cdots  \circ (\mathcal{U}^*_{L} \circ \mathcal{N}^{* \otimes n})\circ \mathcal{V}^{\mathrm{single}*}(H)\right)W(\theta_\mu) , 
%\end{align}
%where $W(\theta_\mu)  \coloneqq  \exp(-i \theta_{\mu}H_\mu)$.
By taking the partial derivative with respect to the parameter $\theta_\mu$, we 
have
\begin{align}
\partial_\mu C& =\Tr\!\left(\Phi_{[1,k]}(\rho_0) \partial_\mu(\Phi^*_{[k,L]}(H))\right)\\
\nonumber
&=i\Tr\!\left(\Phi_{[1,k]}(\rho_0) H_\mu \Phi^*_{[k,L]}(H)\right)-i\Tr\!\left(\Phi_{[1,k]}(\rho_0)  \Phi^*_{[k,L]}(H) H_\mu \right)\\
\nonumber
&= i\Tr\!\left(\Phi_{[1,k]}(\rho_0) \left[H_\mu ,\Phi^*_{[k,L]}(H)\right]\right),
\nonumber
\end{align}
where we have used the fact that $\partial_\mu \exp(-i\theta_\mu H_\mu)=-i H_\mu \exp(-i\theta_\mu H_\mu)$.
\end{proof}
From now on, when using the above formula for the partial derivatives, since we are considering second moment quantities, we will ignore the parametrized gates in the circuit, since they can be absorbed in the two-qubit $2$-design layers. 
We now show that the expected value of the partial derivative with respect any parameter is zero.

\begin{lemma}
\label{le:zeroDer}
The expected value of the partial derivative of the cost function is $0$ with respect any parameter, i.e.,
\begin{align}
    \Ex[\partial_\mu C]=0.
    \label{eq:meanGrad}
\end{align}
\end{lemma}
\begin{proof}
Due to left and right invariance, the $2$-qubit parameterized unitaries can be absorbed in the $2$-design unitaries.
Moreover, by Eq.~\eqref{eq:der}, we have
\begin{align}
\partial_\mu C= i\Tr\!\left(\Phi_{[1,k]}(\rho_0) \left[H_\mu ,\Phi^{*}_{[k,L]}(H)\right]\right).
\end{align}
Since $\Phi^{*}_{[k,L]}(H)$ ends with a layer of single-qubits $2$-design gates and these form a $1$-design (Lemma~\ref{le:qubitonedesign}) by taking the expected value only over that layer, we have
\begin{align}
    \Ex\left[\partial_\mu C\right]&=i\Tr\!\left(\Phi_{[k,L]}(\rho_0) \left[H_\mu ,\Ex\left(\Phi^{*}_{[k,L]}(H)\right)\right]\right)
    \\
    \nonumber
    &=i\Tr\!\left(\Phi_{[k,L]}(\rho_0) \left[H_\mu ,\Tr\!\left(\Phi^{*}_{[k,L]}(H)\right)\frac{I_n}{2^n}\right]\right)\\&=0,
    \nonumber
\end{align}
where we have used the first moment formula (Eq.~\eqref{eq:1momHaar}) and the fact that any operator commutes with the identity.
\end{proof}
The previous Lemma implies that $\Var[\partial_\mu C]=\Ex[(\partial_\mu C)^2]$, so we care only about the latter quantity from now on. We now present a lemma, similar in spirit to Lemma~\ref{le:mixham}, which will be useful to deal with upper and lower bounds of partial derivatives.

\begin{lemma}[Pauli mixing to Gradients]
\label{le:mixhamPARTIAL}
%Let $\mu \in [m]$ be the index of the parameter $\theta_\mu$ which parametrize a gate in the $k$-th layer.
Let $H_\mu$ be a $2$-local Hamiltonian.
Let $f(\cdot) \coloneqq i\Tr\!\left(\Phi_{[1,k]}(\rho_0) \left[H_\mu ,\Phi^{*}_{[k,L]}(\cdot)\right]\right)$ be an operator function. Then we have
\begin{itemize}
    \item Let  $H \coloneqq \sum_{P\in \{I,X,Y,Z\}^{\otimes n}} a_P P$, with $a_P\in \mathbb{R}$ for any $P\in \{I,X,Y,Z\}^{\otimes n}$. We have
    \begin{align}
    \Ex[(f(H))^2]= \sum_{P\in \{I,X,Y,Z\}^{\otimes n}} a^2_P \Ex[(f(P))^2].
\end{align}
%Or equivalently: 
%\begin{align}\Ex[(\partial_\mu C)^2]= \sum_{P\in \{I,X,Y,Z\}^{\otimes n}} a^2_P \Ex[(\partial_\mu C_P)^2],
%\end{align}
%where $C_P=\Tr(P\Phi(\rho_0))$. 

\item Moreover, for any $P\in \{I,X,Y,Z\}^{\otimes n}$, we have 
\begin{align}
    \Ex[(f(P))^2]=\frac{1}{3^{|P|}}\sum_{\substack{Q \in \{I,X,Y,Z\}^{\otimes n}:\\ \supp(Q)=\supp(P) }} \Ex[(f(Q))^2].
\end{align}
\end{itemize}

\end{lemma}
\begin{proof}
We have 
\begin{align}
\label{eq:backw}
(f(H))^2 & = \left( i\Tr\!\left(\Phi_{[1,k]}(\rho_0) \left[H_\mu ,\Phi^{*}_{[k,L]}(H)\right]\right) \right)^2 \\
\nonumber
%& = \left( \Tr\!\left(\Phi_{[1,k]}(\rho_0) \left[\Phi^{*}_{[k,L]}(H),H_\mu \right]\right) \right)^2 \\
& = \left( \sum_{P\in \{I,X,Y,Z\}^{\otimes n}} a_P \Tr\!\left(\Phi_{[1,k]}(\rho_0) \left[\Phi^{*}_{[k,L]}(P),H_\mu \right]\right) \right)^2 \\
\nonumber
& = \sum_{P,Q\in \{I,X,Y,Z\}^{\otimes n}} a_P  a_Q \Tr\!\left(\Phi_{[1,k]}(\rho_0) \left[\Phi^{*}_{[k,L]}(P),H_\mu \right]\right) \Tr\!\left(\Phi_{[1,k]}(\rho_0) \left[\Phi^{*}_{[k,L]}(Q),H_\mu \right]\right)\\
\nonumber
& = \sum_{P,Q\in \{I,X,Y,Z\}^{\otimes n}} a_P  a_Q \Tr\!\left(\Phi_{[1,k]}(\rho_0)^{\otimes 2} \left(\left[\Phi^{*}_{[k,L]}(P),H_\mu \right]\otimes \left[\Phi^{*}_{[k,L]}(Q),H_\mu \right]\right)\right).
\nonumber
\end{align}
We now consider the expected value of this quantity with respect to the final unitary layer in $\Phi^{*}_{[k,L]}$ (specifically, the layer that acts directly on $P$ and $Q$). Such expected value reduces to computing the expected value
\begin{align}
    \Ex\left(\left[\Phi^{*}_{[k,L]}(P),H_\mu \right]\otimes \left[\Phi^{*}_{[k,L]}(Q),H_\mu \right]\right).
\end{align}
By expanding the two commutators, we have
\begin{align}
\begin{aligned}
&\left[\Phi^{*}_{[k,L]}(P),H_\mu \right] \otimes \left[\Phi^{*}_{[k,L]}(Q),H_\mu \right] \\
&= (\Phi^{*}_{[k,L]}(P) \otimes \Phi^{*}_{[k,L]}(Q))(H_\mu\otimes H_\mu) - ( I_n \otimes H_\mu)(\Phi^{*}_{[k,L]}(P) \otimes  \Phi^{*}_{[k,L]}(Q)) ( H_\mu \otimes I_n) \\
&\quad -  ( H_\mu  \otimes I_n) (\Phi^{*}_{[k,L]}(P) \otimes  \Phi^{*}_{[k,L]}(Q)) ( I_n \otimes H_\mu) + (H_\mu\otimes H_\mu)(\Phi^{*}_{[k,L]}(P) \otimes \Phi^{*}_{[k,L]}(Q)) .
\end{aligned}
\label{eq:substGRAD}
\end{align}
Consequently, our attention can be directed solely towards the  expression
\begin{align}
\Ex(\Phi^{*}_{[k,L]}(P)\otimes \Phi^{*}_{[k,L]}(Q) ) &= \Ex\left(\Phi^{*\otimes 2}_{[k,L]}\left(\Ex_{\mathcal{V}}(\mathcal{V}^{\mathrm{single}}(P)\otimes \mathcal{V}^{\mathrm{single}}(Q) )\right)\right)\\
\nonumber
&= \delta_{P,Q}\Ex\left(\Phi^{*\otimes 2}_{[k,L]}\left(\Ex_{\mathcal{V}}(\mathcal{V}^{\mathrm{single}}(P)\otimes \mathcal{V}^{\mathrm{single}}(P) )\right)\right)\\
\nonumber
&= \delta_{P,Q}\frac{1}{3^{|P|}}\sum_{\substack{R \in \{I,X,Y,Z\}^{\otimes n}:\\ \supp(R)=\supp(P) }}\Ex\left(\Phi^{*\otimes 2}_{[k,L]}\left(R\otimes R)\right)\right)\\
\nonumber
&= \delta_{P,Q}\frac{1}{3^{|P|}}\sum_{\substack{R \in \{I,X,Y,Z\}^{\otimes n}:\\ \supp(R)=\supp(P) }}\Ex\left(\Phi^{*}_{[k,L]}(R)\otimes \Phi^{*}_{[k,L]}(R)\right)\nonumber,
\end{align}
where in the first step we singled out `for free' a layer of Haar random gates from $\Phi^{*}_{[k,L]}$, in the second step we applied the Pauli mixing formula Eq.~\eqref{eq:paulimixing} for each of the single qubits gates (similarly as done in Lemma~\ref{le:mixham}). Therefore, by substituting in Eq.~\eqref{eq:substGRAD} and repeating the steps backwards, we can conclude. 
\end{proof}
Because of the previous lemma, computing the variance of a cost function partial derivative defined with respect an Hermitian operator reduces to computing the variance of a cost function partial derivative defined with respect a Pauli operator. Consequently, we have the following corollary.
\begin{corollary}[Partial derivative variance of an Hamiltonian]
\label{le:varHAM}
    Let $\mu \in [m]$ be the index of the parameter $\theta_\mu$ which parametrize a gate in the $k$-th layer. Let  $H \coloneqq \sum_{P\in \{I,X,Y,Z\}^{\otimes n}} a_P P$, with $a_P\in \mathbb{R}$ for any $P\in \{I,X,Y,Z\}^{\otimes n}$. We have
\begin{align}\Var[\partial_\mu C]= \sum_{P\in \{I,X,Y,Z\}^{\otimes n}} a^2_P \Var[\partial_\mu C_P], 
\end{align}
where $C_P\coloneqq \Tr(P\Phi(\rho_0))$ with $\Phi$ and $\rho_0$ are respectively the noisy quantum circuit and the initial state. 
\end{corollary}
\begin{proof}
    This follows immediately from Lemma~\ref{le:zeroDer} and Lemma~\ref{le:mixhamPARTIAL}.
\end{proof}
We now show a worst-case upper bound on the $\alpha$-th order partial derivative that will be useful later. 
\begin{lemma}[$\alpha$-th order Partial derivative upper bound]
\label{le:derupp}
The $\alpha$-th order partial derivative with respect to the parameter $\theta_\mu$ is upper bounded by
\begin{align}
   \left|\partial^{\alpha}_\mu C\right|\le 2^{\alpha} \norm{H}_{\infty}\norm{H_\mu }^{\alpha}_{\infty}
    \label{eq:meanGrad}
\end{align}
\end{lemma}
\begin{proof}
We have $C=\Tr(\Phi(\rho_0)H)=\Tr\!\left(\Phi_{[1,k]}(\rho_0) \Phi^*_{[k,L]}(H)\right)$. Thus
\begin{align}
    \partial^{\alpha}_\mu C=\Tr\!\left(\Phi_{[1,k]}(\rho_0) \,\partial^{\alpha}_\mu\Phi^*_{[k,L]}(H)\right).
\end{align}
Because of H\"older inequality, we have $|\partial^{\alpha}_\mu C|\le \norm{\partial^{\alpha}_\mu\Phi^*_{[k,L]}(H)}_{\infty}$, where we also used that the one-norm of a quantum state is one. We now prove by induction that
\begin{align}
    \norm{\partial^{\alpha}_\mu\Phi^*_{[k,L]}(H)}_{\infty}\le 2^{\alpha} \norm{H}_{\infty}\norm{H_\mu }^{\alpha}_{\infty}.
\end{align}
For $\alpha=1$, we have
\begin{align}
\partial_\mu\Phi^*_{[k,L]}(H)= i \left[H_\mu ,\Phi^{*}_{[k,L]}(H)\right],
\end{align}
where we have used the fact that $\partial_\mu \exp(-i\theta_\mu H_\mu)=-i H_\mu \exp(-i\theta_\mu H_\mu)$ as done in the proof of Lemma~\ref{le:deriv}.
Thus, 
\begin{align}
\norm{\partial_\mu\Phi^*_{[k,L]}(H)}_{\infty}\le 2 \norm{H_\mu}_{\infty}\norm{\Phi^{*}_{[k,L]}(H)}_{\infty}\le 2 \norm{H_\mu}_{\infty}\norm{H}_{\infty},
\end{align}
where we have used the triangle inequality, submultiplicativity of the $p$-norms and in the last step we have used the inequality $\norm{\Phi^{*}(O)}_\infty \le \norm{O}_\infty$ (see, e.g., Ref.~\cite{Bhatia2007PositiveDM}) valid for all operators $O$. This shows the base case.
For $\alpha>1$, we have
\begin{align}
\norm{\partial^\alpha_\mu\Phi^*_{[k,L]}(H)}_{\infty}
=  \norm{\partial^{\alpha-1}_\mu i \left[H_\mu ,\Phi^{*}_{[k,L]}(H)\right]}_{\infty}=  \norm{  \left[H_\mu ,\partial^{\alpha-1}_\mu \Phi^{*}_{[k,L]}(H)\right]}_{\infty} \le 2 \norm{H_\mu}_{\infty}\norm{\partial^{\alpha-1}_\mu \Phi^{*}_{[k,L]}(H)}_{\infty},
\end{align}
where in the last step we have used triangle inequality and submultiplicativity. We can conclude by the using the induction step.
\end{proof}
We introduce a precise definition of the standard notion of the light-cone of an observable with respect to a quantum channel (typically representing a quantum circuit).

\begin{definition}[Light cone]
Let $H$ be an Hermitian operator and $\Phi$ be a quantum channel. The light-cone of $H$ with respect to  $\Phi$ is defined as 
    \begin{align}
       \mathrm{Light}(\Phi,H) \coloneqq  \supp
        (\Phi^*(H)),
    \end{align}
    where $\supp(\cdot)$ is defined in our notation section.
\end{definition}

\begin{figure}[h]
\centering
\includegraphics[width=0.43\textwidth]{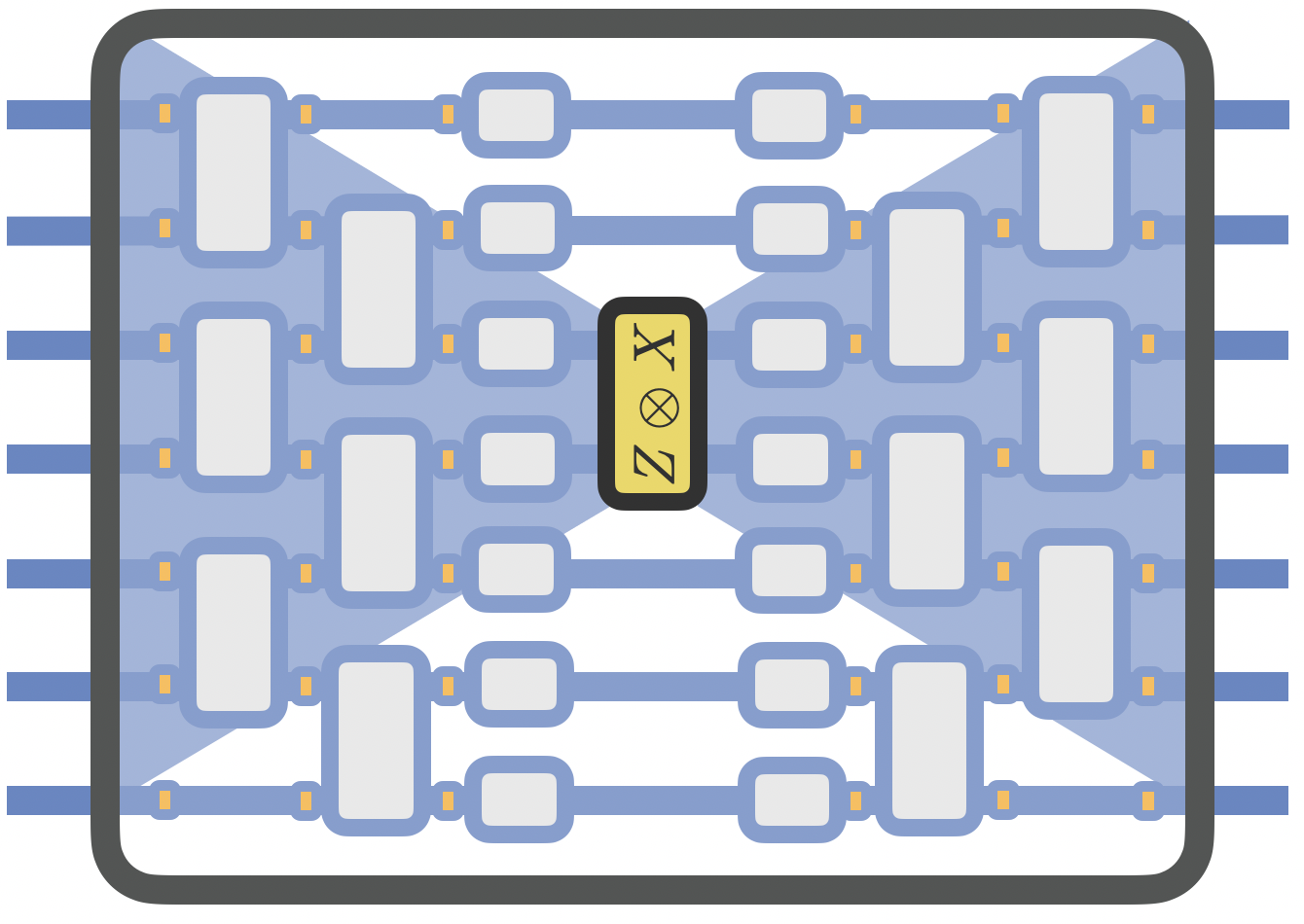}
\caption{The light-cone of a local Pauli observable with respect to  $\Phi_{[1,L]}$ is the set of qubits within the shaded area. The (noisy) gates outside the blue shaded area are contracted trivially due to the fact that any adjoint channel is unital, and thus cannot influence the expectation value of the Pauli.}
\label{Fig:circ33}
\end{figure}
In the subsequent subsection, we need to consider the light cone not with respect to only a specific quantum circuit, but with respect to a family of quantum circuits provided by the support of a considered random quantum circuits probability distribution, denoted as $\mathcal{F}$ (recall that the support of a random variable is defined as the set of all values for which the probability density function is strictly greater than zero). Formally, we define:
\begin{align}
       \mathrm{Light}_{\mathcal{F}}(H) \coloneqq \bigcup_{\Phi \in \mathcal{F}}\supp
        (\Phi^*(H)).
\end{align}
In particular, we consider the family $\mathcal{F}_k$ of quantum circuits corresponding to the support of the probability distribution associated with $\Phi_{[k,L]}$, where $k\in [L]$. To streamline the notation, we refer to $\mathrm{Light}_{\mathcal{F}_k}(H)$ as the `light-cone of $H$ with respect to $\Phi_{[k,L]}$'. 
In Fig.~\ref{Fig:circ33}, we provide a graphical example for a one-dimensional geometrical local quantum circuit.
We now give the following Lemma, which will be useful later on.
\begin{lemma}[Partial derivative is zero outside the light cone]
\label{le:outside}
Let $H$ be an Hermitian operator.
Let $\mu \in [m]$. Consider a parameterized $2$-qubit gate $ \exp(-i \theta_{\mu}H_\mu)$, positioned in the $k$-th layer, such that its Hamiltonian generator $H_{\mu}$ has support outside the light cone of $\Phi_{[k,L]}$ with respect to $H$. Then, the partial derivative is zero $\partial_\mu C=0.$.
\end{lemma}
\begin{proof}
In words, the partial derivative must be zero since the cost function does not depend effectively by the gates outside the light-cone (because they contract trivially).
However, formally this can be seen as follows. By using Eq.~\eqref{eq:der}, we have
     \begin{align}
\partial_\mu C= i\Tr\!\left(\Phi_{[1,k]}(\rho_0) \left[H_\mu ,\Phi^{*}_{[k,L]}(H)\right]\right).
\end{align}
Since $H_{\mu}$ has support outside the light cone of $\Phi_{[k,L]}$ with respect to $H$, by definition of light-cone it follows that $\left[H_\mu ,\Phi^{*}_{[k,L]}(H)\right]=0$.
\end{proof}

\subsection{Absence of barren plateaus, but only few layers are trainable}
\label{sub:absenceBP}
In this subsection, we will show that the gates in the last $\Theta(\log(n))$ layers (in the light cone of a local observable) are the only trainable gates of the circuits. This also implies that the expectation value of an observable can be significantly influenced only by such gates in the last layers.
We first state our main claims here, which will then be detailed further. Leveraging Corollary~\ref{le:varHAM}, we focus on cost functions associated with Pauli observables instead of general Hermitian operator, without loss of generality.
In Subsection~\ref{subsub:ubPAR}, we present the following upper bound assuming constant noise parameters:
\begin{theorem}[Layers before $\Theta(\log(n))$ layers are not trainable]
\label{th:ubvarianceINT}
Let $C\coloneqq \Tr(P\Phi(\rho_0))$ be the cost function, where $P\in \{I,X,Y,Z\}^{\otimes n}$, $\rho_0$ is an arbitrary initial state, and $\Phi$ is a noisy quantum circuit of depth $L$. We assume that the noise is not a unitary channel.
Let $\mu$ denote a parameter of a gate in the $k$-th layer of the circuit. Then, we have
\begin{align}
        \Var[\partial_\mu C]\le \exp(-\Omega(|P|+L-k)).
\end{align}
\end{theorem}
This result immediately implies that the gates before the last $\Theta(\log(n))$ layers are not trainable. Moreover, it directly implies barren plateaus for global cost functions.
\begin{corollary}[Global cost function induced barren plateaus]
\label{th:ubvarianceINT}
Let $C$ be the cost function associated with a Pauli $P$ with $|P|=\Theta(n)$. Then, we have $\Var[\partial_\mu C]\le \exp(-\Omega(n))$.
\end{corollary}
In subsubsection~\ref{subsub:ubPAR}, we establish the following lower bound assuming constant noise parameters.
\begin{theorem}[Last $\Theta(\log(n))$ layers do matter]
\label{th:variance_PD}
Let $C\coloneqq \Tr(P\Phi(\rho_0))$ be the cost function, where $P\in \{I,X,Y,Z\}^{\otimes n}$, $\rho_0$ is an arbitrary state, and $\Phi$ is a non-unital noisy quantum circuit of depth $L$ . We assume noise is not a replacer channel. Let $\mu$ denote a parameter of a gate in the $k$-th layer of the circuit. Then, if the support of such a gate is contained in the light cone of $\Phi_{[k,L]}$ with respect to the Pauli $P$, we have
\begin{align}
       \Var[\partial_\mu C] \ge \exp  \!\left(- O(|P|(L-k))\right),
\end{align}  
otherwise, if the support of the parametrized gate is outside the light cone, we have $\Var[\partial_\mu C] =0$.
\end{theorem}
Note that the variance upper and lower bounds are matched for local cost functions (i.e., $|P|=O(1)$).
Moreover, the latter theorem leads to the following corollary affirming the absence of barren plateaus for local cost functions:
\begin{corollary}[Absence of barren plateaus for local cost functions]
\label{th:bpINTappe}
Let $C$ be a cost function associated with a local Pauli $P$. We assume non-unital noise and that is not a replacer channel. Then, we have 
\begin{align}
    \Ex\!\left[\norm{\nabla C}^2_{2}\right]\ge \Omega(1).
\end{align}
\end{corollary}
However, this absence of barren plateaus in arbitrarily deep quantum circuit is only due to the last $\Theta(\log(n))$ layers which significantly influence the expectation value of local observables. 
Furthermore, in Subsection~\ref{sub:unital}, we show improved upper bound for the onset of 
barren plateaus in the unital noise scenario, improving upon previous works~\cite{nibp}.

\subsubsection{Partial derivative upper bound: Layers before the last $O(\log(n))$ are not trainable} 
\label{subsub:ubPAR}
We are now going to show the upper bound on the partial derivative. Here, we do not make any assumption on the geometrical locality of the circuit. 
\begin{proposition}[Partial derivative upper bound]
\label{th:ubvariance}
Let $C\coloneqq \Tr(P\Phi(\rho_0))$ be the cost function, where $P\in \{I,X,Y,Z\}^{\otimes n}$, $\rho_0$ is an arbitrary initial state, and $\Phi$ is a noisy quantum circuit of depth $L$. We assume that the noise is not a unitary channel. 
Let $\mu$ denote a parameter of a gate in the $k$-th layer of the circuit. Then, we have
\begin{align}
        \Var[\partial_\mu C]\le 4 
 c^{(|P|+L-k-1)}.
\end{align}
\end{proposition}
\begin{proof}
\textbf{Proof method 1:}
We show first a shorter and more immediate proof method, which yields a slightly worse upper bound $O(c^{(|P| +L-k-1)/3})$ albeit always with the desired exponential scaling. Let $K\coloneqq L-k$.
We have $\Var[\partial_\mu C]=\Ex[(\partial_\mu C)^2]$ due to Eq.~\eqref{eq:meanGrad}. 
Due to the Taylor remainder theorem and for any three differentiable functions $f(x)$, this relationship can be expressed using standard finite-difference formulas (see, e.g., \href{https://www.dam.brown.edu/people/alcyew/handouts/numdiff.pdf}{here}):
\begin{align}
    |\partial f(x)| \le \left|\frac{f(x+h)-f(x-h)}{2h}\right| + \frac{h^2}{6}|\mathrm{sup}(\partial^3 f)|,
\end{align}
for any $h\in [0,\infty)$.
Hence
\begin{align}
    (\partial f(x))^2& \le \left|\frac{f(x+h)-f(x-h)}{2h}\right|^2 + \frac{h^4}{36}|\mathrm{sup}(\partial^3 f)|^2 +2 \left|\frac{f(x+h)-f(x-h)}{2h}\right|\frac{h^2}{3}|\mathrm{sup}(\partial^3 f)|\\
    \nonumber
    &\le \left|\frac{f(x+h)-f(x-h)}{2h}\right|^2 + \frac{h^4}{36}|\mathrm{sup}(\partial^3 f)|^2 + \frac{h}{3}\left|\mathrm{sup}(f)\right||\mathrm{sup}(\partial^3 f)|.
    \nonumber
\end{align}
By using this for our function $C$ with respect to the parameter $\theta_\mu$, taking the expected values both terms and using Lemma~\ref{le:derupp}, we have
\begin{align}
\label{eq:derrrr}
   \Ex(\partial_\mu C )^2&\le \Ex \left|\frac{\Tr(P \Phi_{[L-K+1,L]}(\rho))-\Tr(P \Phi_{[L-K+1,L]}(\sigma))}{2h}\right|^2 + 2^6 \frac{h^4}{36}\norm{H_\mu}^6_{\infty} + 2^3 \frac{h}{3}\norm{H_\mu}^3_{\infty}\\
   \nonumber
    &\le \frac{1}{4h^2}c^{K+|P|-1} + 4 h^4\norm{H_\mu}^6_{\infty} + 4 h\norm{H_\mu}^3_{\infty}\\
    \nonumber
    &\le \frac{1}{4}c^{(K+|P|-1 )/3} +8 c^{(K+|P|-1)/3}\\
    &\le 9 c^{(K+|P|-1)/3},
    \nonumber
\end{align}
where we have defined $\rho$ and $\sigma$ to be defined as the state $\Tr(P \Phi_{[1,L-K]}(\rho_0))$ computed respectively respectively in $\theta_\mu +h$ and $\theta_\mu -h$, and we have chosen $h\coloneqq c^{(K+|P|)/3}$.

\noindent \textbf{Proof method 2:} We now establish a tighter upper bound using the partial derivative formula involving the commutator (Eq.~\eqref{eq:der}). The proof follows a similar spirit to the one used in Proposition~\ref{prop:expdecayP} (i.e., effective depth). Instead of using Lemma~\ref{le:mixham} (as in the effective depth proof), we employ the analogous Lemma~\ref{le:mixhamPARTIAL}.
First, we have
\begin{align}
    \Var[\partial_\mu C] = \Ex[(\partial_\mu C)^2] = \Ex[f_0(P)^2],
\end{align}
where we define the function
\begin{align}
    f_j(\cdot)\coloneqq i\Tr\!\left(\Phi_{[1,k]}(\rho_0) \left[H_\mu ,\Phi^{*}_{[k,L-j]}(\cdot)\right]\right).
\end{align}
Using Lemma~\ref{le:mixhamPARTIAL} and averaging over the last layer of single qubit gates in $\Phi^{*}_{[k,L]}$, we have
\begin{align}
   \Ex[(f_0(P)^2] = \frac{1}{3^{|P|}}\sum_{\substack{Q \in \{I,X,Y,Z\}^{\otimes n}:\\ \supp(Q)=\supp(P) }} \Ex[(f_0(Q))^2].
\end{align}
Let us focus on $\Ex[(f_0(Q))^2]$. Taking the adjoint of the last layer of noise on $Q$, we get
\begin{align}
    \mathcal{N}^{* \otimes n}(Q) = \bigotimes_{j \in \supp(Q)} (t_{Q_j }I_j+ D_{Q_j}Q_j) = \sum_{a\in \{0,1\}^{|Q|}}\bigotimes_{j \in \supp(Q)} (t^{a_j}_{Q_j }  D^{1-a_j}_{Q_j}Q^{1-a_j}_j).
\end{align}
We define the function $f_j^{\prime}(\cdot)$ as
\begin{align}
    f^{\prime}_j(\cdot)\coloneqq i\Tr\!\left(\Phi_{[1,k]}(\rho_0) \left[H_\mu ,\Phi^{\prime *}_{[k,L-j]}(\cdot)\right]\right),
\end{align}
where $\Phi^{\prime}_{[k,L-j]}$ is equal to $\Phi_{[k,L-j]}$ but without the last layer of single qubit gates and noise.
Applying Lemma~\ref{le:mixhamPARTIAL} again, 
we have
\begin{align}
    \Ex[(f_0(Q))^2] & = \sum_{a\in \{0,1\}^{|Q|}}\prod_{j \in \supp(Q)} t^{2a_j}_{Q_j}D^{2(1-a_j)}_{Q_j} \Ex[(f_0^{\prime}(\bigotimes_{j \in \supp(Q)}Q^{1-a_j}_j))^2]\\
    \nonumber
    &\le \sum_{a\in \{0,1\}^{|Q|}}\prod_{j \in \supp(Q)} t^{2a_j}_{Q_j}D^{2(1-a_j)}_{Q_j} \max_{R\in \{I,X,Y,Z\}^{\otimes n}} \Ex[(f_0^{ \prime}(R))^2].
    \nonumber
\end{align}
Substituting, we arrive at
\begin{align}
   \Ex[(f_0(P)^2] &= \frac{1}{3^{|P|}}\sum_{\substack{Q \in \{I,X,Y,Z\}^{\otimes n}:\\ \supp(Q)=\supp(P) }} \Ex[(f_0(Q))^2]\\ 
   \nonumber
   &\le \frac{1}{3^{|P|}}\sum_{\substack{Q \in \{I,X,Y,Z\}^{\otimes n}:\\ \supp(Q)=\supp(P) }}\sum_{a\in \{0,1\}^{|Q|}}\prod_{j \in \supp(Q)} t^{2a_j}_{Q_j}D^{2(1-a_j)}_{Q_j} \max_{R\in \{I,X,Y,Z\}^{\otimes n}} \Ex[(f_0^{ \prime}(R))^2]\\
   \nonumber
   &=c^{|P|} \max_{R\in \{I,X,Y,Z\}^{\otimes n}} \Ex[(f_0^{ \prime}(R))^2],
   \nonumber
\end{align}
where we have used the multinomial theorem together with the fact that
\begin{align}
    c=\frac{1}{3^{|P|}}(\|\bold{D}\|^2_2+\|\bold{t}\|^2_2).
\end{align}
As in the proof of Proposition~\ref{prop:expdecayP}, we can assume that the maximum over Pauli in the latter equation is achieved by a Pauli different from the identity (otherwise the RHS would be zero). Moreover, we can assume now that all the two-qubit gates in the circuit are Clifford, as we are computing a second moment, and the Cliffords form a $2$-design. Thus, the two qubit gates of the circuit will also map Pauli to Pauli.
Therefore, the Pauli above will be mapped by the two-qubits Clifford to another Pauli still different from the identity. Since now we have a circuit that ends with a layer of single qubits $2$-design unitaries, which are preceded by a noise layer and a layer of two-qubits $2$-design gates, we are in the same situation we faced at the beginning of the proof with
\begin{align}
    \Ex[(f_0(P)^2]&\le c^{|P|} \max_{Q\in \{I,X,Y,Z\}^{\otimes n} \setminus I_n} \Ex[(f_{1}(Q)^2].
\end{align}
So reiterating the argument to the next layers, and using that the Pauli weight of the considered Pauli at each iteration is at least one, we have
\begin{align}
    \Ex[(f_0(P)^2]&\le c^{|P|+L-k-1} \max_{Q\in \{I,X,Y,Z\}^{\otimes n} \setminus I_n} \Ex[(f_{L-k}(Q)^2].
    \end{align}
By using the definition of $f_{L-k}(\cdot)$, we have
\begin{align}
\label{eq:usefulunit}
    \Ex[(f_0(P)^2]&\le c^{|P|+L-k-1}\Ex\max_{Q\in \{I,X,Y,Z\}^{\otimes n} \setminus I_n} (i\Tr\!\left(\Phi_{[1,k]}(\rho_0) \left[H_\mu ,Q\right]\right))^2\\
    \nonumber
    &\le c^{|P|+L-k-1}\max_{Q\in \{I,X,Y,Z\}^{\otimes n} \setminus I_n} \Ex\norm{\Phi_{[1,k]}(\rho_0)}^2_1 \norm{\left[H_\mu ,Q\right]}_{\infty}^2\\
    \nonumber
    &\le 4 c^{|P|+L-k-1}\max_{Q\in \{I,X,Y,Z\}^{\otimes n} \setminus I_n}  \norm{H_\mu}_{\infty}^2\norm{Q}_{\infty}^2\\
    \nonumber
    &\le 4c^{|P|+L-k-1},
    \nonumber
\end{align}
where we have used the H\"older inequality in the second step, submultiplicativity of the infinity norm in the third step, and in the last step the fact that all the involved norms are $\le 1 $.

\end{proof}
We point out that the previous statement 
could also have been proved using the Parameter Shift Rule~\cite{crooks2019gradients} assuming a restricted class of parameterized gates, that is, of the form $\exp(i\theta_\mu H_\mu)$, with $H_\mu$ such that $H_\mu^2=I$. But for the sake of generality we decided to use the proof methods presented. In this connection, note that the parameter shift rule also applies to noisy circuits, as can be seen, e.g., by making use standard Stinespring dilation arguments.

We point out that by applying Chebyshev's inequality, the upper bound on the variance can be translated into the probability statement:
\begin{align}
    \mathrm{Prob}\left(\left| \partial_\mu C \right| > \varepsilon \right) \le \frac{\mathrm{Var}[\partial_\mu C]}{\varepsilon^2} \le \frac{4}{\varepsilon^2} c^{|P|+L-k-1}.
\end{align}
This equation implies that the probability of sampling a point in the parameters space such that the absolute value of the partial derivative is greater than $\varepsilon$ decays exponentially with both the Pauli weight $|P|$ and $L-k$, which is the distance from the end of the circuit to the layer where the partial derivative is taken.

The previous proposition also directly implies that cost functions associated with global Pauli operators (i.e., $\Theta(n)$ Pauli weight) have all partial derivatives exponentially vanishing in the number of qubits.
\begin{corollary}[Global cost function induced barren plateaus]
\label{th:ubvarianceGLOB}
Let $C$ be a cost function associated with a global Pauli operator, i.e., with $|P|=\Theta(n)$. Let $\mu$ denote a parameter of a gate in any of the layers. Then, we have
\begin{align}
        \Var[\partial_\mu C]\le \exp(-\Omega(n)).
\end{align}
\end{corollary}

\subsubsection{Partial derivative lower bound: the last $\Theta(\log(n))$ layers are the only trainable}
\label{subsub:lbPAR}
Here, we establish a lower bound on the partial derivative variance, valid any fixed circuit architecture (e.g., in constant dimension or all-to-all connectivity). The proof technique employed here is novel and may be of independent interest, allowing lower bounds on other second or third moment quantities of noisy random quantum circuits.
In summary, dealing with a second-moment quantity and considering all $2$-qubit gates in the circuit as local $2$-designs (effectively Clifford gates), we condition on specific \emph{Clifford choices} among various combinations to obtain a non-trivial lower bound.

Here, we assume that the noise is non-unital (i.e., $\|\bold{t}\|_2=\Theta(1)$) and that the noise is not a replacer channel (i.e., the noise parameter $\|\bold{D}\|_2$ is a non-zero constant).

\begin{proposition}[Partial derivative lower bound]
\label{th:variance_PDlb}
Let $C\coloneqq \Tr(P\Phi(\rho_0))$ be the cost function, where $P\in \{I,X,Y,Z\}^{\otimes n}$, $\rho_0$ is an arbitrary initial state, and $\Phi$ is a non-unital noisy quantum circuit of depth $L$. We also assume that the noise is not a replacer channel (otherwise, any partial derivative would be zero). Let $\mu$ denote the parameter $\theta_{\mu}$ of the gate $\exp(-i \theta_{\mu} H_{\mu})$ in the $k$-th layer of the circuit. Then, if the support of such a gate is contained in the light cone of $\Phi_{[k,L]}$ with respect to the Pauli $P$, we have
\begin{align}
       \Var[\partial_\mu C] \ge \exp  \!\left(- \Theta(|P|(L-k))\right),
\end{align}  
otherwise, if the support of the parametrized gate is outside the light cone, we have $\Var[\partial_\mu C] =0$.
\end{proposition}
\begin{proof}
If the support of the parametrized gate is outside the light cone, $\partial_\mu C = 0$, as stated in Lemma~\ref{le:outside}. Therefore, we focus on the case in which the gate is within the light cone.
By employing Lemma~\ref{le:zeroDer} and~\ref{le:deriv} to express the variance of the partial derivative, we have
\begin{align}
    \Var[\partial_\mu C] = \Ex[(\partial_\mu C)^2] = \Ex[f_0(P)^2].
\end{align}
Here, $f_j(\cdot)\coloneqq i\Tr\!\left(\Phi_{[1,k]}(\rho_0) \left[H_\mu ,\Phi^{*}_{[k,L-j]}(\cdot)\right]\right)$. % Based on this, one might be tempted to directly apply Theorem~\ref{th:variance} to lower-bound the variance of the expectation value $\Tr\!\left(\Phi_{[1,k]}(\rho_0)  \tilde{O}\right)$ by $\exp(-\Theta(|\tilde{O}|))$. However, the analysis becomes more intricate as $\tilde{O}$ could be zero for certain gate configurations, and Theorem~\ref{th:variance} assumes non-zero observables. Therefore, a more fine-grained analysis is required.
By applying Lemma~\ref{le:mixhamPARTIAL} and averaging over the last layer of single-qubit gates in $\Phi^{*}_{[k,L]}$, we arrive at:
\begin{align}
    \Var[\partial_\mu C]= \frac{1}{3^{|P|}}\sum_{\substack{R \in \{I,X,Y,Z\}^{\otimes n}:\\ \supp(R)=\supp(P) }} \Ex[(f_0(R))^2]\ge \frac{1}{3^{|P|}}\!\left(\Ex[(f_0(P_X))^2] + \Ex[(f_0(P_Y))^2] + \Ex[(f_0(P_Z))^2]\right),
\end{align}
where $P_X\coloneqq \bigotimes_{j \in \supp(P)} X_j$, and $P_Y$, $P_Z$ are similarly defined. Focusing on $\Ex[(f_0(P_X))^2]$, we define the function $f_j^{\prime}(\cdot)$ as
\begin{align}
    f^{\prime}_j(\cdot)\coloneqq i\Tr\!\left(\Phi_{[1,k]}(\rho_0) \left[H_\mu ,\Phi^{\prime *}_{[k,L-j]}(\cdot)\right]\right),
\end{align}
with $\Phi^{\prime}_{[k,L-j]}$ identical to $\Phi_{[k,L-j]}$ but lacking the last layer of single-qubit gates and noise. Taking the adjoint of the last layer of noise on $P_X$ and applying Lemma~\ref{le:mixhamPARTIAL} again, we obtain
\begin{align}
    \Ex[(f_0(P_X))^2] & = \sum_{a\in \{0,1\}^{|P_X|}}\prod_{j \in \supp(P_X)} t^{2a_j}_{(P_{X})_{j}}D^{2(1-a_j)}_{(P_{X})_{j}} \Ex[(f_0^{\prime}(\bigotimes_{j \in \supp(P_X)}(P_{X})_j^{1-a_j}))^2]
    \nonumber
    \\
    &\ge D_X^{2|P|} \Ex[(f_0^{\prime}(P_X)^2].
    \label{eq:prooflbpar}
\end{align}
Now, we delve into the technical part of this proof. As is customary when dealing with second-moment quantities, we treat our circuits as random Clifford circuits.
The ultimate goal is to ensure that the commutator $\left[H_\mu ,\Phi^{*}_{[k,L-j]}(\cdot)\right]$ is non-zero for some Clifford gate instances. To achieve this, we fix \emph{some} of the $2$-qubit Clifford gates in the circuit. However, caution is required not to fix all the Clifford gates, as this would result in unfavorable scaling. Specifically, each $2$-qubit Clifford $C_{\mathrm{fixed}}$ that we fixed in the circuit contributes a factor of $|\mathrm{C}_2|^{-1}$, where $|\mathrm{C}_2|$ is the size of the $2$-qubit Clifford group $\mathrm{C}_2$. This is expressed by the lower bound
\begin{align}
    \mathbb{E}_{C\sim \mathrm{C}_2} [g(C)]=\frac{1}{|\mathrm{C}_2|}\sum_{C \in \mathrm{C}_2} (g(C))^2 \ge  \frac{1}{|\mathrm{C}_2|}  (g(C_{\mathrm{fixed}}))^2,
\end{align}
for any real function $g(\cdot)$.
Now, we proceed to fix the Cliffords in the circuit.
The strategy involves fixing a few Cliffords such that: 1) one of the single-qubit Pauli in the Pauli decomposition of $P_X$ is connected with $H_{\mu}$ by a path of Clifford gates (which will be responsible for making the commutator non-zero), 2) the chosen Cliffords `protect the Pauli', i.e., they ensure the Pauli weight $|P|$ does not increase throughout the application of the unitary layer, 3) at each layer iteratively, as done in Eq.~\eqref{eq:prooflbpar}, we select only the Pauli operators that have all $X$ in their support that arises when we take the adjoint of the noise. 
Thus, this would give rise to the lower bound
\begin{align}
\Ex[(f_0(P_X))^2] \ge \frac{1}{|\mathrm{C}_2|^{\text{\# fixed clifford}}} D_X^{2|P|(L-k+1)}\Ex[(f_k( \mathcal{V}^{\mathrm{single}}(\tilde{P}_X))^2],
\end{align}
where $\tilde{P}_X$ represents the $P_X$ Pauli operator that has been mapped by all the Clifford circuit, and $\mathcal{V}^{\mathrm{single}}$ is a single-qubit layer of random gates. The term $(D_X^{2|P|})^{(L-k+1)}$ represents the factor obtained at each of the $L-k+1$ layers when encountering a layer of noise, applying Eq.~\eqref{eq:prooflbpar}, and utilizing the fact that the Pauli weight does not increase in the `Clifford path'.
Note that
\begin{align}
    f_k(\mathcal{V}^{\mathrm{single}}(\tilde{P}_X))=\left(\Phi_{[1,k]}(\rho_0) i \left[H_\mu ,\mathcal{V}^{\mathrm{single}}(\tilde{P}_X)\right]\right).
\end{align}
We now show the existence of this particular Clifford gates choice that satisfy the listed desiderata, by fixing some of the $2$-qubit gates to be the identity gate or the SWAP gate (noting that the SWAP gate is Clifford, as it can be expressed as a combination of $3$ CNOT gates).

Since $H_{\mu}$ is in the light-cone with respect to $\Phi_{[k,L]}$, by definition, there must exist a path of $2$-qubit (Clifford) gates that connects $H_\mu$ with one of the single-qubit Paulis $X$ appearing in the tensor product decomposition of $P_X$. Let us choose one such path connecting $H_{\mu}$ with a specific single-qubit Pauli $X$.
We can fix each $2$-qubit gate in this path to be the SWAP gate or the identity gate, in such a way that the $X$ gate has now support overlapping with $H_{\mu}$. The number of gates in this Clifford path is $L-k+1$.
Next, we fix other Cliffords in the circuit to be trivial Identity $I$ Clifford gates, specifically those connecting with the remaining Paulis $X$ in the tensor product decomposition of $P_X$. See Figure~\ref{Fig:circ_coolproof} for an example. Consequently, the non-trivial Pauli $\tilde{P}_X$ remains the same up to permutations of its tensor factor: in particular one of its $X$ Paulis is now swapped to a position where it acts non-trivially with $H_{\mu}$. The count of fixed Cliffords in the circuit is then given by
\begin{align}
    \text{\# fixed Cliffords} \leq |P|(L-k+1),
\end{align}
since, at each layer (of which there are $L-k+1$), we fix at most one gate for each of the single-qubit Paulis in the Pauli decomposition of $P_X$.
Now, we utilize the last layer of single-qubit random gates to map the resulting Pauli, which now shares support with $H_{\mu}$, to a Pauli that does not commute with $H_{\mu}$. Note that such a Pauli exists, as any operator commuting with all the Paulis should be the identity. However, $H_{\mu}$ cannot be the identity because that would contradict the assumption that the support of $H_{\mu}$ is in the light-cone. We denote such new resulting Pauli as $Q$.
\begin{figure}[h]
\centering
\includegraphics[width=0.59\textwidth]{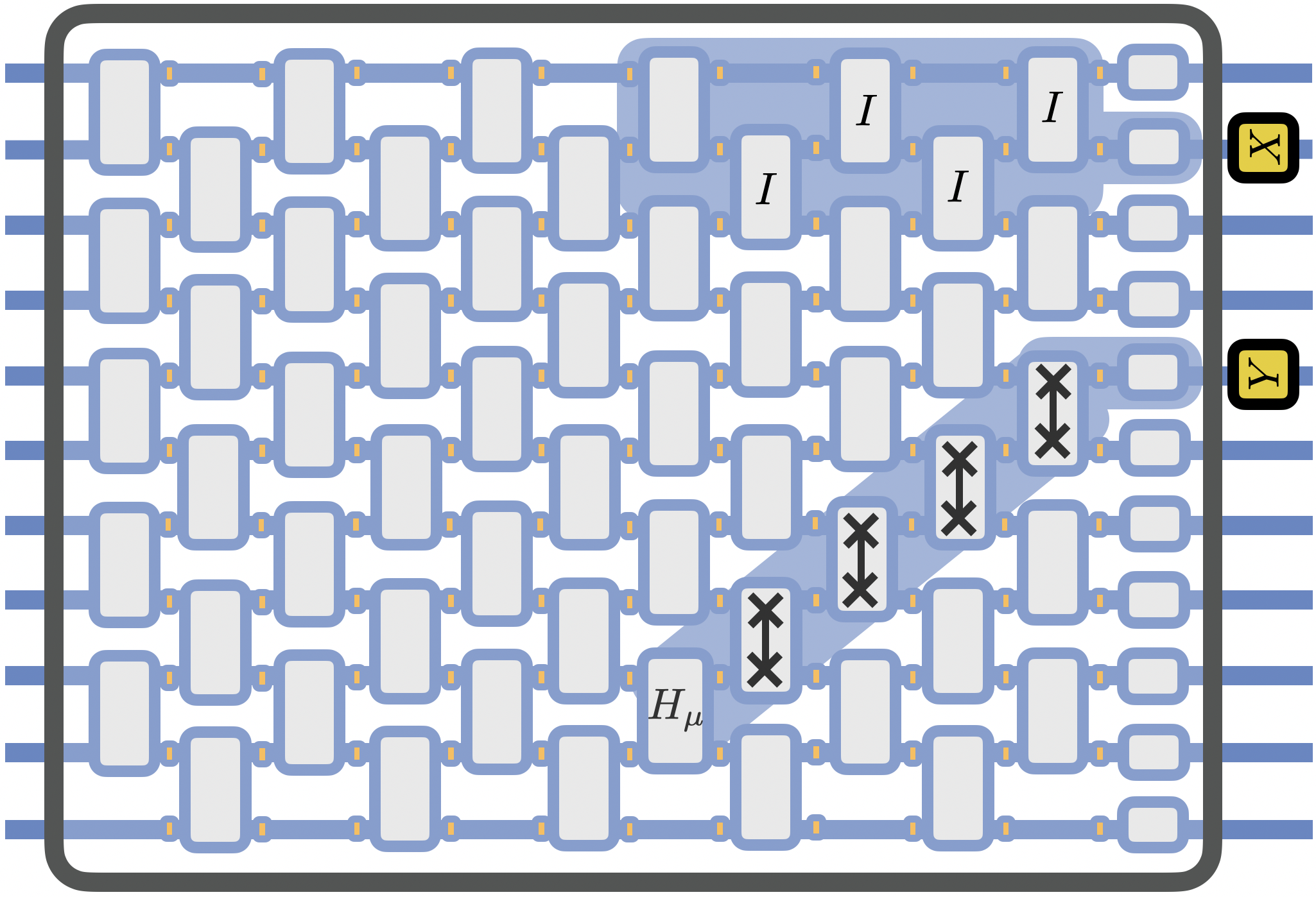}
\caption{Example of Clifford path choices. The shaded region indicates the fixed Clifford gates. Note that we choose some Clifford gates to be SWAP gates so that they connect one of the Paulis to $H_{\mu}$. We protect the other remaining Pauli from spreading across the circuit with identity Clifford gates.} 
\label{Fig:circ_coolproof}
\end{figure}
Thus, by repeating the same for $P_Y$ and $P_Z$, and mapping them at the end to $Q$, we have the lower bound
\begin{align}
    \Var[\partial_\mu C]\ge \frac{1}{|\mathrm{C}_2|^{|P|(L-k+1)}} \frac{1}{3^{|P|}}\!\left( D_X^{2|P|(L-k+1)} +  D_Y^{2|P|(L-k+1)} +  D_Z^{2|P|(L-k+1)}\right) \Ex (\Tr\!\left(\Phi_{[1,k]}(\rho_0) i \left[H_\mu ,Q\right]\right))^2.
    \nonumber
\end{align}
Since $H_\mu$ is a $2$-qubit gate, it can be expanded in its $2$-qubits Pauli decomposition, as
\begin{align}
    H_\mu= \sum_{R \in \{I,X,Y,Z\}^{\otimes 2}} b_R R,
\end{align}
Thus, substituting and using Lemma~\ref{le:mixham}, we have
\begin{align}
   \Ex (\Tr\!\left(\Phi_{[1,k]}(\rho_0) i \left[H_\mu ,Q\right]\right))^2&=\sum_{R \in \{I,X,Y,Z\}^{\otimes 2}} b^2_R \Ex (\Tr\!\left(\Phi_{[1,k]}(\rho_0) Q_R \right))^2,
\end{align}
where we defined $Q_R \coloneqq i \left[R ,Q\right]$.
Here we note that $Q_R$ is a Hermitian operator, and in particular, it is a Pauli operator (or the zero operator). Thus, we have
\begin{align}
   \Ex (\Tr\!\left(\Phi_{[1,k]}(\rho_0) i \left[H_\mu ,Q\right]\right))^2&=\sum_{R \in \{I,X,Y,Z\}^{\otimes 2}} b^2_R \Ex (\Tr\!\left(\Phi_{[1,k]}(\rho_0) Q_R \right))^2\\
   \nonumber
   &\ge   \sum_{R \in \{I,X,Y,Z\}^{\otimes 2}} b^2_R (1- \delta_{[Q,R],\textbf{0}}) \norm{\bold{t}}^{2|Q_R|}_2,\\
      \nonumber
   & \ge \norm{\bold{t}}^{2(|P|+1)}_2  \sum_{R \in \{I,X,Y,Z\}^{\otimes 2}} b^2_R (1- \delta_{[Q,R],\textbf{0}})\\
      \nonumber
   & \ge  \norm{\bold{t}}^{2(|P|+1)}_2   b^2_{\tilde{R}  }
   \nonumber
\end{align}
where in the second step we have used the lower bound on variance expectation values derived in Proposition~\ref{prop:lbvar} and introduced $ \delta_{[Q,R],\textbf{0}}$ which is one if $Q$ and $R$ commute and zero otherwise, in the third step we observed that $|Q_R|\le |P|+1$ and used that $\norm{\bold{t}}_2\le 1$ (Lemma~\ref{le:normal}). In the last step we have used that, since $H_{\mu}$ does not commute with $Q$, there must exists $\tilde{R}\in \{I,X,Y,Z\}^{\otimes 2}$ such that $b_{\tilde{R}}\neq 0$ and $[\tilde{R},Q]\neq 0$.
Now we observe that
\begin{align}
     b^2_{\tilde{R}  }  \ge  \frac{1}{16} \sum_{R \in \{I,X,Y,Z\}^{\otimes 2}} b^2_R= \frac{1}{64}\norm{H_\mu}^2_2 \ge \frac{1}{64} \norm{H_{\mu}}^2_{\infty}.
\end{align}
Putting everything together, we have
\begin{align}
    \Var[\partial_\mu C]&\ge \frac{1}{|\mathrm{C}_2|^{|P|(L-k+1)}} \frac{1}{3^{|P|}}\!\left( D_X^{2|P|(L-k+1)} +  D_Y^{2|P|(L-k+1)} +  D_Z^{2|P|(L-k+1)}\right)\frac{1}{64} \norm{H_{\mu}}^2_{\infty} \norm{\bold{t}}^{2(|P|+1)}_2.
    \nonumber
\end{align}
We note that if all the entries of $\bold{D}=(D_X, D_Y, D_Z)$ are equal to zero, then the noise is a 
replacer channel.
This because it holds that 
\begin{equation}
\mathcal{N}\!\left(\frac{I+ \bold{w}\cdot \boldsymbol{\sigma}}{2}\right)= \frac{I}{2}   + \frac{1}{2}(\bold{t}+D \bold{w})\cdot \boldsymbol{\sigma},
\end{equation}
where $D=\mathrm{diag}(\bold{D})$, $\mathcal{N}$ is a single-qubit noise channel, and $\frac{I+ \bold{w}\cdot \boldsymbol{\sigma}}{2}$ represents a density matrix for $\norm{w}_2\le 1$ (see Eq.~\eqref{eq:normSS}). 
However for assumption, the noise cannot be a replacer channel and so at least one of the entries of $\bold{D}$, say $D_X$, should be nonzero. Further lower bounding, we get
\begin{align}
    \Var[\partial_\mu C]\ge\sum_{Q \in \{X,Y,Z\}} \frac{1}{64}   \left(\frac{D^2_Q \norm{\bold{t}}^2 _2}{3|\mathrm{C}_2|}\right)^{|P|(L-k+1)} \norm{H_{\mu}}^2_{\infty}. 
    \nonumber
\end{align}
This quantity has the claimed scaling, so we can conclude the proof.
\end{proof}
It might be useful to give the same scaling without the asymptotic notation. 
\begin{remark}[Scaling without asymptotic notation]
The lower bound (without the asymptotic notation) we found in the previous Proposition~\ref{th:variance_PDlb} on the partial derivative with respect to the parameter $\theta_\mu$ of the gate $\exp(-i \theta_{\mu} H_{\mu})$ in the $k$-layer of a $L$-depth circuit is:
\begin{align}
    \Var[\partial_\mu C]\ge\sum_{Q \in \{X,Y,Z\}} \frac{1}{64}   \left(\frac{D^2_Q \norm{\bold{t}}^2 _2}{3|\mathrm{C}_2|}\right)^{|P|(L-k+1)} \norm{H_{\mu}}^2_{\infty},  
    \nonumber
\end{align}
where $\bold{D}$ and $\bold{t}$ are the noise parameters (see Lemma~\ref{le:normal}), $|\mathrm{C}_2|$ is the size of the $2$-qubits Clifford group.
\end{remark}
We note that even if the non-unital noise rate $\norm{\bold{t}}_2$ is polynomially small in the number of qubits, then the derived lower bound still indicates absence of barren plateaus (due to the last few layers).
Proposition~\ref{th:variance_PDlb} readily leads to the following conclusion regarding a lower bound on the expected value of the $2$-norm of the gradient.
\begin{corollary}[Lower bound on the expected value of the $2$-norm of the gradient]
\label{cor:lbGrad}
Let us consider a cost function associated to a local Pauli $P$ with $|P|=\Theta(1)$, and the same assumption as in Proposition~\ref{th:variance_PDlb}. Then, we have
    \begin{align}
        \Ex[\norm{\nabla C}_2^2]\ge \Omega\left(1\right).
    \end{align}
\end{corollary}
\begin{proof}
The proof follows immediately by focusing only on the last parameter $\theta_m$ of the last layer (which gate is in the light cone of $P$). In particular, we have
    \begin{align}
        \Ex[\norm{\nabla C}_2^2]\ge \sum^{m}_{\mu=1}\Ex[\partial_\mu C^2]\ge \Ex[\partial_m C^2] = \Var[\partial_m C^2] \ge \Omega\left(1\right),
    \end{align}
where we have used the fact that the partial derivative has zero mean and Proposition~\ref{th:variance_PDlb}.
\end{proof}
We can now rephrase our previous result in terms of a probability statement.
\begin{corollary}[Probability statement]
Assuming that the cost function has a number of free parameters upper bounded by \(O\!\left(\mathrm{poly}(n)\right)\), we have
\begin{align}
    \mathrm{Prob}\left( \norm{\nabla C}_2^2   > \Omega(1) \right)\ge \Omega({1}/{\mathrm{poly}(n)}).
\end{align}
\end{corollary}
\begin{proof}
By applying the probability inequality in Lemma~\ref{le:probstat2} with $f\coloneqq \norm{\nabla C}_2^2$ and utilizing Corollary~\ref{cor:lbGrad}, we arrive at
\begin{align}
    \mathrm{Prob}\left( \norm{\nabla C}_2^2  > \Omega(1) \right) \ge \frac{\Omega(1)}{\sup(|\norm{\nabla C}_2^2|)}.
\end{align}
In order to conclude, we need to establish an upper bound for $\sup(|\norm{\nabla C}_2^2|)$.
This upper bound can be derived as 
\begin{align}
    \sup(|\norm{\nabla C}_2^2|)&\le m \max_{\mu \in [m]} (\partial_\mu C)^2  
    \le  4 m \max_{\mu \in [m]}\norm{ H_\mu}^2_\infty \norm{ P}^2_\infty\le 4 m.
\end{align}
In this equation, $m$ represents the number of parameters, and we employ Lemma~\ref{le:derupp} in the final step.
\end{proof}
In summary, we have shown that the probability of sampling an instance of a circuit in which the gradient is larger than a constant is not exponentially small. 
However, it is important to stress that achieving a large average gradient norm can be accomplished by focusing on the last layers. In fact, the components corresponding to the initial layers of a linear depth circuit are exponentially small.

\subsection{Improved upper bounds for unital noise}
\label{sub:unital}
In this section, we present improved upper bounds on the barren plateaus phenomenon in the context of random quantum circuit ansatz with unital noise. Our derived bounds are tighter compared to those presented in Ref.~\cite{nibp}. Notably, our approach leverages the randomness of the circuit, whereas~\cite{nibp} relied solely on the contraction property of the unital noise channel analyzed.
We start by showing the variance of expectation values in the case of random quantum circuits with unital noise. Up to our knowledge, this was not known before. The noiseless case was instead addressed in Ref.\  \cite{napp2022quantifying}. In the following, due to Lemma~\ref{le:mixham}, we can focus on Pauli observables without loss of generality.
\begin{proposition}[Improved expectation values concentration for unital noise]
\label{prop:expdecayPunital}
Let $P\in \{I,X,Y,Z\}^{\otimes n}$, $\rho_0$ be any quantum state, and $L$ be the depth of the noisy circuit $\Phi$ defined in Eq.~\eqref{eq:randcirc} in arbitrary dimension. Specifically, we assume that the noise is unital. Then, we have
\begin{align}
    \Var[\Tr(P\Phi(\rho_0))] \le 4 c^{|P|+L-1},
\end{align}
where the parameter $c$ is defined in Eq.~\eqref{eq-noise_value}.
\end{proposition}

\begin{proof}
Because our circuit ends with a layer of random single-qubit gates, it holds that $\Ex[\Tr(P \Phi(\rho))]=0$, following from Lemma~\ref{le:qubitonedesign}.
Thus, we focus on $\Ex[\Tr(P\Phi(\rho_0))^2]$. We have
\begin{align}
    \Ex[\Tr(P\Phi(\rho_0))^2]=\Ex\left[\Tr\left(P\Phi\!\left(\rho_0-\frac{I_n}{2^n}\right)\right)^2\right],
\end{align}
where we have used the unitality of the noise channels to get $\Phi(I_n)=I_n$, and the fact that the Pauli operators are traceless. The claim follows by applying the effective depth Theorem~\ref{prop:expdecayP}.
\end{proof}

We now show an upper bound on the partial derivative variance in the case of unital noise. Due to Lemma~\ref{le:mixhamPARTIAL}, we can focus on Pauli observables without loss of generality.
\begin{proposition}[Improved upper bound on the partial derivative for unital noise]
\label{th:ubvarianceUNIT}
Let $C\coloneqq \Tr(P\Phi(\rho_0))$ be the cost function, where $P\in \{I,X,Y,Z\}^{\otimes n}$, $\rho_0$ is an arbitrary initial state, and $\Phi$ is a quantum circuit of depth $L$ in arbitrary dimension. We assume that the noise is unital. Let $\mu$ denote a parameter of any $2$-qubit gate $\exp(-i\theta_\mu H_\mu)$ in the circuit such that $\norm{H_\mu}_{\infty}\le 1$. Then, we have
\begin{align}
    \Var[\partial_\mu C]\le 4 c^{|P|+L-1}.
\end{align}
\end{proposition}

\begin{proof}
By repeating the same steps of the proof of Proposition~\ref{prop:upvar}, namely the upper bound on the variance for non-unital noise, we get
\begin{align}
    \Var[\partial_\mu C]=\Ex[(\partial_\mu C)^2]&\le c^{|P|+L-k-1}\max_{Q\in \{I,X,Y,Z\}^{\otimes n} \setminus I_n} \Ex (\Tr\!\left(\Phi_{[1,k]}(\rho_0) i \left[H_\mu ,Q\right]\right))^2 .
\end{align}
Since $H_\mu$ is a $2$-qubit gate, it can be expanded in its $2$-qubits Pauli decomposition, 
as
\begin{align}
    H_\mu= \sum_{R \in \{I,X,Y,Z\}^{\otimes 2}} b_R R.
\end{align}
Thus, substituting and using Lemma~\ref{le:mixham}, we have
\begin{align}
   \Ex (\Tr\!\left(\Phi_{[1,k]}(\rho_0) i \left[H_\mu ,Q\right]\right))^2&=\sum_{R \in \{I,X,Y,Z\}^{\otimes 2}} b^2_R \Ex (\Tr\!\left(\Phi_{[1,k]}(\rho_0) Q_R \right))^2.
\end{align}
Here, we note that $Q_R \coloneqq i \left[R ,Q\right]$ is a Hermitian operator, and in particular, it is a Pauli operator (or the zero operator). Thus, we have
\begin{align}
   \Ex (\Tr\!\left(\Phi_{[1,k]}(\rho_0) i \left[H_\mu ,Q\right]\right))^2&=\sum_{R \in \{I,X,Y,Z\}^{\otimes 2}} b^2_R \Ex (\Tr\!\left(\Phi_{[1,k]}(\rho_0) Q_R \right))^2\\
   \nonumber
   &\le c^{k}  \sum_{R \in \{I,X,Y,Z\}^{\otimes 2}} b^2_R,
   \nonumber
\end{align}
where in the last step we have used Proposition~\ref{prop:expdecayPunital} and that the Pauli weight of the non-zero Pauli is lower bounded by one.
Now we observe that
\begin{align}
     \sum_{R \in \{I,X,Y,Z\}^{\otimes 2}} b^2_R= \frac{1}{4}\norm{H_\mu}^2_2 \le \norm{H_{\mu}}^2_{\infty}\le 1,
\end{align}
which concludes the proof.
\end{proof}
This result improves upon the partial derivative variance upper bound presented in Ref.\  \cite{nibp}, where the upper bound scaled as
\begin{align}
     \Var[\partial_\mu C]=O(n^{1/2}2^{-\alpha L}),
\end{align}
for some positive constant $\alpha$. It is noteworthy that this latter upper bound has no dependence on the Pauli weight, unlike ours. Furthermore, it includes a $n^{1/2}$ factor in front of the exponential, making it meaningful only at depths $\Omega(\log(n))$.
Moreover, our result is more general than that shown in Ref.~\cite{nibp} also because it extends to any unital noise, whereas the results shown in Ref.~\cite{nibp} apply only to primitive unital noise, which is only a particular type of unital noise (e.g., dephasing is not included in this class).

\newpage
\section{Purity and kernel methods under non-unital noise}
When a circuit is interspersed with primitive, unital noise, the decay in purity can be investigated by employing well-known entropy accumulation techniques (see, for instance, 
Refs.\ \cite{razborov2003upper, kempe2008upper, limitations2021, optimal, hirche2022contraction}). However, this approach has limited applicability under non-unital noise, as the noise channel can potentially decrease the entropy of the system. Here, we address this gap in the literature by providing upper and lower bounds on the purity of a noisy circuit, leveraging prior techniques along with the tools developed in the present work.
As an application, we employ our upper bounds to investigate the limitations of quantum kernel methods under non-unital noise.

\subsection{Purity of average and worst-case circuits}
In this section, we explore the decay in purity under non-unital noise. We propose two distinct approaches: first, we provide upper and lower bounds for average-case circuits under possibly non-unital noise; second, we provide upper bounds for worst-case circuits, under the further assumption that the noise channel can be decomposed into a depolarizing channel followed by an arbitrary channel.

\subsubsection{Purity of an average-case noisy circuit}
We now upper and lower bound the expected purity of the output state of a noisy circuit, as defined in Eq.~\eqref{eq:randcirc}.

\begin{proposition}[Average-case upper and lower bounds on the purity]
Let $\rho$ be a quantum state. Then, at any depth of the noisy circuit $\Phi$, we have
\label{prop:avg-purity}
     \begin{equation}
         \left(\frac{1+{\|\bold{t}\|^2_2}}{2}\right)^n \leq \mathbb{E}\Tr[\Phi(\rho)^2] \leq \left(\frac{1+{\|\bold{t}\|^2_2 + \|\bold{D}\|^2_2}}{2}\right)^n.
     \end{equation}
\end{proposition}

\begin{proof}
We first recall that the purity can be expressed in the Pauli basis as 
    \begin{align}
    \Tr[\Phi(\rho)^2] = \Tr[\mathbb{F}\Phi(\rho)^{\otimes 2}] = \frac{1}{2^n}\sum_{P\in\{I,X,Y,Z\}^{\otimes n}}  \Tr[P^{\otimes 2}\Phi(\rho)^{\otimes 2}]
    =\frac{1}{2^n} \sum_{P\in\{I,X,Y,Z\}^{\otimes n}} \Tr[P\Phi(\rho)]^2.
    \end{align}
    Hence, plugging the upper and lower bound on the expected second moments (Eqs.\ \ref{eq:lbpaulipurity},\ \ref{eq:ub-pauli-purity}) yields the desired results
    \begin{align}
        \mathbb{E}\Tr[\Phi(\rho)^2] = \frac{1}{2^n}\sum_{P\in\{I,X,Y,Z\}^{\otimes n}}  \mathbb{E} \Tr[P\Phi(\rho)]^2 \leq \sum_{P\in\{I,X,Y,Z\}^{\otimes n}}  \left(\frac{\|\bold{t}\|^2_2 + \|\bold{D}\|^2_2}{3}\right)^{|P|}
        \\
        \nonumber
        = \frac{1}{2^n}\sum_{k=0}^n \binom{n}{k} \left({\|\bold{t}\|^2_2 + \|\bold{D}\|^2_2}\right)^{k} = \left(\frac{1+{\|\bold{t}\|^2_2 + \|\bold{D}\|^2_2}}{2}\right)^n,
    \end{align}
    and 
    \begin{align}
        \mathbb{E}\Tr[\Phi(\rho)^2] = \frac{1}{2^n}\sum_{P\in\{I,X,Y,Z\}^{\otimes n}}  \mathbb{E} \Tr[P\Phi(\rho)]^2 \geq \sum_{P\in\{I,X,Y,Z\}^{\otimes n}}  \left(\frac{\|\bold{t}\|^2_2 }{3}\right)^{|P|}
        \\ 
        \nonumber= \frac{1}{2^n}\sum_{k=0}^n \binom{n}{k}   {\|\bold{t}\|^{2k}_2 } = \left(\frac{1+{\|\bold{t}\|^2_2}}{2}\right)^n,
    \end{align}
which ends the proof.
\end{proof}

The upper bound in the above Proposition is exponentially small in \(n\) whenever $\|\bold{t}\|_2^2+\|\bold{D}\|_2^2<1.$
Under this additional assumption,
\begin{equation}
    \mathbb{E}\Tr[\Phi(\rho)^2] \leq 2^{-\Omega(n)}.
\end{equation}
%Similarly, let $\Phi(\rho)_S = \Tr_{n\setminus S} \Phi(\rho)$ the reduced state on a subset $S$ of the qubits of size $|S|=k$. Then Proposition~\ref{prop:avg-purity} implies that
%\begin{equation}
%    \left(\frac{1+{\|\bold{t}\|^2_2}}{2}\right)^k \leq \mathbb{E}\Tr[\Phi(\rho)_S^2] \leq \left(\frac{1+{\|\bold{t}\|^2_2 + \|\bold{D}\|^2_2}}{2}\right)^k. \label{eq:purity-local}
%\end{equation}
%\subsubsection{Average purity of reduced states}
We observe that the bounds given in Proposition \ref{prop:avg-purity} hold also for reduced states, that is states obtained by performing a partial trace on the output state of the circuit. In particular, for any arbitrary state $\rho$, we let $\rho_S \coloneqq  \Tr_{[n]\setminus S} [\rho]$ the reduced state on a subset $S$ of the qubits of size $|S|=k$. Then we have
\begin{equation}
    \left(\frac{1+{\|\bold{t}\|^2_2}}{2}\right)^k \leq \mathbb{E}\Tr[\Phi(\rho)_S^2] \leq \left(\frac{1+{\|\bold{t}\|^2_2 + \|\bold{D}\|^2_2}}{2}\right)^k.  \label{eq:purity-1qubit}
\end{equation}
%Moreover, for $k = 1$, we can derive a more fine grained expression of the purity by employing the results given in (\cite{ErrorMitigationObstructions}, Section VII.A).
%We simplify the notation by writing $\rho_i \coloneqq  \rho_{\{i\}} \coloneqq  \Tr_{[n]\setminus \{i\}} [\rho]$.
%Then we obtain:

%\begin{align}
    %&\mathbb{E} \Tr[\Phi(\rho)_i^2]  = \mathbb{E} \Tr[\Phi(\rho)_i'^2] \nonumber
    %\\
    %\nonumber = &\mathbb{E}  \left(\frac{4}{3} - \frac{2}{3} \Tr[\Phi(\rho)_i''^2]\right) \Tr[\mathcal{N}(I/2)^2] + \left(\frac{3}{4}\Tr[\Phi(\rho)_i''^2]- \frac{2}{3} \right)\Tr[J(\mathcal{N})^2]
    %\\ 
    %= &\frac{1}{3}\mathbb{E}\Tr[\Phi(\rho)_i''^2]\|\bold{D}\|_2^2 + \frac{1}{2}\left(1+\|\bold{t}\|_2^2 - \frac{\|\bold{D}\|_2^2}{3}\right),
%\end{align}
%where $\mathcal{J}(\mathcal{N}):= \mathcal{N} \otimes \mathcal{I}\left(\frac{\ketbra{\Omega}{\Omega}}{4^n}\right)$ is the Choi state of the channel $\mathcal{N}$.
%Thus, by inspecting the above identity we find that the upper bound in Proposition \ref{prop:avg-purity} is matched whenever $\mathbb{E}
%\Phi(\rho)_i''$ is pure, and the lower bound is matched whenever $\mathbb{E}\Phi(\rho)_i''$ is maximally mixed.

\subsubsection{Purity of a worst-case noisy circuit}
\label{subsub:purityy}
In this section, we consider a layered circuit $\mathcal{C}$ of the form 
\begin{align}
     \C=\widetilde{\mathcal{N}}^{\otimes n} \circ \mathcal{U}_L\circ \cdots   \circ\widetilde{\mathcal{N}}^{\otimes n}  \circ \mathcal{U}_1,
\end{align}
where we do not make any assumption on the structure of each unitary layer $\mathcal{U}_i =U_i^\dag(\cdot) U_i$.
In contrast, we will make a further assumption on the noise channel. In particular, we will model the local noise as the composition of two single-qubit channels, namely a local depolarizing channel ${\mathcal{N}_p^{(\mathrm{dep})}}(X) = p \frac{I}{2} \Tr[X] + (1-p) X$ and arbitrary noise channel $\mathcal{N}$ expressed in the normal form, i.e.,
\begin{align}
    \widetilde{\mathcal{N}} = {\mathcal{N}} \circ {\mathcal{N}_p^{(\mathrm{dep})}}, 
\end{align}
where $\mathcal{N}(I+\bold{w}\cdot \boldsymbol{\sigma}) = I+(\bold{t}+D\bold{w})\cdot \boldsymbol{\sigma}$.
Under this stronger assumption, we provide two upper bounds on the purity of a worst-case noisy circuit.
We also remark that order of ${\mathcal{N}}$ and $ {\mathcal{N}_p^{(\mathrm{dep})}}(X)$ does not play a central role in our analysis, therefore the same results could be derived inverting their order.

Let us recall the definition of quantum relative entropy and quantum sandwiched Rényi divergence~\cite{MDS+13,WWY14}.
Let $\rho, \sigma$ be two quantum states. If $\operatorname{supp}(\rho) \subseteq \operatorname{supp}(\sigma)$, we define the quantum relative entropy as
\begin{align}
    D(\rho \| \sigma):=\operatorname{Tr}(\rho \log \rho)-\operatorname{Tr}(\rho \log \sigma) .
\end{align}
For a parameter $\alpha \in(0,1) \cup(1, \infty)$, the quantum Rényi divergence of order $\alpha$ is defined as
\begin{align}
    D_\alpha(\rho \| \sigma):=\frac{1}{\alpha-1} \log \operatorname{Tr}\left[\left(\sigma^{\frac{1-\alpha}{2\alpha}}\rho\sigma^{\frac{1-\alpha}{2\alpha}}\right)^{\alpha}\right] .
     \label{eq:relativeentropy}
\end{align}
This definition applies when $\operatorname{supp}(\rho) \subseteq \operatorname{supp}(\sigma)$, for $\alpha \in (1, \infty)$. In the limit $\alpha \rightarrow 1$, the quantum Rényi divergence reduces to the quantum relative entropy, i.e., $\lim _{\alpha \rightarrow 1} D_\alpha(\rho \| \sigma)=$ $D(\rho \| \sigma)$.
The $\infty$-relative entropy is defined, for $\operatorname{supp}(\rho) \subseteq \operatorname{supp}(\sigma)$, as
\begin{equation}
D_{\infty }(\rho \| \sigma):=\inf \left\{\gamma: \rho \leq 2^\gamma \sigma\right\} .
\end{equation}
It is useful to recall that for $\alpha>\beta>0$, the Rényi divergences satisfy the monotonicity property, i.e., $D_\alpha(\rho \| \sigma) \geq D_\beta(\rho \| \sigma)$.

By mean of the \emph{data-processed triangle inequality} (\cite{Christandl_2017}, Theorem 3.1), the authors of Refs.\ \cite{limitations2021, optimal}, obtained an upper bound on the purity of the output of a non-unital channel, which is exponentially small in $n$ when the unital component of the noise `dominates' the non-unital one. We rephrase such result within our model, giving an explicit expression in terms of $p$ and $\bold{t}$.

\begin{corollary}[Worst-case upper bound on the purity]
\label{cor:purity-worst2}
Let $\rho$ an arbitrary quantum state and assume $p>0$ and $\|\bold{t}\|_2 \neq 1$. Then for any constant noise parameters, we have
\begin{equation}
    D_2\left(\mathcal{C}(\rho)\bigg\|\frac{I}{2^n}\right) \leq n\left((1-p)^{2L} + \|\bold{t}\|_2\frac{1 - (1-p)^{2L}}{2p - p^2}\right)\coloneqq  n\cdot \delta_L. 
\end{equation}
This implies the following upper bound on the purity
\begin{equation}
    \Tr[\mathcal{C}(\rho)^2] \leq 2^{ n(\delta_L -1)}.
\end{equation}
\end{corollary}
\begin{proof}
We first recall that $\Tr[\rho^2] = 2^{-n + D_2\left(\rho\|{I}/{2^n}\right)}$, then first bound implies the second.
We note the following
\begin{equation}
    D_{\infty}\left(\mathcal{N}^{\otimes n}\left(\frac{I}{2^n}\right)\bigg\| \frac{I}{2^n}\right) = n D_{\infty}\left(\mathcal{N}\left(\frac{I}{2}\right)\bigg\| \frac{I}{2}\right) = n\log(1+\|\bold{t}\|_2) \leq  n \|\bold{t}\|_2 ,
\end{equation}
where the second identity is a special case of Lemma 23 in Ref.~\cite{rubboli2023mixed}.
Moreover,  Lemma C.1 in Ref.~\cite{optimal} ensures 
\begin{equation}
    D_2\left(\mathcal{C}(\rho)\bigg\| \frac{I}{2^n}\right) \leq (1-p)^{2L}  D_2\left(\rho\bigg\| \frac{I}{2^n}\right) + \sum_{t=0}^L (1-p)^{2t} D_{\infty}\left(\mathcal{N}^{\otimes n}\left(\frac{I}{2^n}\right)\bigg\| \frac{I}{2^n}\right) .
\end{equation}
Then the desired upper bound on $D_2\left(\mathcal{C}(\rho)\|{I}/{2^n}\right)$ immediately follows. 
\end{proof}
Note that the term $\delta_L$ converges exponentially fast to $\|\bold{t}\|_2/({2p - p^2})$, and thus in this regime the bound is non-trivial if $\|\bold{t}\|_2\leq 2p - p^2$. 
%Moreover, for this regime of the noise, if $L \geq c\cdot \log n$ for a sufficiently large constant $c$, we obtain that $\Tr[\rho^2] = 2^{-\Omega(n)}$. 

\subsection{Quantum machine learning under non-unital noise:  Kernel methods}
\label{sec:kernels}
Quantum kernel methods offer a hopeful avenue for advancing quantum machine learning. However, despite certain positive results, as documented in Ref.\  \cite{liu2021rigorous}, these methods remain susceptible to trainability challenges. In particular, the work of \citet{thanasilp2022exponential} has demonstrated that various factors, such as circuit randomness  and unital noise, can potentially compromise their trainability, in analogy to the phenomenon of barren plateaus for cost functions. 
Here we incorporate both unital and non-unital noise in our analysis and we show that fidelity kernels exponentially concentrate even at constant depth. This starkly improves the result of Ref.\ (\cite{thanasilp2022exponential}, Theorem 3) , which predicts exponential concentration at linear depth for unital noise.
%Moreover, we show that projected kernels do not incur in exponential concentration at any depth on average-case circuits interspersed by non-unital noise, in analogy with the lack of barren plateaus for local cost functions presented in Section \ref{sub:local}.

\subsubsection{Preliminaries on quantum kernel methods}
Consider an $n$-qubit data-embedding channel $\Phi_{\boldsymbol{x}}$ parametrized by a point $\boldsymbol{x}\in\mathcal{X}$, so that
\begin{equation}
\label{eq:embedding}
    \rho(\boldsymbol{x}) = \Phi_{\boldsymbol{x}}(\rho_0),
\end{equation}
where $\rho_0$ is the initial state of the circuit, usually set as $\rho_0 = \ketbra{0^n}{0^n}$.
A kernel $\kappa : \mathcal{X}\times\mathcal{X}\rightarrow \mathbb{R}^+$ is a similarity measure between pair of points ${\boldsymbol{x}},{\boldsymbol{y}}\in\mathcal{X}$. 
In particular, quantum kernels rely on the quantum embedding scheme described in the Equation~\ref{eq:embedding} above.
We consider the fidelity quantum kernel \cite{havlivcek2019supervised, schuld2021supervised}, defined as
\begin{equation}
    \kappa^{FQ}({\boldsymbol{x}},{\boldsymbol{y}}) = \Tr[\rho({\boldsymbol{x}})\rho({\boldsymbol{y}})].
\end{equation}
\begin{comment}
The projected quantum kernel \cite{Huang_2021} is defined as
\begin{equation}
    \kappa^{PQ}({\boldsymbol{x}},{\boldsymbol{y}}) := \exp \left(-\gamma \sum_{k=1}^n \|\rho_k({\boldsymbol{x}})-\rho_k({\boldsymbol{y}})\|_2^2\right),
\end{equation}
where $\rho_k({\boldsymbol{x}}) = \Tr_{\overline{k}}\rho({\boldsymbol{x}})$ is the reduced density matrix of the $k$-{th} qubit and $\gamma$ is a positive hyperparameter.
\end{comment}
%The definition of the projected kernel is a special case of the ones given in Ref.\  \cite{thanasilp2022exponential} and can be obtained by setting the hyperparamter $\gamma = 1/n$ in the original definition. We emphasize that all the results given in this Chapter can be generalized to the general case $\gamma \in \left[ \frac{1}{\poly(n)},\poly(n)\right]$. 
Kernel-based learning methods are notable for their capacity to transform data from the original space $\mathcal{X}$ into a higher-dimensional feature space, which in our case coincides with the a $2^n$-dimensional Hilbert space. In this new feature space, inner products are computed, enabling the training of decision boundaries like support vector machines, as explained in reference~\cite{schuld2021supervised}.

\begin{figure}[h]
\label{Fig_final/:noisy-circuit}
\centering

\begin{comment}
\scalebox{0.9}{
\begin{tikzpicture}[scale=1.0,thick]
%draw circuit
\foreach[count=\i] \y in {1,2,3,4}
{\tikzmath{\j=int(\i-4);\jx=int(10-\i);\xoff = 0.2;}
\draw (0.5,\y) -- (12.2,\y);
\draw[fill=white,rounded corners] (0.7,\y-0.35) rectangle (1.7,\y+0.35);
\node at (1.2,\y) {\footnotesize $V^{\scalebox{0.6}{(\j)}}(x)$};
\draw[fill=white,rounded corners] (10.9,\y-0.35) rectangle (11.9,\y+0.35);
\node at (11.4,\y) {\footnotesize $V^{\scalebox{0.6}{(\jx)}}(x)$};
\node at (0.18,\y) {$\ket{0}$};
\draw[fill=white,rounded corners] (12+\xoff,\y-0.3) rectangle (12.6+\xoff,\y+0.3);
\draw[-latex,semithick] (12.3+\xoff,\y-0.22) -- (12.46+\xoff,\y+0.22);
\draw[semithick] (12.1+\xoff,\y-0.16) to[out=60,in=120] (12.5+\xoff,\y-0.16);
}
%draw configuration
\foreach \x/\y in {
3.3/4,3.3/3,
5.1/2,5.1/3,
6.9/4,6.9/3,
8.7/1,8.7/2,
10.5/2,10.5/3}
{\draw[fill=white] (\x-0.05, \y) circle (0.25);
\node at (\x-0.05,\y) {\scalebox{0.65}{$\mathcal{N}$}};}
\foreach[count=\i] \x/\y in {2/3, 3.8/2, 5.6/3, 7.4/1, 9.2/2}
{\draw[fill=white,rounded corners] (\x-0.1,\y-0.3) rectangle (\x+0.9,\y+1.3);
\node at (\x+0.42,\y+0.5) {\scalebox{0.86}{$U^{(\i)}(x)$}};}
\end{tikzpicture}
}
\end{comment}
\includegraphics[width=0.7\textwidth]{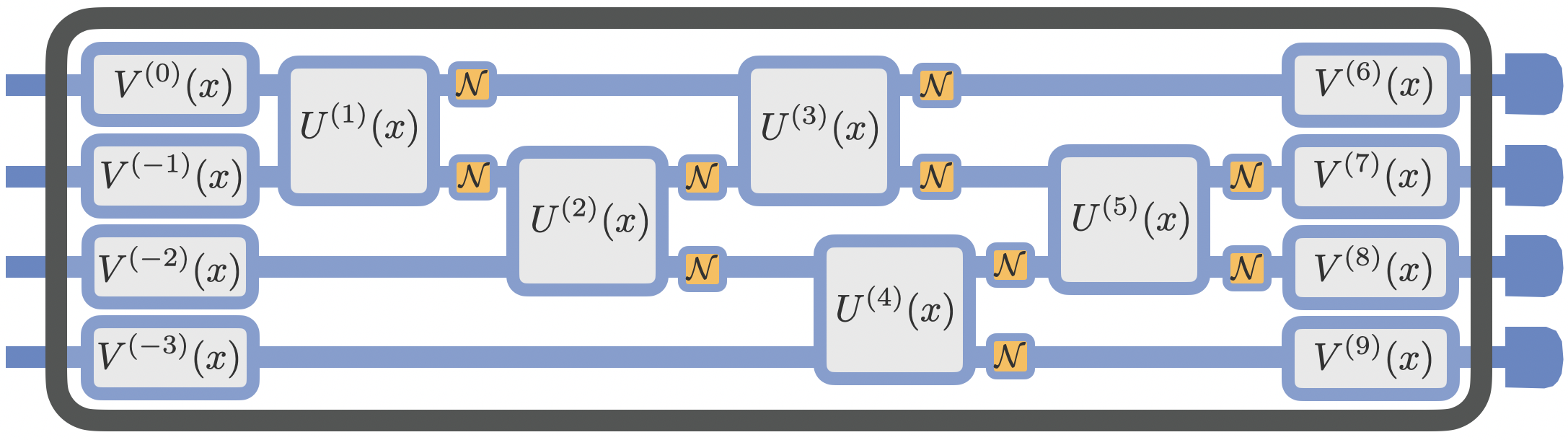}
\caption[Example of a noisy quantum circuit diagram on $n=4$ qudits and $s=5$ two-qudit gates]{Example of a noisy quantum circuit on $n=4$ qubits with two-qubit and single-qubit gates, parametrized by the input vector ${\boldsymbol{x}} \in \mathcal{X}$. }
\label{fig:noisycircuitdiagram}
\end{figure}

\subsubsection{Kernel-based supervised learning}

To better suit our results, we sketch how kernel methods can be used to perform supervised learning. We consider a training set of labelled inputs $\mathcal{S} = \left\{{\boldsymbol{x}}^{(i)},f\left({\boldsymbol{x}}^{(i)}\right)\right\}_{i\in[m]}$, where $f(\cdot)$ is some unknown function that we want to learn. Thus our goal is to find a function $h$ approximating $f$. Thanks to the Representer Theorem (see, for instance, \cite{shalev2014understanding}, Theorem 16.1), the optimal function can be expressed as  
\begin{equation}
    h (z) = \sum_{i=1}^m a_i \kappa \left({\boldsymbol{x}}^{(i)},z\right),
\end{equation}
where the $\boldsymbol{a} = (a_1,a_2,\dots,a_m)$ is a vector of parameters to be optimized with respect to a suitable loss function.

Then, to enable the implementation of kernel methods, it is necessary to estimate the Gram matrix. This matrix, denoted as $\mathcal{G}$, comprises the kernels derived from pairs of inputs within the training set ${\boldsymbol{x}}^{(1)}, {\boldsymbol{x}}^{(2)},\dots, {\boldsymbol{x}}^{(m)}$, and is defined as

\begin{equation}
\forall i\in[m] : \mathcal{G}[i,j] = \kappa \left({\boldsymbol{x}}^{(i)},{\boldsymbol{x}}^{(j)}\right).
\end{equation}
We recall that kernels exhibit exponential concentration with respect to a distribution $\mathcal{D}$ over $\mathcal{X}$, if there exists a real number $\mu\in\mathbb{R}$ and a value $\delta\in2^{-\Omega(n)}$ such that
\begin{equation}
    \Pr_{{\boldsymbol{x}},{\boldsymbol{y}}\sim\mathcal{D}} [|\kappa({\boldsymbol{x}},{\boldsymbol{y}}) - \mu| \geq \delta ] \in2^{-\Omega(n)}.
\end{equation}
In this case, all the entries of the Gram matrix are exponentially close to $\mu$ with exponentially high probability, making the optimization of the vector $\boldsymbol{a}$ an information-theoretically hard task.

\subsubsection{Assumption on the training data distribution.}
We assume that each point in the training set is sampled from a distribution $\mathcal{D}: \mathcal{X} \rightarrow [0,1]$ and we denote by $\nu'$ the corresponding induced distribution over quantum channels. Moreover, we also assume that $\nu'$ satisfies Definition\ \ref{def:localdesignlayer}, i.e., it is 2-local 2-design layer distribution.
\begin{comment}
\begin{enumerate}
    \item each layer is invariant under left-multiplication by a layer of single-qubit Clifford gates;
    \item moreover, the circuit is ended by a layer of single-qubit Clifford gates sampled uniformly at random.
\end{enumerate}
The second and third assumptions will play a pivotal role in the proof of absence of exponential concentration for the projected quantum kernels.
We also remark that these assumptions could be further relaxed, since our computation only involve (up to) fourth moments.
\end{comment}

\subsubsection{Fidelity quantum kernels: Exponential concentration at any depth}

The Cauchy-Schwarz inequality implies that the fidelity quantum kernel can be upper bounded by the square root of the purities of the output states
 \begin{equation}
     \Tr[\rho(\boldsymbol{x})\rho(\boldsymbol{x'})] \leq \sqrt{\Tr[\rho(\boldsymbol{x})^2]  \Tr[\rho(\boldsymbol{x'})^2] }. 
 \end{equation}
By a direct application Proposition~\ref{prop:avg-purity}, we obtain the following result.

\begin{corollary}[Exponential concentration of quantum kernels, average-case circuit]
\label{cor:kernel-avg}

\begin{equation}
    \mathbb{E}_{\boldsymbol{x},\boldsymbol{x'}}\kappa^{FQ}(\boldsymbol{x},\boldsymbol{x'}) \leq \left(\frac{1+{\|\bold{t}\|^2_2 + \|\bold{D}\|^2_2}}{2}\right)^n.
\end{equation}
\end{corollary}
In a similar fashion, we can derive a worst-case concentration bound by employing Corollary~\ref{cor:purity-worst2}.
\begin{corollary}[Exponential concentration of quantum kernels, worst-case circuit]
\label{cor:kernel-worst}
Let $\mathcal{U}_x,\mathcal{U}_{x'}$ two noisy circuits interspersed by $L$ layers of local noise, modeled by the channel $\mathcal{N}^{\otimes n} \circ \mathcal{N}_p^{(\mathrm{dep})\otimes n} $. Denote by $\rho(\boldsymbol{x}) = \mathcal{U}_{\boldsymbol{x}}(\rho_0)$ and $\rho({\boldsymbol{x'}}) = \mathcal{U}_{{\boldsymbol{x'}}}(\rho_0)$ the output states of the noisy circuits. Assume that $p = \Theta(1)$ and $\|\bold{t}\|_2 \neq 1$. Then the fidelity quantum kernel $\kappa^{FQ}({\boldsymbol{x}},{\boldsymbol{x'}}) = \Tr[\rho(\boldsymbol{x})\rho(\boldsymbol{x'}) ] $ satisfies the upper bound
\begin{equation}
    \kappa^{FQ}(\boldsymbol{x},\boldsymbol{x'}) \leq 2^{ n(\delta_L -1)},
\end{equation}
where $\delta_L\coloneqq (1-p)^{2L} + \|\bold{t}\|_2\frac{1 - (1-p)^{2L}}{2p - p^2}$.
\end{corollary}
As mentioned in Subsection~\ref{subsub:purityy}, this bound can be exponentially vanishing in the number of qubits for certain range of parameters.  

We also emphasize when $\norm{\bold{t}} = 0$, our bound predicts that the kernel $\kappa^{FQ}(\boldsymbol{x},\boldsymbol{x'}) $ is at most $ 2^{-n(2p-p^2)} = 2^{-\Omega(n)}$, even after a single layer of noise, whereas Theorem 3 in \cite{thanasilp2022exponential} only predicts that $|\kappa^{FQ}(x,x') -1/2^n| \leq (1-p)^2 = \Theta(1)$.
Thus, compared to the previous literature, our result is exponentially tighter with respect to the number of layers for the local depolarizing noise. An analogous bound for local Pauli noise can be derived along the lines of Supplementary Lemma 6 in Ref.\ \cite{nibp}.

\section{Miscellaneous}
\label{sec:misc}
\subsubsection{Useful lemmas: concentration inequalities}
\begin{lemma}[Large variance implies significant probability of deviation]
\label{le:probstat}
Let $f$ be a real function depending by parameters distributed according to a probability distribution $\mu$. Then, the inequality
\begin{align}
    \mathrm{Prob}\left(\left| f  - \mathbb{E}[f] \right| > \sqrt{\frac{\mathrm{Var}[f]}{2}} \right) \ge \frac{\mathrm{Var}[f]}{8\sup(|f|)^2}
\end{align}
 holds,
where the expected value and variance are taken with respect to the probability distribution $\mu$.
\end{lemma}

\begin{proof}
 Let $T > 0$ be a real value that we will fix later. We have
\begin{align}
    \Var[f]&\coloneqq \int\!\left( f  - \Ex\left[f\right] \right)^2 d\mu\\
    \nonumber
    &=\int_{\left| f  - \Ex\left[f\right] \right|\le T}\!\left( f  - \Ex\left[f\right] \right)^2 d\mu + \int_{\left| f  - \Ex\left[f\right] \right|> T}\!\left( f  - \Ex\left[f\right] \right)^2 d\mu\\
    \nonumber
    &\le T^2 \int_{\left| f  - \Ex\left[f\right] \right|\le T}\!1\, d\mu + \int_{\left| f  - \Ex\left[f\right] \right|> T}\!\left( f  - \Ex\left[f\right] \right)^2 d\mu\\
    \nonumber
    &\le T^2\left(1- \int_{\left| f  - \Ex\left[f\right] \right|> T}\!1\, d\mu\right)  + 4(\sup(|f|))^2\int_{\left| f  - \Ex\left[f\right] \right|> T}\!1 d\mu\\
    \nonumber
    &= T^2  + \left(4\left(\sup(|f|)\right)^2 - T^2 \right)\mathrm{Prob}\left(\left| f  - \Ex\left[f\right] \right| > T  \right),
    \nonumber
\end{align}
where in the fourth step we have used that $\left| f  - \Ex\left[f\right] \right|\le \left| f \right| +\left| \Ex\left[f\right] \right|\le 2 \sup(|f|) $.
Therefore, rearranging the previous inequality, we have
\begin{align}
    \mathrm{Prob}\left(\left| f  - \Ex\left[f\right] \right| > T  \right)\ge \frac{ \Var[f] - T^2 }{4\sup(|f|)^2-T^2}.
\end{align}
By choosing $T\coloneqq \frac{1}{\sqrt{2}}\sqrt{\Var[f]}$, we get
\begin{align}
    \mathrm{Prob}\left(\left| f  - \Ex\left[f\right] \right| > \sqrt{\frac{\Var[f]}{2}} \right)\ge \frac{  \Var[f] }{8\sup(|f|)^2-\Var[f]}\ge \frac{  \Var[f] }{8\sup(|f|)^2}.
\end{align}
\end{proof}

\begin{lemma}[Large first moments]
\label{le:probstat2}
Let $f$ be a real function depending by parameters distributed according to a probability distribution $\mu$. Then, the inequality 
\begin{align}
    \mathrm{Prob}\left( f   > \frac{\Ex[f]}{2} \right) \ge \frac{\Ex[f]}{2\sup(|f|)}
\end{align}
holds, where the expected value and variance are taken with respect to the probability distribution $\mu$.
\end{lemma}

\begin{proof}
Let $T$ a real value $T>0$. Then, we have
\begin{align}
    \Ex[f]&=\int\! f d\mu=\int_{\left| f \right|\le T}\! f d\mu + \int_{\left| f \right|> T}\!f d\mu\le T  + \sup(|f|) \,\mathrm{Prob}\left( f   > T \right).
\end{align}
Now, if we assume $T=\Ex[f]/2$ and rearrange the inequality, we obtain
\begin{align}
    \mathrm{Prob}\left( f   > \frac{\Ex[f]}{2} \right) \ge \frac{\Ex[f]}{2\sup(|f|)}.
\end{align}
\end{proof}

\subsection{Trace distance decay for worst-case circuits under local depolarizing noise}
\label{subsub:tracedep}
As documented by the previous literature\ \cite{DanielPaper, hirche2022contraction, nibp}, the output of any circuit affected by unital, primitive noise converges exponentially fast in the depth to the maximally mixed state with respect to the trace distance. The proof relies on the Pinsker's inequality and on the contraction coefficients of the quantum Rényi divergence of order $2$. See the definition of Rényi divergence in Eq.~(\ref{eq:relativeentropy}).
For the sake of simplicity, we will consider the special case of the depolarizing noise, and refer to Refs.\ \cite{hirche2022contraction, nibp} for an extension to arbitrary Pauli channels with normal form parameters satisfying $D_P < 1$ for all $P\in\{X,Y,Z\}$.
We will need the following lemma.
\begin{lemma}[Strong data-processing inequality. Adapted from Theorem 6.1 in
Ref.\ \cite{DanielPaper}]\label{lem:muller}
 Let ${\mathcal{N}_p^{(\mathrm{dep})}}$ be the single-qubit depolarizing channel of rate $p$, i.e.,  ${\mathcal{N}_p^{(\mathrm{dep})}}(X) =   p\frac{I}{2} \Tr(X) + (1-p)X$.
  Let $\Phi:= \bigcirc_{i=1}^L (\mathcal{N}_p^{(\mathrm{dep})\otimes n} \circ \mathcal{U}_i) $ be a circuit of $L$ unitary layers interspersed by local depolarizing noise. Then, for every state $\rho$, we have
  \begin{align}
D_2\left(\Phi(\rho) \bigg  \|  \frac{I}{2^n}\right)  \leq (1-p)^{2L} n.
\end{align}
\end{lemma}
\begin{comment}

\begin{proof}
    We first observe that the channel $\mathcal{N}$ can be rewritten as a composition of a depolarizing channel ${\mathcal{N}_p^{(\mathrm{dep})}}(X) = \frac{I}{2}\Tr{X} + (1-p) X $ and another unital channel $\mathcal{M}(I + \bold{w}\cdot \boldsymbol{\sigma}) = I + \frac{1}{(1-p)}({D\bold{w}})\cdot \boldsymbol{\sigma}$, i.e.
    \begin{align}
        &{\mathcal{N}_p^{(\mathrm{dep})}} \circ \mathcal{M} (I + \bold{w}\cdot \boldsymbol{\sigma}) = {\mathcal{N}_p^{(\mathrm{dep})}} \left( I + \frac{1}{(1-p)}({D\bold{w}})\cdot \boldsymbol{\sigma} \right) \\= &p I + (1-p)\left( I + \frac{1}{(1-p)}({D\bold{w}})\cdot \boldsymbol{\sigma} \right) = \mathcal{N}(I + \bold{w}\cdot \boldsymbol{\sigma}).
     \end{align}
Then we have
\begin{align}
    &D_2\left(\mathcal{N}^{\otimes n}(\rho) \bigg  \|  \frac{I}{2^n}\right) = D_2\left(\mathcal{N}_p^{(\mathrm{dep})\otimes n}  \circ \mathcal{M}^{\otimes n}(\rho) \bigg  \|  \frac{I}{2^n}\right) \\&\leq (1-p)^2 D_2\left(\mathcal{M}^{\otimes n}(\rho) \bigg  \|  \frac{I}{2^n}\right) \leq (1-p)^2 D_2\left(\rho \bigg  \|  \frac{I}{2^n}\right).
\end{align}
Iterating over multiple layers, we obtain
\begin{align}
D_2\left(\Phi(\rho) \bigg  \|  \frac{I}{2^n}\right)  \leq 
    D_2\left(\bigcirc_{i=1}^L (\mathcal{N}^{\otimes n}\circ \mathcal{U}_i) (\rho) \bigg  \|  \frac{I}{2^n}\right)  \leq  (1-p)^{2L} D_2\left(\rho \bigg  \|  \frac{I}{2^n}\right) \leq (1-p)^{2L} n.
\end{align}
\end{proof}
Then, the desired result for the trace distance easily follows.
\end{comment}
Then the desired result follows by a direct application of 
Pinsker's inequality.
\begin{proposition}[Deviation from maximal mixedness]
\label{prop:sdpi}
     Under the same assumptions of Lemma\ \ref{lem:muller}, we obtain
     \begin{align}
         \bigg  \| \Phi(\rho) - \frac{I}{2^n}\bigg  \|_1 \leq \sqrt{2n}(1-p)^{L} . 
     \end{align}
     And therefore, for all states $\rho$ and $\sigma$, we have
    \begin{align}
         \norm{\Phi(\rho) - \Phi(\sigma)}_1 \leq 2\sqrt{2n}(1-p)^{L}. 
     \end{align}
\end{proposition}
\begin{proof}
Combining Pinsker's inequality with the monotonicity of the the quantum Rényi divergence, we obtain
\begin{align}
     \bigg  \| \Phi(\rho) - \frac{I}{2^n}\bigg  \|_1^2  \leq  2D\bigg (\Phi(\rho) \bigg \| \frac{I}{2^n}\bigg ) \leq  2D_2\bigg (\Phi(\rho) \bigg \| \frac{I}{2^n}\bigg ).
\end{align}
Hence, Proposition\ \ref{prop:sdpi} implies 
\begin{align}
    \bigg  \| \Phi(\rho) - \frac{I}{2^n}\bigg  \|_1 \leq \sqrt{2n}(1-p)^{L}. 
\end{align}
Thus, a direct application of the triangle inequality yields the desired result
\begin{align}
    \norm{\Phi(\rho) - \Phi(\sigma)}_1 \leq  \bigg  \| \Phi(\rho) - \frac{I}{2^n}\bigg  \|_1  +  \bigg  \| \Phi(\sigma) - \frac{I}{2^n}\bigg  \|_1 \leq 2\sqrt{2n}(1-p)^{L}.
\end{align}
\end{proof}
We emphasize that the above result does not require randomness, unlike our Theorem\ \ref{th:effective}. On the other hand, Theorem\ \ref{th:effective} yields non trivial bound even at constant depth, while the above depolarizing-noise result is vacuous at sub-logarithmic depth.
We remark that a non-vacuous bound for shallow depths could be derived by means of the quantum Bretagnolle-Huber inequality (see, for instance, Ref.\ \cite{canonne2022short}, (Lemma B.1 in 
Ref.\ \cite{angrisani2023unifying}) and references therein). 

Informally, Proposition\ \ref{prop:sdpi} says that the output of a noisy circuit becomes computationally trivial at super-logarithmic depths, provided that the noise is unital and primitive.
This poses severe constraints on the capabilities of noisy devices, as exemplified by the following result.

\begin{corollary}[Exponential concentration, unital case]
    Under the same assumptions of Lemma\ \ref{lem:muller}, for every state $\rho$ and for every observable $O$, we obtain
     \begin{align}
        \left|\Tr[O\Phi(\rho)] - \frac{\Tr[O]}{2^n}\right| \leq \sqrt{2n}(1-p)^{L} \norm{O}_\infty.
     \end{align}
\end{corollary}
\begin{proof}
First, we notice that the LHS can be rearranged as 
    \begin{align}
     \left|\Tr[O\Phi(\rho)] - \frac{\Tr[O]}{2^n} \right| =  \left| \Tr\left[O\left(\Phi(\rho)- \frac{1}{2^n}\right)\right] \right|.
    \end{align}
Hence, we obtain 
    \begin{align}
        \left|\Tr\left[O\left(\Phi(\rho)- \frac{1}{2^n}\right)\right] \right| \leq \left\| \Phi(\rho) - \frac{I}{2^n}\right\|_1\norm{O}_\infty \leq \sqrt{2n}(1-p)^{L} \norm{O}_\infty,
    \end{align}
where the first inequality follows from the H{\"o}lder's inequality, and the second one is a consequence of Proposition\ \ref{prop:sdpi}.
\end{proof}

\subsection{Numerical simulations}
%\subsubsection*{Numerical results}
In this section we present numerical results that corroborate our analytical results and explore regimes that go beyond the assumptions and results of our theorems, such as the one of assuming the $2$-qubit gates to be distributed according to a local $2$-design. 
We start by considering a brickwork architecture as depicted in Fig.~\ref{Fig:circ}, where each 2-qubit gate takes the form
\begin{align}
U_{i,i+1}\left(\theta_1,\theta_2,\theta_3,\theta_4\right)\coloneqq \left(R_{Y}(\theta_4)\otimes R_{Y}(\theta_3)\right) \mathrm{CNOT}_{i,i+1} \left(R_{X}(\theta_2)\otimes R_{X}(\theta_1)\right).
\end{align}
We consider the noise model given by the composition of amplitude damping and depolarizing channels:
\begin{align}
\mathcal{N}_{(p,q)}^{(\mathrm{dep,amp})}\coloneqq \mathcal{N}_p^{(\mathrm{dep})}\circ \mathcal{N}_q^{(\mathrm{amp})},
\end{align} 
where two noise channels are defined respectively in Eq.~\eqref{eq:ampdampnoise} and Eq.~\eqref{eq:depnoise} with $p,q\in [0,1]$.

 Moreover, we assume, in contrast to our circuit model in Fig.~\ref{Fig:circ}, that the circuit ends with a layer of noise (instead of a layer of single-qubits gates).
Furthermore, we consider an expectation value with respect to the observable $Z_1$.
\begin{figure}[h]
\centering
\includegraphics[width=0.495\textwidth]{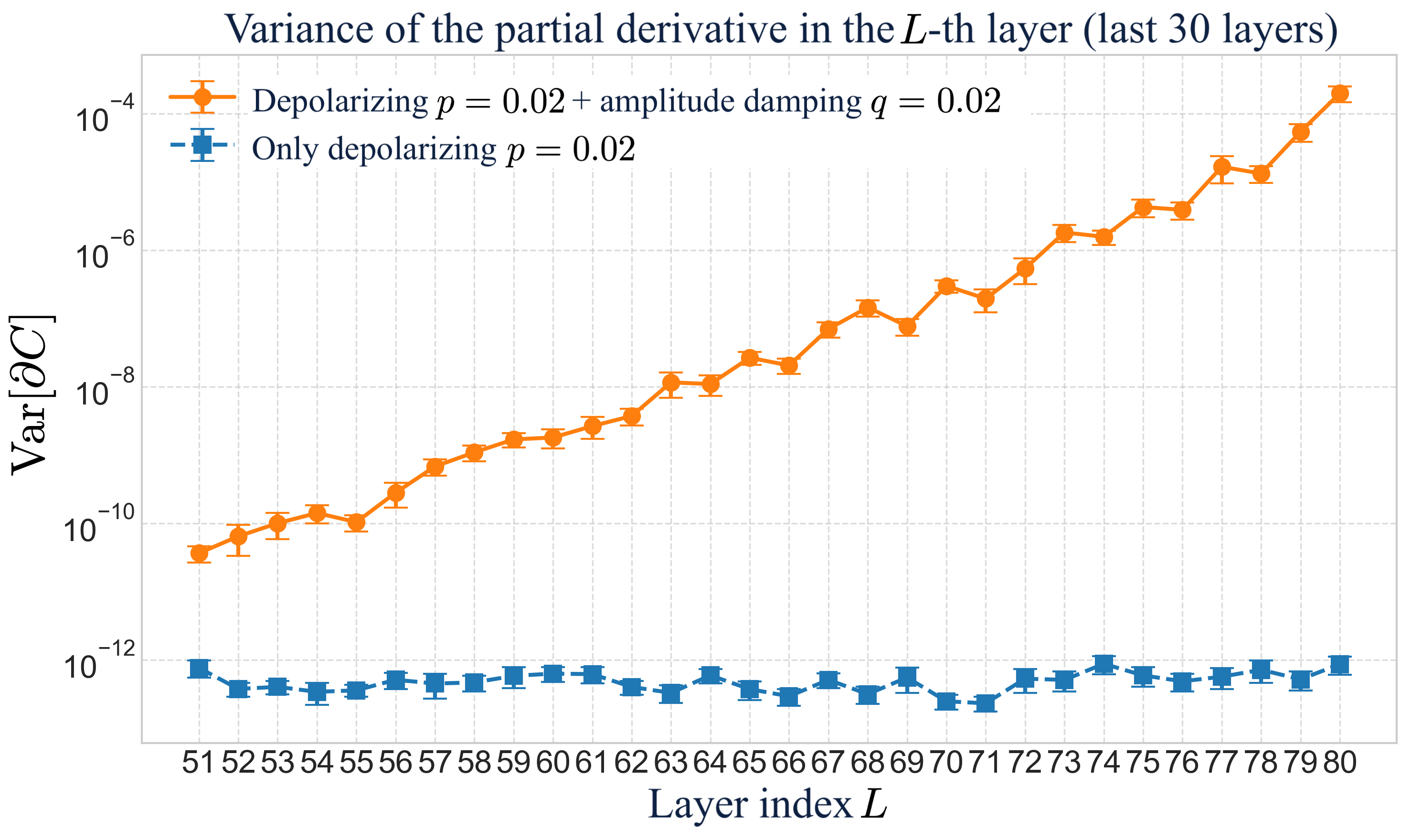}
\includegraphics[width=0.495\textwidth]{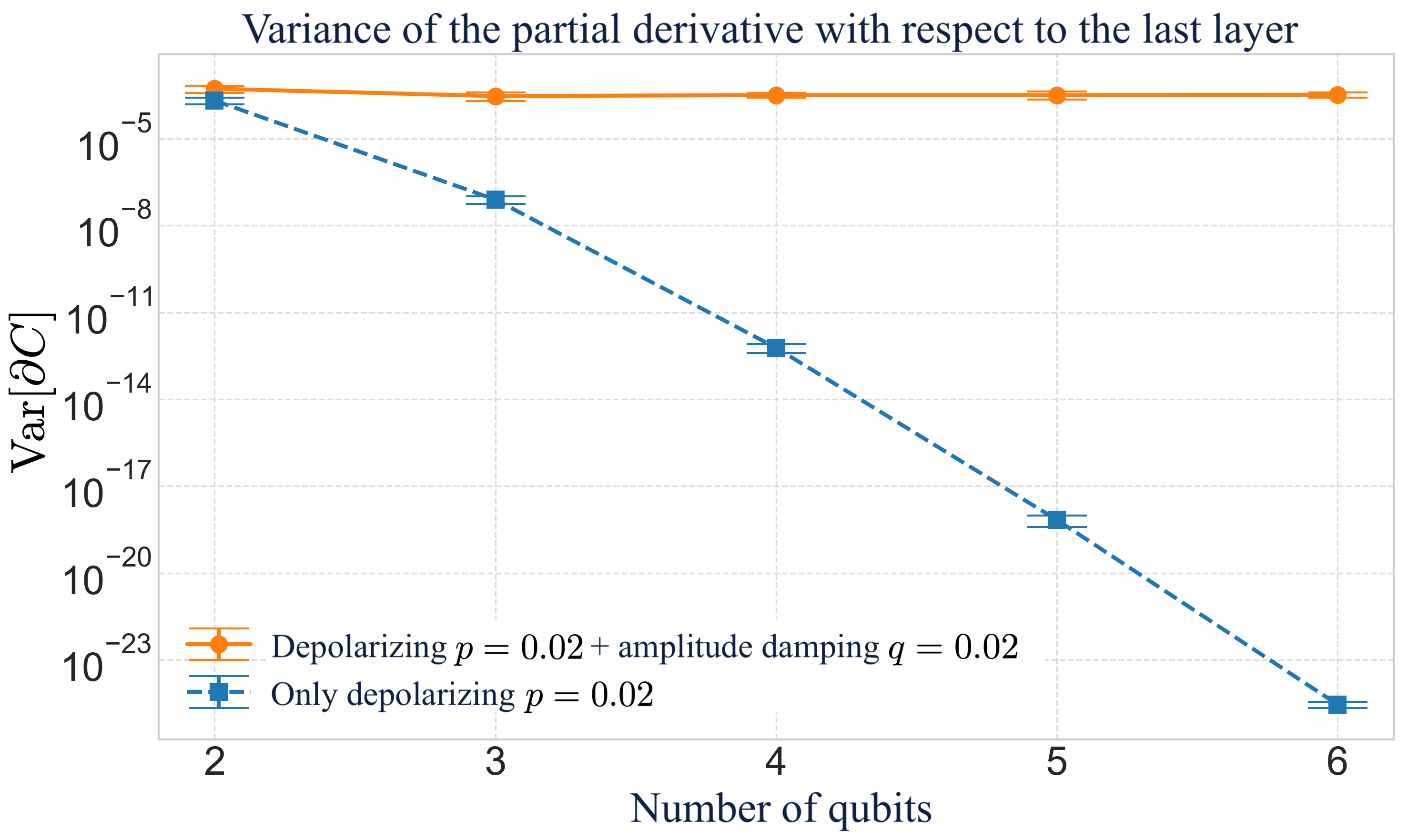}
\caption{
\textbf{Left:} Variance of the partial derivative with respect to a gate in the $L$-th layer of a $4$-qubit one-dimensional quantum circuit of depth $80$. The expectation value becomes exponentially less \emph{sensitive} to gates far from the end of the circuit. Noise parameters are set to $p=q=0.02$. When amplitude damping is switched off ($q=0$), all partial derivatives exhibit comparable scaling across layers. All circuit parameters are sampled uniformly in $[0,2\pi)$. 
\textbf{Right:} Variance of the partial derivative in the last layer versus the number of qubits, for circuit depth $20n$. With both amplitude damping and depolarizing noise ($p=q=0.02$), the last-layer variance remains constant with system size, whereas for $q=0$ it decays exponentially. In both plots, error bars indicate the standard deviation over $20$ random parameter samples, with all parameters drawn uniformly from $[0,2\pi)$.}
\label{Fig:AMPDEPlayer}
\end{figure}

In Fig.~\ref{Fig:AMPDEPlayer} (Left), we can clearly observe that the partial derivatives taken at the end of the circuit are significantly larger compared to those taken at the beginning of the circuit. This confirms the exponential decay that we proved in Theorem~\ref{th:ubvariance}.
It is noteworthy that if we were to deactivate the amplitude damping component, we would observe an average exponential concentration in all partial derivatives, regardless of the layer at which the derivative is taken. This aligns with the findings of the study on (depolarizing-)noise-induced barren plateaus~\cite{nibp} and with our Proposition~\ref{th:ubvarianceUNIT}.
Furthermore, in our experiments, we observe that the partial derivatives in the final layers remain constant as the number of qubits increases, as depicted in Fig.~\ref{Fig:AMPDEPlayer} (Right). This observation corroborates the conclusions drawn in our Corollary~\ref{th:bpINTappe}. In particular, this implies that the $2$-norm of the gradient remains constant on average with respect to the number of qubits. %From this plot, we can see that the same proven behaviour emerge even when concluding with a layer of noise and not with a layer of single-qubit random gates.
Another consideration is whether our theorems heavily depend on the assumption that the $2$-qubit gate is sampled by a $2$-design, and whether this assumption can be relaxed. Our evidence suggests that we might relax such assumption: in our numerical simulations, we observe a similar trend of what we proved even with more structured ansatz like QAOA, 
the \emph{quantum approximate optimization algorithm}~\cite{farhi2014quantum}, as demonstrated in Fig.~\ref{Fig:QAOA}. This provides further support for the notion that the assumption regarding the $2$-qubit gates being sampled by a $2$-design is not crucial.
\begin{figure}[h]
\centering
\includegraphics[width=0.48\textwidth]{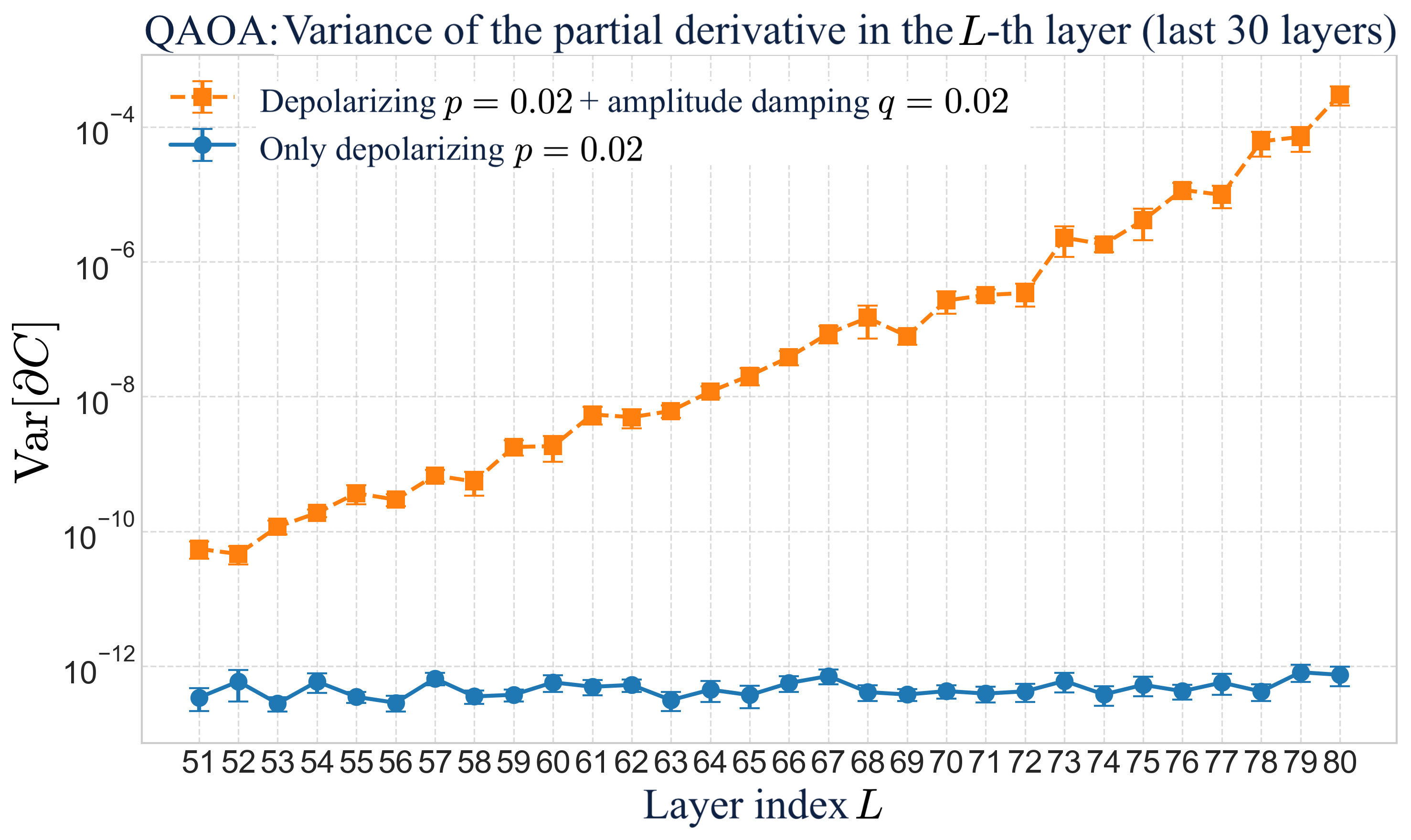}
\includegraphics[width=0.48\textwidth]{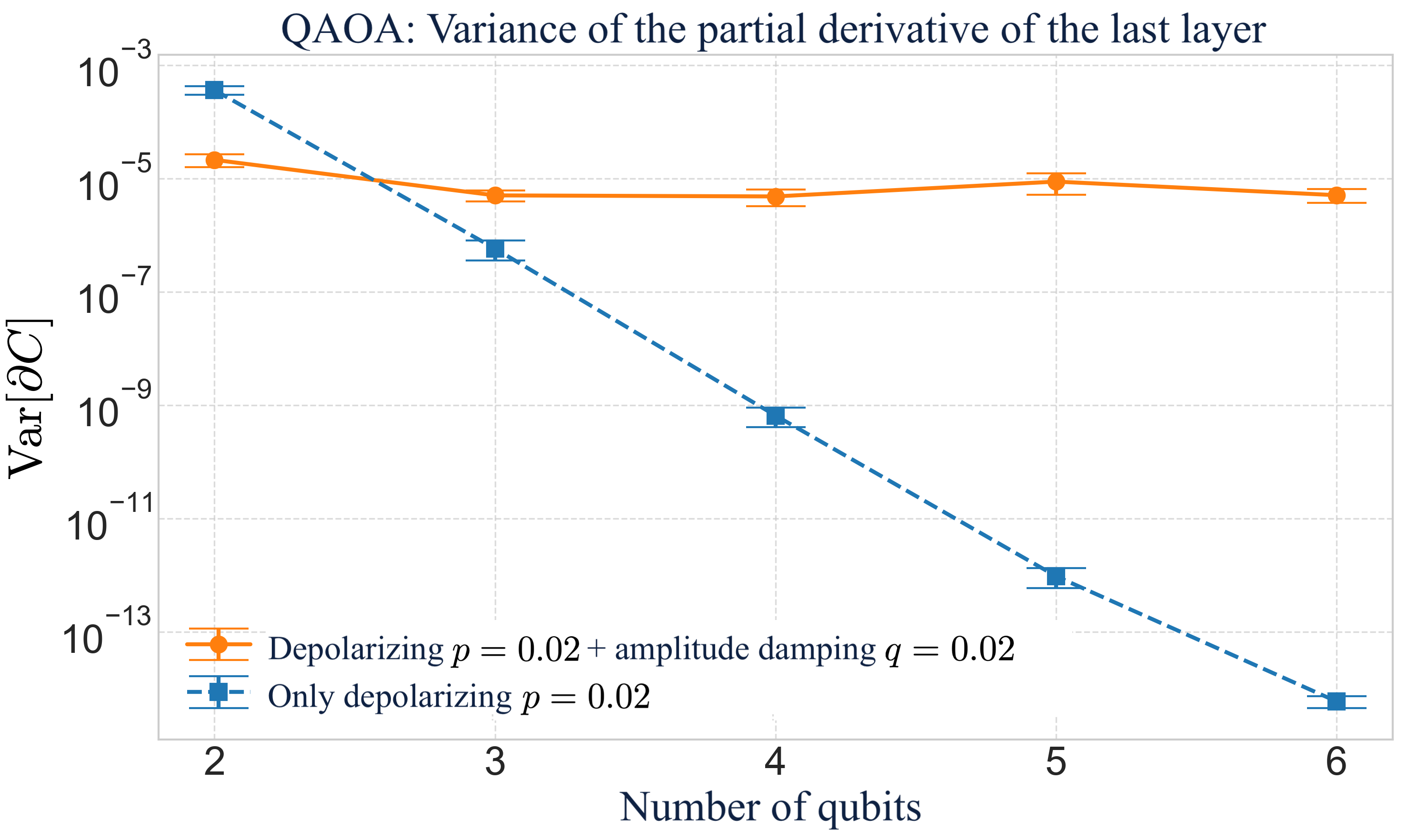}
\caption{
The same quantities as in Fig.~\ref{Fig:AMPDEPlayer}, but using a QAOA circuit ansatz. The results show no qualitative change, confirming that our conclusions are not limited to unstructured random ansätze. The noise parameters are $p=q=0.02$, and all circuit parameters $\{\gamma_i,\beta_i\}$ are sampled uniformly in $[0,2\pi)$. As an observable, we use $X_1$ instead of $Z_1$ (the latter has vanishing expectation value in the noiseless case for any fixed QAOA circuit due to symmetry). The QAOA ansatz takes the form $\prod_{i=1}^{D} \exp(-i\beta_i H_x) \exp(-i\gamma_i H_z) \ket{+}^{\otimes n}$, with $H_z=\sum_{i=1}^{n-1}Z_i Z_{i+1}$ and $H_x=\sum_{i=1}^n X_i$. Error bars represent the standard deviation over $20$ random parameter samples.}
\label{Fig:QAOA}
\end{figure}

\end{document}